\documentclass[12pt]{article}

\usepackage[english]{babel}
\usepackage[utf8]{inputenc}
\usepackage{amsmath}
\usepackage{graphicx}
\usepackage{amssymb}
\usepackage{amsthm}
\usepackage{bm}
\usepackage{tikz-cd}
\usepackage{mathrsfs}
\usepackage[colorinlistoftodos]{todonotes}
\usepackage{enumitem}
\usepackage{yfonts}
\usepackage{ dsfont }
\usepackage{geometry}
\usepackage{setspace}
\usepackage{tikz}
\usetikzlibrary{positioning}
\usetikzlibrary{decorations.pathreplacing}
\usetikzlibrary{patterns}

\numberwithin{equation}{section}
\title{Mathematical theory for the interface mode in a waveguide bifurcated from a Dirac point
\thanks{J. Lin was partially supported by the NSF grant DMS-2011148,
and H. Zhang was partially supported by Hong Kong RGC grant GRF 16304621.}}

\author{Jiayu Qiu\thanks{Department of Mathematics, 
 HKUST,  Clear Water Bay, Kowloon, Hong Kong SAR, China.
    \tt jqiuaj@connect.ust.hk.} \,\,\, Junshan Lin
  \thanks{Department of Mathematics and Statistics, Auburn University, Auburn, AL 36849.  \tt jzl0097@auburn.edu.}\,\,\,
  Peng Xie\thanks{Department of Mathematics, 
 HKUST,  Clear Water Bay, Kowloon, Hong Kong SAR, China.
    \tt mapengxie@ust.hk.} \,\,\,
 Hai Zhang
  \thanks{Department of Mathematics, 
 HKUST,  Clear Water Bay, Kowloon, Hong Kong SAR, China.
    \tt haizhang@ust.hk.}}

\date{\today}

\newtheorem{theorem}{Theorem}[section]
\newtheorem{lemma}[theorem]{Lemma}

\newtheorem{corollary}[theorem]{Corollary}

\newtheorem{remark}[theorem]{Remark}
\newtheorem{proposition}[theorem]{Proposition}
\newtheorem{assumption}[theorem]{Assumption}
\allowdisplaybreaks[4]

\begin{document}

\maketitle
\begin{abstract}
In this paper, we prove the existence of a bound state in a waveguide that consists of two semi-infinite periodic structures separated by an interface. The two periodic structures are perturbed from the same periodic medium with a Dirac point and they possess a common band gap enclosing the Dirac point. The bound state, which  is called interface mode here, decays exponentially away from the interface with a frequency located in the common band gap and can be viewed as a bifurcation from the Dirac point. Using the layer potential technique and asymptotic analysis, we first characterize the band gap opening for the two perturbed periodic media and derive the asymptotics of the Bloch modes near the band gap edges. By formulating the eigenvalue problem for the waveguide with two semi-infinite structures using a boundary integral equation over the interface and analyzing the characteristic values of the associated boundary integral operator, we prove the existence of the interface mode for the waveguide when the perturbation of the periodic medium is small.
\end{abstract}

%\tableofcontents
\section{Introduction}
\subsection{Background}
The localization of acoustic and electromagnetic waves allows for the confinement of waves in a small volume or guiding the waves along the desired direction, and therefore has significant applications in the design of novel acoustic and optical devices \cite{ozawa2019, qiu2011,sheng1990,soukoulis2012}. Several strategies have been proposed to 
achieve wave localization in different settings. For instance, a periodic medium with a single defect by changing the medium in a finite region can induce the so-called defect modes that are localized near the defect region. Such perturbations do not change the essential spectrum of the underlying differential
operators but create isolated eigenvalues of finite multiplicity in a band gap of the periodic medium, with the eigenmodes decaying exponentially \cite{Ammari2018-SIAMAP, Ammari2019linedefect, figotin1997_jsp,figotin1998}. Typically, the perturbations in such settings need to be large to create the defect modes \cite{figotin1998-1, Kuchment2009, Vu2014}.
Another way to create wave localization is by perturbing the periodic medium randomly in the whole domain. It can be shown that localized wave modes arise in the band gaps of the periodic medium, which is known as the Anderson localization \cite{figotin1994,figotin1996,figotin1997,sheng1990}.

In this paper, we explore the idea of creating wave localization in a waveguide by perturbing a periodic structure differently along the positive and negative parts of the waveguide axis.
In particular, we prove the existence of a localized wave mode (called bound state or interface mode) for such a configuration near the Dirac point of the periodic medium, which is a special vertex in the spectral band structure when two dispersion curves (surfaces) intersect in a linear (conic) manner.  The investigation of Dirac points has attracted intense research efforts in recent years due to their important role in topological insulators. For example, Dirac points
have been shown to occur in photonic graphene \cite{ablowitz2012nonlinear, fefferman2012honeycomb} and photonic/phononic models with honeycomb lattices \cite{ammari2020honeycomb, Cassier2021, LLZ2022}.
In general, topological phase transition takes place at a Dirac point. As such, near a Dirac point, eigenmodes localized around an interface can be generated by applying proper perturbations to the periodic operator on both sides of the interface. The bifurcation of eigenvalues from Dirac points was rigorously analyzed for one-dimensional Schr\"{o}dinger operators
\cite{fefferman2017topologically}, two-dimensional Schr\"{o}dinger operators 
\cite{fefferman2016edge, drouot2020edge, Drouot2019TheBC}, and two-dimensional 
elliptic operators with smooth coefficients \cite{lee2019elliptic}, all of which use domain wall models such that two periodic materials are ``connected'' adiabatically over a length scale that is much larger than the period of the structure. Note that such an adiabatic domain wall model may not be realistic for photonic/phononic materials with a sharp interface. In this work, we investigate the model when two different periodic media are glued directly such that the medium coefficient attains a jump across the interface. Our goal is to prove the existence of an interface mode near a Dirac point in the context of a waveguide, for which the interface between two periodic media is located at the origin of the waveguide axis. 
Such an interface mode is bifurcated from the Dirac point, and the corresponding eigenfrequency is located in the common band gap of the two periodic media enclosing the Dirac point. Furthermore, we prove that the eigenspaces at the band edges are swapped for the two periodic media, which demonstrates the topological phase transition of the medium at the Dirac point.

We point out that a bound state/interface mode arising from the bifurcation of a Dirac point in a one-dimensional system has been investigated under several different configurations. The existence and stability of interface mode eigenvalues were established by the transfer matrix method and the oscillatory theory for Sturm-Liouville operators in \cite{LZ2021}. In another work, \cite{TZ2022}, the bulk-interface correspondence for 1D topological materials with inversion symmetry was established. Defect modes for 1D dislocated periodic media were studied in \cite{drouot2020} and it was shown that the
defect modes arise as bifurcations from the Dirac operator eigenmodes. Finally, the existence of a stable interface mode in a finite chain of high contrast bubbles that consists of two chains of different topological properties was reported numerically \cite{ammari2020topologically}. However, the analytical techniques in \cite{drouot2020, LZ2021, TZ2022} only apply particularly to the 1D problems.

In this paper, we develop a new approach based on layer potential techniques and asymptotic analysis to investigate the existence of in-gap bound states for the waveguide. The approach overcomes the difficulty of discontinuous coefficients for the inhomogeneous medium in the waveguide and addresses the challenges brought by the presence of the sharp interface separating two different periodic media. The mathematical framework can be applied to study interface or edge modes in other photonic/phononic systems, which will be reported elsewhere. We refer \cite{ammari2018mathematical} for a systematic review of the application of layer potential technique to various wave propagation problems in photonic/phononic materials consisting of subwavelength resonators.

\subsection{Statement of main results} \label{sec-1-2}
We start with the periodic waveguide in Figure 1 that attains a Dirac point in its band structure. Let $\Gamma_{-}=\mathbf{R}\times\{0\}$ and $\Gamma_{+}=\mathbf{R}\times\{\frac{1}{2}\}$ be two parallel walls of the waveguide in the plane $\mathbf{R}^2$, and let $\Gamma=\{0\}\times (0,\frac{1}{2})$. Denote $\mathbb{Z}^*:=\mathbb{Z}\backslash\{0\}$. Here $0$ is excluded from the index set in order to simplify the notation in the sequel.
An array of identical obstacles $\{D_n\}_{n\in\mathbb{Z}^*}$ are arranged periodically in the center of the waveguide along the $x_1$-direction with period $1/2$. Here 
$D_n = z_n + D$ with $z_n=(\frac{2|n|-1}{4}\text{sgn}(n),\frac{1}{4})$ for $n\in \mathbb{Z}^*$ and $D$ is an open set centered at the origin. The domain outside the obstacles is denoted by $\Omega:=\mathbf{R}\times (0,\frac{1}{2})\backslash \cup_{n\in \mathbb{Z}^*}\overline{D_n}$. We consider a time-harmonic scalar wave that propagates in the waveguide at frequency $\omega$. The wave field $v$ satisfies the following Helmholtz equation 
\begin{equation}\label{eq:Helmholtz}
    \Delta v(x) + \omega^2 v(x)=0 \quad \mbox{for} \; x \in \Omega.
\end{equation}
We impose the following Neumann boundary condition on the waveguide walls
\begin{equation}\label{eq:bnd_waveguide_wall}
\frac{\partial v(x)}{\partial x_2}=0 \quad \mbox{for} \; x\in \Gamma_{-}\bigcup\Gamma_{+},  
\end{equation}
and the Dirichlet boundary condition on the boundaries of the obstacles 
\begin{equation}
    v(x)=0 \quad \mbox{for} \; x \in \cup_{n\in \mathbb{Z}^*} \partial D_n.
\end{equation}
We point out that the above assumptions on the boundary conditions are not essential. The method developed in the paper applies to other boundary conditions; for instance, the Neumann boundary condition can be imposed over the obstacle boundaries.

% the following equations: 
% \begin{equation*}
%   \left\{
%    \begin{aligned}
%        &(\Delta_x+ \lambda) v(x)=0,\quad x \in \Omega,\\
%        &v(x)=0,\quad x \in \cup_{n\in \mathbb{Z}^*} \partial D_n ,\\
%        &\frac{\partial}{\partial x_2}v(x)=0,\quad x\in \Gamma_{-}\bigcup\Gamma_{+} 
%          \end{aligned}
%    \right.
% \end{equation*}
% Note that the choice of the boundary conditions in the equation is not an essential issue. For example, the method shown in this paper will work if we replace the Dirichlet condition on the particles with the Neumann condition.
\begin{figure}
    \centering
    \begin{tikzpicture}

\draw (-0.2,1)--(8.2,1);
\draw (-0.2,-1)--(8.2,-1);
\node at (4,1.3) {$\Gamma_{+}$};
\node at (4,-1.3) {$\Gamma_{-}$};

\node[right] at (8.4,1) {$\cdots$};
\node[right] at (8.4,-1) {$\cdots$};
\node[left] at (-0.4,1) {$\cdots$};
\node[left] at (-0.4,-1) {$\cdots$};
\draw[dashed] (0,-1)--(0,1);
\draw[dashed] (2,-1)--(2,1);
\draw[dashed] (4,-1)--(4,1);
\draw[dashed] (6,-1)--(6,1);
\draw[dashed] (8,-1)--(8,1);

\draw (1,0) ellipse(0.5 and 0.3);
\draw (3,0) ellipse(0.5 and 0.3);
\draw (5,0) ellipse(0.5 and 0.3);
\draw (7,0) ellipse(0.5 and 0.3);

\node[font=\fontsize{10}{10}\selectfont] at (1,0.5) {$D_{1}$};
\node[font=\fontsize{10}{10}\selectfont] at (3,0.5) {$D_{2}$};
\node[font=\fontsize{10}{10}\selectfont] at (5,0.5) {$D_{3}$};
\node[font=\fontsize{10}{10}\selectfont] at (7,0.5) {$D_{4}$};

\draw[dashed] (1,0)--(1,-1);
\draw[dashed] (3,0)--(3,-1);
\draw[decorate,decoration={brace,mirror}] (1,-1.1) -- (3,-1.1);
\node[below,font=\fontsize{10}{10}\selectfont] at (2,-1.1) {$\frac{1}{2}$};
\end{tikzpicture}
    \caption{The waveguide with periodically arranged obstacles.}
    \label{unperturbed waveguide}
\end{figure}
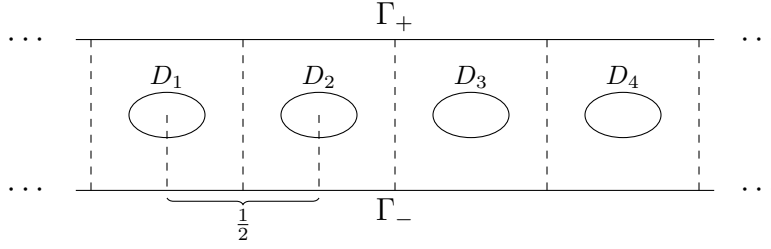

We shall make the following assumptions on the spectrum of the operator $-\Delta$ for the above period-$1/2$ structure $\Omega$. The reader is referred to Section 2.1 for terminologies arising from the Floquet-Bloch theory.
 
\begin{assumption} \label{assump0}
\begin{enumerate}
\item [(1)]
The first spectral band $(p,\mu_1(p))$ of the period-$1/2$ structure $\Omega$ is smooth for $p\in [-2\pi, 2\pi]$, and can be extended analytically in a complex neighborhood of $\mathbf{R}$. Moreover, the eigenspace corresponding to the first spectral band is one dimensional and is given by $\text{span}\{v_1(\cdot;p)\}$. The analytic continuation of $v_1(\cdot;p)$ from $\mathbf{R}$ to its complex neighborhood also holds.

\item [(2)]
The slope of $\mu_1(\cdot)$ at $p=\pi$ satisfies 
$\mu_1'(\pi)=\alpha_* >0$.    

\item [(3)]
$\lambda_*:= \mu_1(\pi)$ is not in the point spectrum of the period-$1/2$ structure.

\item [(4)]
$\mu_1(p)\neq\lambda_*$ for $p\neq\pi$, and consequently $\lambda_* =\max_{0\leq p \leq \pi}\mu_1(p)=\min_{\pi\leq p \leq 2\pi} \mu_1(p)$. In addition, the maximum of the first spectral band is smaller than or equal to the minimum of the second spectral band. 

\end{enumerate}
\end{assumption}
    
\begin{remark}
All the statements in Assumption \ref{assump0} can be proved rigorously in the case when the size of the particle $D$ is sufficiently small by using asymptotic analysis.
\end{remark}
\begin{remark}
It is generally believed that the point spectrum is empty for a periodic structure shown in Figure 1. However, as pointed out in \cite{kuchment2016overview}, its rigorous proof is an open question.
\end{remark}
\begin{remark}
The statement (4) in Assumption \ref{assump0} is similar to the spectral no-fold assumption in \cite{fefferman2016edge} (cf. Definition 7.1). We conjecture that it is not essential for the main result Theorem \ref{main result} to hold. This issue shall be investigated in a separate paper.
\end{remark}
In the sequel, we shall treat the above period-1/2 structure $\Omega$ as a period-1 structure whose 
primitive cell is given by $Y=(0,1)\times(0,\frac{1}{2})$ and the associated Brillouin zone is $[0, 2\pi]$. For the period-1 structure, each period consists of two obstacles.
We denote by $(\lambda_n(p), u_n(\cdot; p))$ ($n\geq 1$) the eigenpairs for the $n$-th spectral band. 
Note that the spectral band structure of the period-1 structure can be obtained by `folding' the bands of the period-1/2 structure. This allows for the creation of a Dirac point $(p_*=\pi, \lambda_*)$ in the spectrum of the period-1 structure in the Brillouin zone as stated in the following proposition.

\begin{proposition}
\label{Existence of Dirac points of the period-1 structure}
A linear band crossing occurs at $\left(p_{*},\lambda_{*}\right)$ between the first and second dispersion curves of the period-1 structure $\Omega$ with the following properties:
\begin{enumerate}
    \item [(1)] $\lambda_*=\lambda_1 (p_{*})=\lambda_{2} (p_{*})$, and $\lambda_*$ is an eigenvalue of multiplicity 2.
    
    \item [(2)]
Denote by $\mu_2(p):= \mu_1(2\pi -p)$ and $v_2(\cdot; p):= \overline{v_1(\cdot; 2\pi -p)}$, then the first and second spectral bands and the corresponding eigenfunctions can be chosen as below: 
\begin{equation*}
    \lambda_1(p)=
    \left\{
    \begin{aligned}
    &\mu_1(p),\quad p\in[0,\pi), \\
    &\mu_2(p),\quad p\in[\pi,2\pi],
    \end{aligned}
    \right.
    \quad
    u_1(\cdot,p)=
    \left\{
    \begin{aligned}
    &v_1(\cdot;p),\quad p\in[0,\pi), \\
    &v_2(\cdot;p),\quad p\in[\pi,2\pi],
    \end{aligned}
    \right.
\end{equation*}
and
\begin{equation*}
    \lambda_2(p)=
    \left\{
    \begin{aligned}
    &\mu_2(p),\quad p\in[0,\pi), \\
    &\mu_1(p),\quad p\in[\pi,2\pi],
    \end{aligned}
    \right.
    \quad
    u_2(\cdot,p)=
    \left\{
    \begin{aligned}
    &v_2(\cdot;p),\quad p\in[0,\pi), \\
    &v_1(\cdot;p),\quad p\in[\pi,2\pi].
    \end{aligned}
    \right.
\end{equation*}
\end{enumerate}
\end{proposition}

\begin{remark} \label{remark on the q-p boundary condition}
The function $v_1(\cdot;p)$ above stands for the Bloch eigenmode when the primitive cell is $(0,\frac{1}{2})\times (0,\frac{1}{2})$. Specifically, $v_1(\cdot;p)$ satisfies the quasi-periodic boundary condition $v_1(x_1+\frac{1}{2}, x_2;p)=e^{i\frac{p}{2}}v_1(x_1,x_2;p)$.
\end{remark}
%\begin{remark}
%The last property listed in the above Assumption means the dispersion locus near a Dirac point is the union of two smooth and transversely intersecting curves. 
%\end{remark}

We assume that the periodic structure $\Omega$ attains reflection symmetry, which is a consequence of the following assumption. 
\begin{assumption} \label{hypo_reflection symmetry}
The obstacle $D$ attains the reflection symmetry such that its boundary $\partial D$ is parameterized by $\left\{(\theta,r(\theta)):0\leq\theta\leq 2\pi \right\}$ in the polar coordinate, in which $r(\theta)=r(\pi-\theta)$ for $0\leq \theta \leq \pi$.
\end{assumption}

We have the following corollary for the eigenmodes at the Dirac point $(p_*,\lambda_*)$.
\begin{corollary} \label{even odd mode and root function}
Under Assumptions \ref{assump0} and \ref{hypo_reflection symmetry}, the eigenmodes at the Dirac point $(p_*,\lambda_*)$ can be chosen to be odd and even functions with respect to $x_1$. More precisely, $\text{span}\{u_1(\cdot;\pi),u_2(\cdot;\pi)\}=\text{span}\{\phi_1,\phi_2\}$ with
\begin{equation*}
    \phi_1(-x_1,x_2)=-\phi_1(x_1,x_2), \quad
    \phi_2(-x_1,x_2)=\phi_2(x_1,x_2).
\end{equation*}
Let
\begin{eqnarray} \label{eq24}
\bm{\varphi}_1(x)&=\left(\partial_n\phi_1|_{\partial D_1}(x+z_1),\partial_n\phi_1|_{\partial D_2}(x+z_2)\right)^T\in (H^{-\frac{1}{2}}(\partial D))^2, \label{eq24-1}\\
\bm{\varphi}_2(x)&=\left(\partial_n\phi_2|_{\partial D_1}(x+z_1),\partial_n\phi_2|_{\partial D_2}(x+z_2)\right)^T\in (H^{-\frac{1}{2}}(\partial D))^2, \label{eq24-2}
\end{eqnarray}
then 
$$
\bm{\varphi}_1=(\varphi_{ref},\varphi),\quad
\bm{\varphi}_2=(\varphi,-\varphi_{ref}),
$$
for some real-valued function $\varphi\in H^{-\frac{1}{2}}(\partial D)$ and 
$
\varphi_{ref}(x_1, x_2):=\varphi(-x_1, x_2).
$
%Moreover, $\varphi$ is a real-valued function, i.e. $\phi:(r(\theta),\theta)\to\mathbf{R}$.
\end{corollary}
\begin{proof}
    See Appendix A.
\end{proof}

For ease of notation, we denote $v_i=v_i(\cdot,\pi)$ for $i=1,2$. 

\begin{remark} \label{link between two types of modes}
The eigenmodes $\phi_1, \phi_2$ introduced in the above lemma are linear combinations of $v_1$ and $v_2$. We note that $\phi_1|_{\Gamma}$ vanishes due to the odd parity. Thus both $v_1|_{\Gamma}$ and $v_2|_{\Gamma}$ are proportional to $\phi_2|_{\Gamma}$. This relationship will be used later.  
\end{remark}
% Define the Sobolev space
% \begin{equation*} \label{eq5}
%    H_{b}^2(\Omega):=\left\{u\in H^2(\Omega): \frac{\partial u}{\partial x_2}|_{\Gamma_{-}\bigcup\Gamma_{+}}=0,u|_{\partial D_{n}}=0\enspace, \forall n\in\mathbb{Z}^*\right\},
% \end{equation*} 
% and the parity operator $\mathcal{P}$ by
% \begin{equation*} \label{eq6}
%    \mathcal{P}:L^2(\Omega)\to L^2(\Omega),\quad(\mathcal{P}\psi)(x_1,x_2)=\psi(-x_1,x_2).
% \end{equation*}
% Following Assumption \ref{hypo_reflection symmetry}, it is clear that $H_{b}^2(\Omega)\subset L_2(\Omega)$ is invariant under the operation $\mathcal{P}$, i.e. $\mathcal{P}(H_{b}^2(\Omega))\subset H_{b}^2(\Omega)$.

%perturbed structure
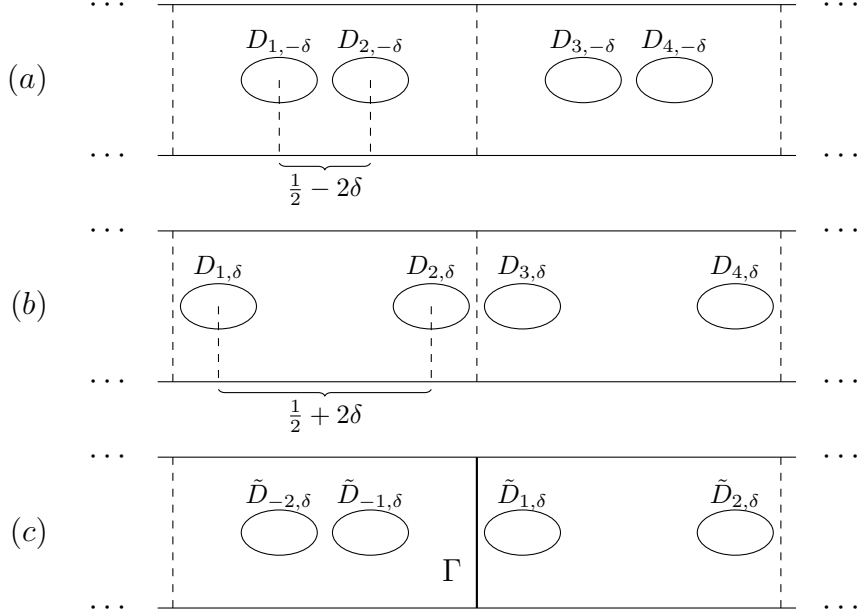
\begin{figure}
    \centering
    \begin{tikzpicture}
%negative perturb
\node[left] at (-1.5,0) {$(a)$};
\draw (-0.2,1)--(8.2,1);
\draw (-0.2,-1)--(8.2,-1);
\node[right] at (8.4,1) {$\cdots$};
\node[right] at (8.4,-1) {$\cdots$};
\node[left] at (-0.4,1) {$\cdots$};
\node[left] at (-0.4,-1) {$\cdots$};
\draw[dashed] (0,-1)--(0,1);
\draw[dashed] (4,-1)--(4,1);
\draw[dashed] (8,-1)--(8,1);

\draw (1.4,0) ellipse(0.5 and 0.3);
\draw (2.6,0) ellipse(0.5 and 0.3);
\draw (5.4,0) ellipse(0.5 and 0.3);
\draw (6.6,0) ellipse(0.5 and 0.3);

\node[font=\fontsize{10}{10}\selectfont] at (1.4,0.5) {$D_{1,-\delta}$};
\node[font=\fontsize{10}{10}\selectfont] at (2.6,0.5) {$D_{2,-\delta}$};
\node[font=\fontsize{10}{10}\selectfont] at (5.4,0.5) {$D_{3,-\delta}$};
\node[font=\fontsize{10}{10}\selectfont] at (6.6,0.5) {$D_{4,-\delta}$};

\draw[dashed] (1.4,0)--(1.4,-1);
\draw[dashed] (2.6,0)--(2.6,-1);
\draw[decorate,decoration={brace,mirror}] (1.4,-1.1) -- (2.6,-1.1);
\node[below,font=\fontsize{10}{10}\selectfont] at (2,-1.1) {$\frac{1}{2}-2\delta$};

%positive perturb
\node[left] at (-1.5,-3) {$(b)$};
\draw (-0.2,-2)--(8.2,-2);
\draw (-0.2,-4)--(8.2,-4);
\node[right] at (8.4,-2) {$\cdots$};
\node[right] at (8.4,-4) {$\cdots$};
\node[left] at (-0.4,-2) {$\cdots$};
\node[left] at (-0.4,-4) {$\cdots$};
\draw[dashed] (0,-4)--(0,-2);
\draw[dashed] (4,-4)--(4,-2);
\draw[dashed] (8,-4)--(8,-2);

\draw (0.6,-3) ellipse(0.5 and 0.3);
\draw (3.4,-3) ellipse(0.5 and 0.3);
\draw (4.6,-3) ellipse(0.5 and 0.3);
\draw (7.4,-3) ellipse(0.5 and 0.3);

\node[font=\fontsize{10}{10}\selectfont] at (0.6,-2.5) {$D_{1,\delta}$};
\node[font=\fontsize{10}{10}\selectfont] at (3.4,-2.5) {$D_{2,\delta}$};
\node[font=\fontsize{10}{10}\selectfont] at (4.6,-2.5) {$D_{3,\delta}$};
\node[font=\fontsize{10}{10}\selectfont] at (7.4,-2.5) {$D_{4,\delta}$};

\draw[dashed] (0.6,-3)--(0.6,-4);
\draw[dashed] (3.4,-3)--(3.4,-4);
\draw[decorate,decoration={brace,mirror}] (0.6,-4.1) -- (3.4,-4.1);
\node[below,font=\fontsize{10}{10}\selectfont] at (2,-4.1) {$\frac{1}{2}+2\delta$};

%mixed perturb
\node[left] at (-1.5,-6) {$(c)$};
\draw (-0.2,-5)--(8.2,-5);
\draw (-0.2,-7)--(8.2,-7);
\node[right] at (8.4,-5) {$\cdots$};
\node[right] at (8.4,-7) {$\cdots$};
\node[left] at (-0.4,-5) {$\cdots$};
\node[left] at (-0.4,-7) {$\cdots$};
\draw[dashed] (0,-7)--(0,-5);
\draw[thick] (4,-7)--(4,-5);
\node[right] at (3.4,-6.5) {$\Gamma$};
\draw[dashed] (8,-7)--(8,-5);

\draw (1.4,-6) ellipse(0.5 and 0.3);
\draw (2.6,-6) ellipse(0.5 and 0.3);
\draw (4.6,-6) ellipse(0.5 and 0.3);
\draw (7.4,-6) ellipse(0.5 and 0.3);

\node[font=\fontsize{10}{10}\selectfont] at (1.4,-5.5) {$\tilde{D}_{-2,\delta}$};
\node[font=\fontsize{10}{10}\selectfont] at (2.6,-5.5) {$\tilde{D}_{-1,\delta}$};
\node[font=\fontsize{10}{10}\selectfont] at (4.6,-5.5) {$\tilde{D}_{1,\delta}$};
\node[font=\fontsize{10}{10}\selectfont] at (7.4,-5.5) {$\tilde{D}_{2,\delta}$};

\end{tikzpicture}
\caption{Waveguide with perturbations: (a) The distance between each pair of two obstacles in one periodic cell is $\frac{1}{2}-2\delta$; (b) The distance between each pair of two obstacles in one periodic cell is $\frac{1}{2}+2\delta$; (c) The waveguide obtained by gluing the structures in (a) and (b) along the interface $\Gamma$.}
    \label{perturbed waveguide}
\end{figure}

For technical reasons, we make the following assumption which guarantees the analyticity of the Green's function for the empty waveguide, and will be used in the proof of Proposition 3.5. 

\begin{assumption} \label{hypo_singular frequency}
The Dirac energy level $\lambda_*$ is away from the singular frequencies such that $\lambda_*\neq (2m+1)\pi$ for $m\in\mathbb{N}$.
\end{assumption}

We now introduce perturbations to the periodic structure along the negative and positive parts of the waveguide axis. Let $0<\delta\ll 1$. The obstacle $D_n$ is translated to $\tilde{D}_{n,\delta}$ ($n\in\mathbb{Z}^*$) with the mass center given by
$$
\Tilde{z}_{n,\delta}=\left\{
\begin{aligned}
    &z_n+\delta\bm{e}_1,\quad \text{n is even},\\
    &z_n-\delta\bm{e}_1,\quad \text{n is odd}.
\end{aligned}
\right.
$$
The corresponding structure is depicted in Figure 2(c), which can be viewed as a joint structure that glues the two semi-infinite periodic structures in Figure 2(a) and 2(b) along the interface $\Gamma$. Define $\tilde{\Omega}_{\delta}:=\mathbf{R}\times(0,\frac{1}{2})\backslash \bigcup_n \overline{\tilde{D}_{n,\delta}}$.
An interface mode $u$ for the waveguide in Figure 2(c) is a finite-energy solution to the following spectral problem:
\begin{equation} \label{eq7}
    \left\{
    \begin{aligned}
        &(\Delta_x+\lambda)u(x;\lambda)=0,\quad x \in \tilde{\Omega}_{\delta},\\
        &u(x;\lambda)=0,\quad x \in\cup_{n\in \mathbb{Z}^*} \partial \tilde{D}_{n,\delta}, \\
        &\frac{\partial}{\partial x_2} u(x;\lambda)=0,\quad x\in \Gamma_{-}\bigcup\Gamma_{+}.
    \end{aligned}
    \right.
\end{equation}
We define the Sobolev space
\begin{equation*} \label{eq8}
    H_{b}^1(\tilde{\Omega}_{\delta};\Delta):=\left\{u\in H^1(\tilde{\Omega}_{\delta}):\Delta u\in L^2(\tilde{\Omega}_{\delta}),  \frac{\partial u}{\partial x_2}|_{\Gamma_{-}\bigcup\Gamma_{+}}=0,u|_{\partial \tilde{D}_{n,\delta}}=0, n\in \mathbb{Z}^*\right\}.
\end{equation*}
Then an interface mode corresponds to an eigenfunction of the operator
$$
\tilde{\mathcal{L}}_{\delta}:H_{b}^1(\tilde{\Omega}_{\delta};\Delta)\subset L^2(\tilde{\Omega}_{\delta})\to L^2(\tilde{\Omega}_{\delta}),\quad \phi\mapsto -\Delta\phi .
$$
Now we can state the main result of the paper as follows:
\begin{theorem} \label{main result}
Under the Assumptions \ref{assump0}, \ref{hypo_reflection symmetry} and \ref{hypo_singular frequency}, if $t_*$ defined in \eqref{eq39} is nonzero, then there exists $\delta_0>0$ such that for $|\delta|<\delta_0$, the spectral problem \eqref{eq7} attains an eigenpair $(u^\star,\lambda^\star) \in L^2(\tilde{\Omega}_\delta)\times\mathbf{R}$ for $\lambda^\star$ near $\lambda_*$. Moreover,
$$
\sigma_p(\tilde{\mathcal{L}}_{\delta})\bigcap\left(\lambda_*-\big|\frac{t_*}{\gamma_*}\big|\delta,\lambda_*+\big|\frac{t_*}{\gamma_*}\big|\delta\right)=\{\lambda^\star\},
$$
where $\gamma_*$ and $t_*$ are defined in \eqref{eq38}. In addition, $u^\star$ decays exponentially away from the interface $\Gamma$ as $|x_1|\to\infty$.
\end{theorem}

\begin{remark}
The assumption $t_* \neq 0$ guarantees that the perturbed structures in Figure 2(a) and 2(b) has a common band gap near the Dirac point $(p_*, \lambda_*)$ of the unperturbed structure; See Corollary \ref{common band gap}. This assumption can be verified when the size of the obstacles is sufficiently small by using asymptotic analysis. 
\end{remark}

\subsection{Outline}
\quad\enspace The rest of the paper is organized as follows. In Section 2, we briefly review the Floquet-Bloch theory for periodic differential operators and 
introduce the Green's functions for periodic waveguide structures.
We also recall the Gohberg-Sigal theory which shall be used in the investigation of the spectral problem \eqref{eq7}.  In Section 3, we present the asymptotic expansions of Bloch eigenvalues and eigenfunctions near the Dirac point $(p_*, \lambda_*)$ for the periodic structures in Figure 2(a)(b). In particular, we prove that a band gap is opened near the Dirac point after the perturbation. Furthermore, the eigenspaces at the band edges are swapped for these two periodic structures. These results are presented in Theorem \ref{dispersion relation near the dirac point} and Corollary \ref{common band gap}. Finally, in Section 4, we reformulate the spectral problem \eqref{eq7} by using a boundary integral equation and prove Theorem \ref{main result} by analyzing the characteristic values of the associated boundary integral operator.

\section{Preliminaries}

\subsection{Floquet-Bloch theory for periodic structures}
For clarity, we focus on the periodic waveguide structure in Figure 1. Let $\mathcal{L}:H_{b}^1(\Omega;\Delta)\subset L^2(\Omega)\to L^2(\Omega)$ be the Laplacian operator. The eigenvalue problem for the waveguide is to solve for all eigenpairs that satisfy the following equations
\begin{equation} \label{eq14}
    \left\{
    \begin{aligned}
        &(\mathcal{L}-\lambda)u(x;\lambda)=0,\quad x \in \Omega ,\\
        &u(x;\lambda)=0,\quad x \in \cup_{n\in \mathbb{Z}^*} \partial D_n, \\
        &\frac{\partial}{\partial x_2} u(x;\lambda)=0,\quad x\in \Gamma_{-}\bigcup\Gamma_{+}.
    \end{aligned}
    \right.
\end{equation}
To study the spectrum $\sigma(\mathcal{L})$ of $\mathcal{L}$, we consider a family of operators $\mathcal{L}(p)$, where $p$ lies in the Brillouin zone $B:=[0, 2\pi]$. More precisely, $\mathcal{L}(p)$
is the Laplacian operator restricted to the function space with the quasi-periodic boundary condition:
$$
\mathcal{L}(p):H_{p,b}^1(\Omega;\Delta)\subset L^2_{loc}(\Omega)\to L^{2}_{loc}(\Omega),\quad \phi\to -\Delta\phi.
$$
In the above, 
$$ 
\begin{aligned}
    H_{p,b}^1(\Omega;\Delta):=\Big\{u\in H_{loc}^1(\Omega):&\Delta u\in L^2_{loc}(Y) ,\frac{\partial u}{\partial x_2}|_{\Gamma_{-}\bigcup\Gamma_{+}}=0,u|_{\partial D_n}=0, n\in\mathbb{Z}^* ,\\
    &u(x+\bm{e}_1)=e^{ip}u(x),\frac{\partial u}{\partial x_1}(x+\bm{e}_1)=e^{ip}\frac{\partial u}{\partial x_1}(x)
    \Big\}. 
\end{aligned}
$$
For each $p\in B$, the spectral theory for self-adjoint operators (cf. \cite{Conca1995FluidsAP, reed1972methods}) states that $\sigma(\mathcal{L}(p))$ consists of a discrete set of real eigenvalues
\begin{equation*} %\label{eq17}
0<\lambda_1(p)\leq\lambda_2(p)\leq\cdots\leq\lambda_n(p)\leq\cdots.
\end{equation*}
We denote by $u_{n}(x;p)$ the eigenfunction  associated with the eigenvalue $\lambda_n(p)$, which is also called the $n$-th Bloch mode at quasi-momentum $p$. 
%In electrodynamics, the eigenvalue $\lambda$ and the frequency of the photonic mode $w$ are related by $\lambda=w^2$ \cite{jackson1999classical}. 
We have the following standard results  for the dispersion relation $\lambda_n(p)$ (cf. \cite{kuchment2016overview}):

\begin{proposition}
(1) $\lambda_n(p)$ ($n\geq 1$) is Lipschitz continuous with respect to $p\in B$;\\
(2) $\lambda_n(p)$ ($n\geq 1$) can be periodically extended to $p\in\mathbf{R}$ with period $2\pi$, i.e. $\lambda_n(p+2\pi)=\lambda_j(p)$. Moreover, $\lambda_n(p)=\lambda_j(-p)$.
\end{proposition}

For each integer $n$, let
\begin{equation*} \label{eq18}
    \lambda_n^{-}=\min_{p\in\mathcal{B}}\lambda_n(p),\quad
    \lambda_n^{+}=\max_{p\in\mathcal{B}}\lambda_n(p).
\end{equation*}
The Floquet-Bloch theory states that the spectrum $\sigma(\mathcal{L})$ is given by
\begin{equation*} \label{eq19}
\sigma(\mathcal{L})=\bigcup_{p\in\mathcal{B}}\sigma(\mathcal{L}(p))=\bigcup_{n\geq 1}I_n,\quad I_n:=\left[\lambda_n^{-},\lambda_n^{+}\right].
\end{equation*}
We say that a band gap opens between the $n$-th and $n+1$-th band
if $\lambda_n^{+}<\lambda_{n+1}^{-}$; See Figure 3 for an example.

%band structure and fold
\begin{figure}
    \centering
    \begin{tikzpicture}[scale=0.2]

\draw[thick,->] (0,0)--(15,0);
\draw[thick,->] (0,0)--(0,28);
\node[below] at (0,0) {$0$};
\node[below] at (14,0) {$2\pi$};
\node[right] at (15.2,0) {$p$};
\node[above] at (0,28) {$\lambda$};
\draw[dashed] (14,0)--(14,27.2);
\draw plot [smooth] coordinates {(0,0.45) (4,1) (10,5) (14,5.55)};
\draw plot [smooth] coordinates {(0,5.55)  (4,5)  (10,1)  (14,0.45)};
\draw[dashed] (7,3)--(0,3);
\draw[thick,red] (0,0.45)--(0,3);
\draw[thick,blue] (0,5.55)--(0,3);
\node[left,red] at (0,1.725) {$I_1$};
\node[left,blue] at (0,4.275) {$I_2$};

\draw plot [smooth] coordinates {(0,27.2)  (4,25)  (10,9)  (14,6.8)};
\draw plot [smooth] coordinates {(0,6.8)  (4,9)  (10,25)  (14,27.2)};
\draw[dashed] (7,17)--(0,17);
\draw[thick,green] (0,6.8)--(0,17);
\draw[thick,yellow] (0,27.2)--(0,17);
\node[left,green] at (0,11.9) {$I_3$};
\node[left,yellow] at (0,22.1) {$I_4$};
\end{tikzpicture}
    \caption{An example of the band structure: the first four bands of the spectrum are depicted in the figure. Note that there is a spectral gap between the second and third bands.
}
    \label{band sturcture and fold}
\end{figure}
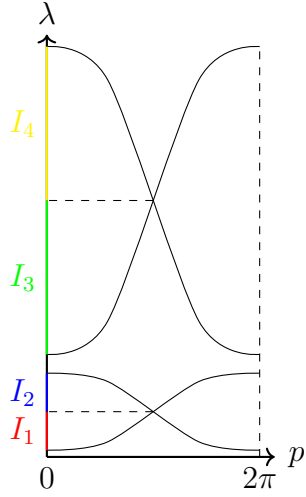

%Now we investigate the band structure for more information. 
%Note that in Figure 3, the first and second branch (also, the third and fourth branch) touches at $\lambda_1^+(=\lambda_2^-)$, which is the Dirac point \cite{https://doi.org/10.48550/arxiv.2101.05966} described in assumption 1.1. Typically, a phase transition happens at the Dirac point \cite{https://doi.org/10.48550/arxiv.2101.05966}. The reason we're interested in the existence of the Dirac point is that we can usually introduce a properly designed perturbation to the system to open the band gap near the Dirac point and introduce an eigenvalue therein, as indicated in Theorem 1.4.

\subsection{The Green's function and representation of solutions for the periodic structure}

We present some properties of the Green's function $G(x, y, \lambda)$ for the waveguide problem \eqref{eq14} at $\lambda=\lambda_*$, which is defined to be the unique physical solution to the following equations:
\begin{equation*} 
    \left\{
    \begin{aligned}
        &(\Delta_x +\lambda)G(x,y;\lambda_*)=\tilde\delta(x-y),\quad x,y \in \Omega,\\
        &G_\delta(x,y;\lambda_*)=0,\quad x \in \cup_{n\in\mathbb{Z}^*}\partial D_{n}, \\
        &\frac{\partial}{\partial x_2} G(x,y;\lambda_*)=0,\quad x\in \Gamma_{-}\bigcup\Gamma_{+}.
    \end{aligned}
    \right.
\end{equation*}
Here $\tilde\delta(\cdot)$ denotes the Dirac delta function. We follow closely the discussions in \cite{fliss2016solutions}. The existence and uniqueness of the Green's function can be established by using the limiting absorbing principle, which gives rise to a physical solution to the wave propagation problem 
\eqref{eq14} as the limit of the unique finite-energy solution of the corresponding model with absorption when the absorption tends to zero. More precisely, the Green's function for the waveguide problem \eqref{eq14} can be defined as the limit of the Green's function $G(x,y;\lambda+ i\varepsilon)$ for the waveguide problem \eqref{eq14} with $\lambda$ being replaced by $\lambda+ i\varepsilon$ as $\varepsilon \to 0$. 

We apply the Floquet-Bloch theory to derive the asymptotic behavior of $G(x, y, \lambda_*)$ at infinity. Recall that $(\lambda_{n}(p),u_n(x;p)), n\geq 1$ are the Bloch eigenpairs associated with the periodic structure. First, 
by using the Floquet transform, the following spectral representation of the Green's function holds:
\begin{eqnarray*} \label{eq88_1}
    G(x,y;\lambda_*)= \lim_{\varepsilon \to 0^+}G(x,y;\lambda_* + i \varepsilon) 
    &=\displaystyle{\frac{1}{2\pi} \lim_{\varepsilon \to 0^+}
    \int_0^{2\pi}\frac{u_{1}(x;p)\overline{u_{1}(y;p)}}{\lambda_*-\lambda_{1}(p) +i \varepsilon}dp + \frac{1}{2\pi}\lim_{\varepsilon \to 0^+}
    \int_0^{2\pi}\frac{u_{2}(x;p)\overline{u_{2}(y;p)}}{\lambda_*-\lambda_{2}(p) +i\varepsilon}dp } \\
     & +  \displaystyle{ \frac{1}{2\pi}
    \int_0^{2\pi}\sum_{n=3}^{+\infty}\frac{u_{n}(x;p)\overline{u_{n}(y;p)}}{\lambda_*-\lambda_{n}(p)}dp. }
\end{eqnarray*}
Note that by Proposition \ref{Existence of Dirac points of the period-1 structure}, there holds
$$
\int_0^{2\pi}\frac{u_{1}(x;p)\overline{u_{1}(y;p)}}{\lambda_*-\lambda_{1}(p) +i \varepsilon}dp + \int_0^{2\pi}\frac{u_{2}(x;p)\overline{u_{2}(y;p)}}{\lambda_*-\lambda_{2}(p) +i \varepsilon}dp = \int_0^{2\pi}\frac{v_{1}(x;p)\overline{v_{1}(y;p)}}{\lambda_*-\mu_{1}(p) +i \varepsilon}dp + \int_0^{2\pi}\frac{v_{2}(x;p)\overline{v_{2}(y;p)}}{\lambda_*-\mu_{2}(p) +i \varepsilon}dp. 
$$
Taking the limit as $\varepsilon \to 0$ yields (see Remark 8 in \cite{fliss2016solutions}): 
\begin{equation} \label{eq-G}
 \begin{aligned}
    G(x,y;\lambda_*)=&-\frac{i}{2}\frac{v_1(x)\overline{v_1(y)}}{\alpha_*} +\frac{1}{2\pi}p.v.\int_{0}^{2\pi}\frac{v_1(x;p)\overline{v_1(y;p)}}{\lambda_*-\mu_1(p)}dp \\
    &-\frac{i}{2}\frac{v_2(x)\overline{v_2(y)}}{\alpha_*} + \frac{1}{2\pi}p.v.\int_{0}^{2\pi}\frac{v_2(x;p)\overline{v_2(y;p)}}{\lambda_*-\mu_2(p)}dp \\
    &+ 
    \frac{1}{2\pi}\int_{0}^{2\pi}\sum_{n\geq 3}\frac{u_n(x;p)\overline{u_n(y;p)}}{\lambda_*-\lambda_n(p)}dp.  
    \end{aligned}
\end{equation}
Moreover, it can be shown that $G$ admits the following decomposition for fixed $y$ (see Remark 9 in \cite{fliss2016solutions}):
\begin{equation} \label{eq-73}
G(x,y;\lambda_*)=G_0^+(x,y;\lambda_*)-i\frac{v_1(x)\overline{v_1(y)}}{\alpha_*}, \quad x_1 \to  \infty, 
\end{equation}
and
\begin{equation} \label{eq73}
    G(x,y;\lambda_*)=G_0^-(x,y;\lambda_*)-i\frac{v_2(x)\overline{v_2(y)}}{\alpha_*},
    \quad x_1 \to  -\infty.
\end{equation}
where $G_0^{+}(x,y;\lambda)$ ($G_0^{-}(x,y;\lambda)$) decays exponentially as $x_1\to +\infty$ ($x_1\to -\infty$). 
We remark that (\ref{eq-73})-(\ref{eq73}) define the so-called radiation condition for the Green's function $G$, with
$v_1$ and $v_2$ being the right and left propagating modes respectively.

We present some properties for $v_1$ and $v_2$ in the next two lemmas.  

\begin{lemma} \label{conjugate vs reflection}
There exists some $\tau\in\mathbb{C}$ s.t. $|\tau|=1$ and
$$
v_2(x_1,x_2)=\tau v_1(-x_1,x_2).
$$
\end{lemma}
\begin{proof}
With Remark \ref{remark on the q-p boundary condition} and Assumption \ref{hypo_reflection symmetry}, it is straightforward to check that $v_1(-x_1,x_2)$ is also a Bloch mode with the same quasi-periodic boundary condition as $v_2$. Then the proof is completed by recalling that the dimension of the eigenspace is one as stated in Assumption \ref{assump0}.
\end{proof}
Let 
$$
\mathcal{V}(\lambda):=\big\{u\in H^1_{loc}(\Omega): (\Delta+\lambda)u(x)=0\enspace \mbox{in} \; \Omega,
\frac{\partial u}{\partial x_2}\big|_{\Gamma_{-}\bigcup\Gamma_{+}}=0,
u\big|_{\partial D_n}=0, n\in\mathbb{Z}^*\big\}.
$$
Then $v_1,v_2\in \mathcal{V}(\lambda_*)$. Moreover, we have the following proposition.
\begin{lemma}\label{orthogonality of the propagating mode}
The right- and left-propagating modes $v_1$ and $v_2$ satisfy
\begin{equation} \label{eq78}
    \int_{\Gamma}\frac{\partial v_1}{\partial x_1}\overline{v_1}dx_2=\frac{i}{2}\mu_1'(\pi)=\frac{i}{2}\alpha_*,\quad\int_{\Gamma}\frac{\partial v_2}{\partial x_1}\overline{v_2}dx_2=-\frac{i}{2}\alpha_*.
\end{equation}
Moreover, for each $u\in\mathcal{V}(\lambda_*)$ that decays exponentially as $x_1\to\infty$, there holds
\begin{equation} \label{eq79}
     \int_{\Gamma}\frac{\partial v_1}{\partial x_1}\overline{u}dx_2=\int_{\Gamma}\frac{\partial \overline{u}}{\partial x_1}v_1 dx_2=\int_{\Gamma}\frac{\partial v_2}{\partial x_1}\overline{u}dx_2=\int_{\Gamma}\frac{\partial \overline{u}}{\partial x_1}v_2 dx_2=0.
\end{equation}
\end{lemma} 

\begin{proof}
We consider the following sesquilinear form defined on $\mathcal{V}(\lambda)$: 
\begin{equation*} \label{eq76}
    q(v,w;x)= \int_{\Gamma_x}\big(\frac{\partial v}{\partial x_1}\overline{w}-\frac{\partial \overline{w}}{\partial x_1}v\big)dx_2,\quad \Gamma_x:=\{x_2:(x,x_2)\in\Omega\},\quad
    v,w\in\mathcal{V}(\lambda). 
\end{equation*}
It is shown in \cite{fliss2016solutions} that $q(v,w;x)$ is independent of $x\in\mathbf{R}$, and moreover, 
\begin{equation*} \label{eq77}
q(v_n,v_m)=i\mu_n'(\pi)\delta_{nm}, \quad n,m = 1, 2.
\end{equation*}
Using Lemma \ref{conjugate vs reflection}, we have
\begin{equation*}
0=q(v_2,v_1)=\int_{\Gamma}\big(\frac{\partial v_2}{\partial x_1}\overline{v_1}-\frac{\partial \overline{v_1}}{\partial x_1}v_2\big)dx_2=-\tau\int_{\Gamma}\big(\frac{\partial v_1}{\partial x_1}\overline{v_1}+\frac{\partial \overline{v_1}}{\partial x_1}v_1\big)dx_2.
\end{equation*}
This, combined with the equality $q(v_1,v_1)= i\mu_1'(\pi)$ for $n=1, 2$, yield 
$\int_{\Gamma}\frac{\partial v_1}{\partial x_1}\overline{v_1}dx_2=\frac{i}{2}\mu_1'(\pi)$. Similarly, we have $\int_{\Gamma}\frac{\partial v_2}{\partial x_1}\overline{v_2}dx_2=-\frac{i}{2}\alpha_*$. Next, since $q(v_i,u;x)$ ($i=1,2$) is independent of $x$, we have for positive integer $N$,
\begin{equation*}
    \begin{aligned}
    &\int_{\Gamma}\big(\frac{\partial v_1}{\partial x_1}\overline{u}-\frac{\partial \overline{u}}{\partial x_1}v_1\big)dx_2=q(v_1,u)=\int_{\Gamma_N}\big(\frac{\partial v_1}{\partial x_1}\overline{u}-\frac{\partial \overline{u}}{\partial x_1}v_1\big)dx_2, \\
    &\int_{\Gamma}\big(\frac{\partial v_1}{\partial x_1}\overline{u}+\frac{\partial \overline{u}}{\partial x_1}v_1\big)dx_2=-\frac{1}{\tau}q(v_2,u)=-\frac{1}{\tau}\int_{\Gamma_N}\big(\frac{\partial v_2}{\partial x_1}\overline{u}-\frac{\partial \overline{u}}{\partial x_1}v_2\big)dx_2.
    \end{aligned}
\end{equation*}
Since $u$ decays exponentially, 
\begin{equation*}
    \begin{aligned}
    &\int_{\Gamma}\big(\frac{\partial v_1}{\partial x_1}\overline{u}-\frac{\partial \overline{u}}{\partial x_1}v_1\big)dx_2=\lim_{N\to +\infty}\int_{\Gamma_N}\big(\frac{\partial v_1}{\partial x_1}\overline{u}-\frac{\partial \overline{u}}{\partial x_1}v_1\big)dx_2=0, \\
    &\int_{\Gamma}\big(\frac{\partial v_1}{\partial x_1}\overline{u}+\frac{\partial \overline{u}}{\partial x_1}v_1\big)dx_2=-\frac{1}{\tau}\lim_{N\to +\infty}\int_{\Gamma_N}\big(\frac{\partial v_2}{\partial x_1}\overline{u}-\frac{\partial \overline{u}}{\partial x_1}v_2\big)dx_2=0.
    \end{aligned}
\end{equation*}
By adding those two identities together, we conclude that $\displaystyle{\int_{\Gamma}\frac{\partial v_1}{\partial x_1}\overline{u}dx_2=
\int_{\Gamma}\frac{\partial \overline{u}}{\partial x_1}v_1dx_2=0
}$. Following the same argument, $\displaystyle{\int_{\Gamma}\frac{\partial v_2}{\partial x_1}\overline{u}dx_2=\int_{\Gamma}\frac{\partial \overline{u}}{\partial x_1}v_2 dx_2=0}$.
\end{proof}

We next present some useful properties of the Green's function $G(x,y;\lambda_*)$.

\begin{lemma} \label{commutativity and parity of the Green function}
For $x,y\in\Omega$ and $x\neq y$, the Green function in \eqref{eq-G} satisfies
\begin{equation} \label{eq80}
    G(x,y;\lambda_*)=G(y,x;\lambda_*).
\end{equation}
On the other hand, when $y\in \Gamma:=\{0\}\times (0,\frac{1}{2})$, there holds
\begin{equation} \label{eq80_1}
    G((x_1,x_2),y;\lambda_*)=G((-x_1,x_2),y;\lambda_*),\quad \forall (x_1,x_2)\in \Omega.
\end{equation}
\end{lemma}

\begin{proof}
We only prove \eqref{eq80} and the proof of \eqref{eq80_1} follows similarly. For $x\neq y$, define $u(z):=G(z,x;\lambda_*)$ and $w(z):=G(z,y;\lambda_*)$. Let $\Omega_{-N,N}:=\Omega\bigcap \big((-N,N)\times (0,\frac{1}{2})$\big) where $N$ is a positive integer.
Then $x,y\in\Omega_{N}$ for sufficiently large $N$. An integration by part yields
\begin{equation} \label{eq134}
    \begin{aligned}
    w(x)-u(y)
    &=\int_{\Omega_{-N,N}}(\Delta+\lambda_*)u(z)w(z)-(\Delta+\lambda_*)w(z)u(z)dz \\
    &=\int_{\Gamma_N}\left(\frac{\partial u}{\partial z_1}w-\frac{\partial w}{\partial z_1}u\right)dz_2
    -\int_{\Gamma_{-N}}\left(\frac{\partial u}{\partial z_1}w-\frac{\partial w}{\partial z_1}u\right)dz_2.
    \end{aligned}
\end{equation}
By \eqref{eq-73} and Lemma \ref{orthogonality of the propagating mode}, as $N\to +\infty$, we have
\begin{equation*}
    \lim_{N\to +\infty}\int_{\Gamma_N}\frac{\partial u}{\partial z_1}w-\frac{\partial w}{\partial z_1}udz_2
    =-\frac{1}{\alpha_*^2}\left(\overline{v_1(x)\cdot v_1(y)}\int_{\Gamma}\frac{\partial v_1}{\partial z_1}v_1dz_2
    -\overline{v_1(x)\cdot v_1(y)}\int_{\Gamma}\frac{\partial v_1}{\partial z_1}v_1dz_2\right)=0.
\end{equation*}
Similarly,
\begin{equation*}
    \lim_{N\to +\infty}\int_{\Gamma_{-N}}\frac{\partial u}{\partial z_1}w-\frac{\partial w}{\partial z_1}udz_2
    =-\frac{1}{\alpha_*^2}\left(\overline{v_2(x)\cdot v_2(y)}\int_{\Gamma}\frac{\partial v_2}{\partial z_1}v_2 dz_2
    -\overline{v_2(x)\cdot v_2(y)}\int_{\Gamma}\frac{\partial v_2}{\partial z_1}v_2 dz_2\right)=0.
\end{equation*}
Thus $w(x)=u(y)$, which implies $G(y,x;\lambda_*)=G(x,y;\lambda_*)$ from definition of $w$ and $u$.
\end{proof}

Finally, as an application of Green's function, we present a representation formula for the solution to the Helmholtz equation in the semi-infinite structure $\Omega_+:=\{x=(x_1,x_2):x\in\Omega,x_1>0\}$. Denote
\begin{equation*}
    \mathcal{V}^+(\lambda):=\big\{u\in H^1_{loc}(\Omega^+):
    (\Delta+\lambda)u(x)=0\; \mbox{in} \; \Omega^+,
    \frac{\partial u}{\partial x_2}\big|_{\Gamma_{-}\bigcup\Gamma_{+}}=0, u\big|_{\partial D_n}= 0, n\geq 1\big\}.
\end{equation*}

\begin{proposition}\label{representation formula}
If $u\in\mathcal{V}^+(\lambda_*)$ satisfies the right-going condition such that
\begin{equation*} \label{eq81}
u(x)=c_{1}v_1(x)+u^+_0(x)\text{ for some $c_{1}\in\mathbb{C}$ and $u^+_0(x)$ decays exponentially as $x_1\to +\infty$},
\end{equation*}
then
\begin{equation}  \label{eq-repre-u}
u(x)=2\int_{\Gamma}G(x,y;\lambda_*)\frac{\partial u}{\partial x_1}(0^+,y_2)dy_2,\quad x\in\Omega^+ .
\end{equation}
%where the expression is continuous to the interface $\Gamma$.
\end{proposition}

\begin{proof}[Proof]
Suppose $u\in\mathcal{V}^+(\lambda_*)$. Define the following even extension of $u$:
\begin{equation} \label{eq82}
    \tilde{u}(x_1,x_2)=\left\{
    \begin{aligned}
    &u(x_1,x_2),\quad x_1\geq 0, \\
    &u(-x_1,x_2),\quad x_1<0.
    \end{aligned}
    \right.
\end{equation}
Let $\Omega_{x,y}:=\Omega\bigcap ((x,y)\times (0,\frac{1}{2}))$. For $y=(y_1,y_2)\in\Omega^+$ and positve integer $N$, an integration by parts shows that
\begin{equation*}
    \begin{aligned}
    \tilde{u}(y)&=\int_{\Omega_{0,N}}(\Delta_x+\lambda_*)G(x,y;\lambda_*)\tilde{u}(x)-(\Delta_x+\lambda_*)\tilde{u}(x)G(x,y;\lambda_*)dx \\
    &=\int_{\Gamma_N}\left(
    \frac{\partial}{\partial x_1}G(x,y;\lambda_*)\tilde{u}(x)-\frac{\partial}{\partial x_1}\tilde{u}(x)G(x,y;\lambda_*)
    \right)dx_2 \\
    &\quad -\int_{\Gamma}\left(
    \frac{\partial}{\partial x_1}G(x,y;\lambda_*)\tilde{u}(0^+,x_2)-\frac{\partial}{\partial x_1}\tilde{u}(0^+,x_2)G(x,y;\lambda_*)
    \right)dx_2,
    \end{aligned}
\end{equation*}
where $\tilde{u}(0^+,x_2)=u(x)\big|_{\Gamma}$. By the right-going condition of $u(\cdot)$, we have
\begin{equation*}
    \begin{aligned}
    \lim_{N\to +\infty}\int_{\Gamma_N}\Big(
    \frac{\partial}{\partial x_1}G(x&,y;\lambda_*)\tilde{u}(x)-\frac{\partial}{\partial x_1}\tilde{u}(x)G(x,y;\lambda_*)
    \Big)dx_2 \\
    &=-\frac{ic_1}{\alpha_*}\overline{v_1(y)}\lim_{N\to +\infty}\int_{\Gamma_N}\left(
    v_1(x)\frac{\partial}{\partial x_1} v_1(x)-v_1(x)\frac{\partial}{\partial x_1}v_1(x)
    \right)dx_2=0.
    \end{aligned}
\end{equation*}
Therefore,
\begin{equation} \label{eq83}
\begin{aligned}
\tilde{u}(y)
&=-\int_{\Gamma}\left(
\frac{\partial}{\partial x_1}G(x,y;\lambda_*)\tilde{u}(0^+,x_2)-\frac{\partial}{\partial x_1}\tilde{u}(0^+,x_2)G(x,y;\lambda_*)
\right)dx_2 \\
&=-\int_{\Gamma}\left(
\frac{\partial}{\partial x_1}G(x,y;\lambda_*)u(0^+,x_2)-\frac{\partial}{\partial x_1}u(0^+,x_2)G(x,y;\lambda_*)
\right)dx_2.
\end{aligned}
\end{equation}
Similarly, an integration by parts in $\Omega_{-N,0}$ yields
\begin{equation*} \label{eq84}
    \begin{aligned}
    0&=\int_{\Omega_{-N,0}}(\Delta_x+\lambda_*)G(x,y;\lambda_*)\tilde{u}(x)-(\Delta_x+\lambda_*)\tilde{u}(x)G(x,y;\lambda_*)dx \\
    &=-\int_{\Gamma_{-N}}\left(
    \frac{\partial}{\partial x_1}G(x,y;\lambda_*)\tilde{u}(x)-\frac{\partial}{\partial x_1}\tilde{u}(x)G(x,y;\lambda_*)
    \right)dx_2 \\
    &\quad +\int_{\Gamma}\left(
    \frac{\partial}{\partial x_1}G(x,y;\lambda_*)\tilde{u}(0^-,x_2)-\frac{\partial}{\partial x_1}\tilde{u}(0^-,x_2)G(x,y;\lambda_*)
    \right)dx_2 
      \end{aligned}
\end{equation*}
By letting $N\to \infty$, we have
\begin{equation*}
    \begin{aligned}
    \lim_{N\to +\infty}\int_{\Gamma_{-N}}\Big(
    &\frac{\partial}{\partial x_1}G(x,y;\lambda_*)\tilde{u}(x)-\frac{\partial}{\partial x_1}\tilde{u}(x)G(x,y;\lambda_*)
    \Big)dx_2 \\
    &=-\frac{ic_1}{\alpha_*}\overline{v_2(y)}\lim_{N\to +\infty}\int_{\Gamma_{-N}}\left(
    v_1(-x_1,x_2)\frac{\partial v_2}{\partial x_1}(x)+v_2(x)\frac{\partial v_1}{\partial x_1}(-x_1,x_2)
    \right)dx_2 \\
    &=-\frac{ic_1}{\tau\alpha_*}\overline{v_2(y)}\lim_{N\to +\infty}\int_{\Gamma_{-N}}\left(
    v_2(x)\frac{\partial v_2}{\partial x_1}(x)-v_2(x)\frac{\partial v_2}{\partial x_1}(x)
    \right)dx_2=0,
    \end{aligned}
\end{equation*}
where \eqref{eq73} is applied in the first equality, and Lemma \ref{conjugate vs reflection} is used in the second one. Thus, we obtain 
\[
0= \int_{\Gamma}\left(
    \frac{\partial}{\partial x_1}G(x,y;\lambda_*)\tilde{u}(0^-,x_2)-\frac{\partial}{\partial x_1}\tilde{u}(0^-,x_2)G(x,y;\lambda_*)
    \right)dx_2.
\]
Note that $\tilde{u}(0^-,x_2)=u(0^+,x_2)$ and $\frac{\partial}{\partial x_1}\tilde{u}(0^-,x_2)=-\frac{\partial}{\partial x_1}u(0^+,x_2)$. Thus,
\[
0= \int_{\Gamma}\left(
    \frac{\partial}{\partial x_1}G(x,y;\lambda_*)u(0^+,x_2)+\frac{\partial}{\partial x_1}u(0^+,x_2)G(x,y;\lambda_*)
    \right)dx_2.
\]
The above equality, combined with \eqref{eq83}, gives 
the representation formula \eqref{eq-repre-u}.
\end{proof}

\subsection{The Green's functions for the perturbed periodic structures}
We introduce the Green's functions for the perturbed periodic structures in Figure 2(a)(b). Let $G_\delta (x,y;\lambda)$ be the  Green's function associated with the perturbed waveguide $\Omega_{\delta}$ in Figure 2(b) that is obtained by using the limiting absorption principle (see also Section 2.2). More precisely, $G_\delta (x,y;\lambda)$ is the unique physical solution that satisfies the following equations: 
\begin{equation} \label{eq85}
    \left\{
    \begin{aligned}
        &(\Delta_x +\lambda)G_\delta(x,y;\lambda)=\tilde\delta(x-y),\quad x,y \in \Omega_\delta,\\
        &G_\delta(x,y;\lambda)=0,\quad x \in \cup_{n\in\mathbb{Z}^*}\partial D_{n,\delta}, \\
        &\frac{\partial}{\partial x_2} G_\delta(x,y;\lambda)=0,\quad x\in \Gamma_{-}\bigcup\Gamma_{+}.
    \end{aligned}
    \right.
\end{equation}
Here again $\tilde\delta(\cdot)$ denotes the Dirac delta function. We are particularly interested in the Green's function when $\lambda$ lies in a spectral band gap of the periodic structure. 
For such $\lambda$, there is no propagating Bloch mode, thus $G_\delta (x,y;\lambda)$ decays exponentially at infinity. This gives the so-called radiation condition for $G_\delta (x,y;\lambda)$. Similarly, we define the Green's function $G_{-\delta} (x,y;\lambda)$ for the periodic structure Figure 2(a). 

We denote $\{(\lambda_{j,\delta}(p), u_{j,\delta}(x;p))\}_{j\geq 1}$ the Bloch eigenpairs of the perturbed periodic structure $\Omega_{\delta}$ for each  $p\in B=[0,2\pi]$. Then $G_{\pm\delta}$ attains the following spectral representation for $\lambda$ in band gaps:
\begin{equation} \label{eq88}
    G_{\pm\delta}(x,y;\lambda)=\frac{1}{2\pi}
    \int_{0}^{2\pi}\sum_{n=1}^{+\infty}\frac{u_{n,\pm\delta}(x;p)\overline{u_{n,\pm\delta}(y;p)}}{\lambda-\lambda_{n,\pm\delta}(p)}dp.
\end{equation}

Finally, we present two Parseval-type identities using the Bloch eigenpairs $(\lambda_{j,\delta}(p), u_{j,\delta}(x;p))$:
\begin{lemma}
For $u\in H^1(\Omega_\delta)$ given by $\displaystyle{u(x) =  \int_{0}^{2\pi}\sum_{n=1}^{\infty} a_{n,\delta}(u;p)u_{n,\delta} (x;p)dp}$,
      \begin{equation} \label{eq87_4}
          \|u\|_{L^2(\Omega_{\delta})}^2 = \int_{0}^{2\pi}\sum_{n=1}^{\infty} |a_{n,\delta}(u;p)|^2 dp,
      \end{equation}
      \begin{equation} \label{eq87_5}
          \|u\|_{H^1(\Omega_{\delta})}^2 = \int_{0}^{2\pi}\sum_{n=1}^{\infty}(1+ \lambda_{n,\delta}(p)) |a_{n,\delta}(u;p)|^2 dp.
      \end{equation}
\end{lemma}
\begin{proof}
Let $C_{\delta}:=Y\cap \Omega=(0,1)\times (0,\frac{1}{2})\backslash\overline{D_1\cup D_2}$. Then \eqref{eq87_4} can be derived from the Parseval identity of the Floquet transform \cite{kuchment2016overview}. On the other hand, \eqref{eq87_5} follows from the identity
\begin{equation*}
    \int_{C_\delta}\nabla u_{n,\delta}(x;p)\cdot\overline{\nabla u_{m,\delta}(x;p)}dx=\lambda_{n,\delta}(p)\int_{C_\delta}u_{n,\delta}(x;p)\overline{u_{m,\delta}(x;p)}dx=\delta_{nm}\lambda_{n,\delta}(p).
\end{equation*}
\end{proof}

\subsection{Gohberg and Sigal theory}
We briefly introduce the Gohberg and Sigal theory, especially the analytic version of the generalized Rouché theorem, which is used to solve the characteristic values of integral operators. We refer to Chapter 1.5 of \cite{ammari2018mathematical} for a thorough exposition of the topic.

Let $X$ and $Y$ be two Banach spaces. An operator $A\in\mathcal{B}(X,Y)$ is said to be \textbf{Fredholm} if the subspace $\text{Ker} A$ is finite-dimensional and the subspace $\text{Ran} A$ is closed in $Y$ and of finite codimension. The index of a Fredholm operator $A$ is defined as
\begin{equation*}
    \text{ind} (A):=\text{dim Ker}A-\text{codim Ran} A.
\end{equation*}
The following propositions show the stability of the index.
\begin{proposition}
If $A:X\to Y$ is a Fredholm operator and $K:X\to Y$ is compact, then  $A+K$ is a Fredholm operator and $\text{ind}(A+K)=\text{ind} (A)$.
\end{proposition}
\begin{proposition} \label{stability of index under norm perturbation}
If $A:X\to Y$ is a Fredholm operator, then there exists $\epsilon>0$ such that for $B\in\mathcal{B}(X,Y)$ and $\|B\|<\epsilon$, $A+B$ is a Fredholm operator and
\begin{equation*}
    \text{ind}(A+B)=\text{ind}(A).
\end{equation*}
\end{proposition}
Let $\mathfrak{U}(z_0)$ be the set of all operator-valued functions with values in $\mathcal{B}(X,Y)$, which are holomorphic in some neighborhood of $z_0$, except possibly at $z_0$. Then the point $z_0$ is called a \textbf{characteristic value} of $A(z)\in\mathfrak{U}(z_0)$ if there exists a vector-valued function $\phi(z)$ with values in $X$ such that
\begin{enumerate}
    \item $\phi(z)$ is holomorphic at $z_0$ and $\phi(z_0)\neq 0$,
    \item $A(z)\phi(z)$ is holomorphic at $z_0$ and vanishes at this point.
\end{enumerate}
Here $\phi(z)$ is called a \textbf{root function} of $A(z)$ associated with the characteristic value $z_0$, and $\phi(z_0)$ is called an \textbf{eigenvector}. By this definition, there exists an integer $m(\phi)\geq 1$ and a vector-valued function $\psi(z)\in Y$, holomorphic at $z_0$, such that
\begin{equation*}
    A(z)\phi(z)=(z-z_0)^{m(\phi)}\psi(z),\quad \psi(z_0)\neq 0.
\end{equation*}
The number $m(\phi)$ is called the \textbf{multiplicity} of the root function $\phi(z)$. For $\phi_0\in\text{Ker}A(z_0)$, the \textbf{rank} of $\phi_0$, which is denoted by $\text{rank}(\phi_0)$, is defined as the maximum of the multiplicities of all root functions $\phi(z)$ with $\phi(z_0)=\phi_0$.

Suppose that $n=\text{dim Ker}A(z_0)<+\infty$ and the ranks of all vectors in $\text{Ker}A(z_0)$ are finite. A system of eigenvectors $\phi_0^j$ ($j=1,2,\cdots,n$) is called a \textbf{canonical system of eigenvectors} of $A(z)$ associated to $z_0$ if for $j=1,2,\cdots,n$, $\text{rank}(\phi_0^j)$ is the maximum of the ranks of all eigenvectors in the direct complement in $\text{Ker}A(z_0)$ of the linear span of the vectors $\phi_0^{1},\cdots,\phi_0^{n-1}$. We call
\begin{equation*}
    N(A(z_0)):=\sum_{j=1}^{n}\text{rank}(\phi_0^{j})
\end{equation*}
the \textbf{null multiplicity} of the characteristic value $z_0$ of $A(z)$. Suppose that $A^{-1}(z)$ exists and is holomorphic in some neighborhood of $z_0$, except possibly at $z_0$. Then the number
\begin{equation*}
    M(A(z_0)):= N(A(z_0))- N(A^{-1}(z_0))
\end{equation*}
is called the \textbf{multiplicity} of $z_0$.

Now, let $V$ be a simply connected bounded domain with a rectifiable boundary $\partial V$. For an analytic operator-valued (or meromorphic operator-valued as in \cite{ammari2018mathematical}) function $A(z)$, it is \textbf{normal with respect to $\partial V$} if 1)$A(z)$ is Fredholm for all $z\in V$; 2)$A^{-1}(z)$ exists in $\overline{V}$, except for a finite number of points; 3)$A(z)$ is continuous for $z\in \partial V$. For such a function $A(z)$, the full multiplicity $\mathcal{M}(A(z);\partial V)$ counts the number of characteristic values of $A(z)$ in $V$ (computed with their multiplicities). Namely,
\begin{equation*}
    \mathcal{M}(A(z);\partial V):=\sum_{i=1}^{\sigma}M(A(z_i)),
\end{equation*}
where $z_i$ ($i=1,2,\cdots,\sigma$) are all characteristic values of $A(z)$ lying in $V$.
The generalized Rouché theorem is stated as follows:
\begin{theorem}[Generalized Rouché theorem]
\label{generalized rouche theorem}
Let $A(z)$ and $B(z)$ be two analytic operator-valued functions that are normal with respect to $\partial V$, and $B(z)$ is continuous on $\partial V$ satisfying the condition
\begin{equation*}
    \|A^{-1}(z)B(z)\|_{\mathcal{B}(X,Y)}<1,\quad z\in \partial V.
\end{equation*}
Then $A(z)+B(z)$ is also normal with respect to $\partial V$ and
\begin{equation*}
    \mathcal{M}(A(z);\partial V)=\mathcal{M}(A(z)+B(z);\partial V).
\end{equation*}
\end{theorem}

\section{Band-gap opening at the Dirac point $(p_*, \lambda_*)$} 
In this section, we investigate the band-gap opening at the Dirac point $(p_*, \lambda_*)$ by perturbing the periodic waveguide in Figure 1.
The perturbed structure in Figure 2(b) is obtained by shifting the obstacle $D_n$ in Figure 1 to $D_{n,\delta}$ for each $n\in\mathbb{Z}^*$, with the mass centers given by
$$
z_{n,\delta}=\left\{
\begin{aligned}
    &z_n-\delta\bm{e}_1,\quad \text{$n$ is odd and positive, or, $n$ is even and negative},\\
    &z_n+\delta\bm{e}_1,\quad \text{$n$ is even and positive, or, $n$ is odd and negative}. 
\end{aligned}
\right.
$$
Denote the perturbed domain by $\Omega_{\delta}:=\mathbf{R}\times(0,1)\backslash \bigcup_{n\in \mathbb{Z}^*} \overline{D_{n,\delta}}$. 
The associated partial differential operator is 
\begin{equation} \label{eq21}
    \begin{aligned}
&\qquad\qquad\mathcal{L}_\delta:H_{b}^1(\Omega_\delta)\subset L^2(\Omega_\delta)\to L^2(\Omega_\delta),\quad \phi\mapsto -\Delta\phi ,\\
    &H_{b}^1(\Omega_\delta;\Delta):=\left\{u\in H^1(\Omega_\delta): \Delta u\in L^2(\Omega), \frac{\partial}{\partial x_2} u|_{\Gamma_{-}\bigcup\Gamma_{+}}=0, u|_{\partial D_{n,\delta}}=0, n\in \mathbb{Z}^*\right\}.
\end{aligned}
\end{equation}
The band structure of the operator $\mathcal{L}_{\delta}$ can be obtained by solving the following eigenvalue problem for each $p\in[0,2\pi]$:
\begin{equation} \label{eq31}
    \left\{
    \begin{aligned}
        &(\mathcal{L}_\delta -\lambda)u(x;p,\lambda)=0,\quad x \in \Omega_\delta ,\\
        &u(x;p,\lambda)=0,\quad x \in \cup_{n\in\mathbb{Z}^*} \partial D_{n,\delta}, \\
        &\frac{\partial}{\partial x_2} u(x;p,\lambda)=0,\quad x\in \Gamma_{-}\bigcup\Gamma_{+} ,\\
        &u(x+\bm{e}_1)=e^{ip}u(x),\quad \frac{\partial u}{\partial x_1}(x+\bm{e}_1)=e^{ip}\frac{\partial u}{\partial x_1}(x).
    \end{aligned}
    \right.
\end{equation}
Similarly, by replacing $\delta$ above with $-\delta$, one can obtain the eigenvalue problem for the perturbed structure shown in Figure 2(a). The corresponding partial differential operator is denoted by $\mathcal{L}_{-\delta}$.

We aim to prove that a common band gap exists for $\sigma(\mathcal{L}_{\delta})$ and $\sigma(\mathcal{L}_{-\delta})$ near the Dirac point $(p_*, \lambda_*)$ of the operator $\mathcal{L}$. Furthermore, we derive the asymptotic expansions of the dispersion relations and the corresponding Floquet-Bloch eigenmodes near the Dirac point.
In particular, we show that the eigenspaces at the band edges are swapped for $\mathcal{L}_{\delta}$ and $\mathcal{L}_{-\delta}$.

\subsection{Boundary-integral equation formulations for eigenvalue problems}
We present a boundary-integral equation formulation for the eigenvalue problem \eqref{eq31}. 
We begin with the case when $\delta=0$, which corresponds to the unperturbed periodic structure. For each $p\in[0,2\pi]$, let $G^e(x,y;p,\lambda)$ be the quasi-periodic Green's function for the empty waveguide that solves the following equations:
\begin{equation*} \label{eq25}
    \left\{
    \begin{aligned}
        &(\Delta_x +\lambda)G^e(x,y;p,\lambda)=\tilde\delta(x-y),\quad x,y \in \mathbf{R}\times(0,\frac{1}{2}),\\
        &\frac{\partial}{\partial x_2} G^e(x,y;p,\lambda)=0,\quad x\in \Gamma_{-}\bigcup\Gamma_{+},\\
        &G^e(x+\bm{e}_1,y;p,\lambda)=e^{ip}G^e(x,y;p,\lambda)\text{ for each $y$},
    \end{aligned}
    \right.
\end{equation*}
where $\tilde\delta(\cdot)$ is the Dirac delta function. $G^e(x,y;p,\lambda)$ can be expressed explicitly as a Fourier series as follows:
\begin{equation} \label{eq26}
    \begin{aligned}
G^e(x,y;p,\lambda)=&\sum_{m \in \mathbb{Z}} \sum_{n \in \mathbb{Z}} \frac{e^{\mathrm{i} p_{m}\left(x_{1}-y_{1}\right)}}{\lambda-p_{m}^{2}-(2n \pi)^{2}}\left(e^{\mathrm{i} 2n \pi\left(x_{2}-y_{2}\right)}+e^{\mathrm{i} 2n \pi\left(x_{2}+y_{2}\right)}\right) \\
=&\sum_{m \in \mathbb{Z}} \frac{e^{\mathrm{i} p_{m}\left(x_{1}-y_{1}\right)}}{\sqrt{p_{m}^{2}-\lambda}\left(e^{-\sqrt{p_{m}^{2}-\lambda}}-e^{\sqrt{p_{m}^{2}-\lambda}}\right)}\left(e^{\sqrt{p_{m}^{2}-\lambda}\left(1-2\left|x_{2}-y_{2}\right|\right)}\right.\\
&\left.+e^{-\sqrt{p_{m}^{2}-\lambda}\left(1-2\left|x_{2}-y_{2}\right|\right)}+e^{\sqrt{p_{m}^{2}-\lambda}\left(1-2\left(x_{2}+y_{2}\right)\right)}+e^{-\sqrt{p_{m}^{2}-\lambda}\left(1-2\left(x_{2}+y_{2}\right)\right)}\right),
\end{aligned}
\end{equation}
where $p_m=p+2m\pi$. We can also expand $G^e(x,y;p,\lambda)$ into the following Floquet series: 
\begin{equation} \label{eq26_1}
G^e(x,y;p,\lambda)=\sum_{n\geq 1}\frac{u_n^e(x;p)\overline{u_n^e(y;p)}}{\lambda-\lambda_n^e(p)},
\end{equation}
where $\{\lambda_n^e(p),u_n^e(x;p)\}_{n\geq 1}$ are the Bloch eigenpairs in the empty waveguide.

Using the Green's function $G^e$, we may express the eigenfunction for \eqref{eq31} as
\begin{equation} \label{eq27}
         u(x;p,\lambda)=\int_{\partial D} G^e(x,y+z_1;p,\lambda) \varphi_{1}(y) d\sigma(y)+\int_{\partial D} G^e(x,y+z_2;p,\lambda) \varphi_{2}(y) d\sigma(y),
\end{equation}
where $\varphi_i\in  H^{-\frac{1}{2}}(\partial D)$ ($i=1,2$) and $d\sigma(y)$ is area element for the boundary surface $\partial D$.  Here and after, $H^{-\frac{1}{2}}(\partial D)$ and $H^{\frac{1}{2}}(\partial D)$ are the standard Sobolev spaces defined on the closed curve $\partial D$.
By imposing the Dirichlet boundary conditions on the obstacle boundaries $\partial D_1$ and $\partial D_2$, we obtain the following homogeneous system:
\begin{equation*} 
\left\{
\begin{aligned}
\int_{\partial D} G^e(x+z_1,y+z_1;p,\lambda) \varphi_{1}(y) d\sigma(y)+\int_{\partial D} G^e(x+z_1,y+z_2;p,\lambda) \varphi_{2}(y) d\sigma(y) =0, \quad x\in \partial D, \\
\int_{\partial D} G^e(x+z_2,y+z_1;p,\lambda) \varphi_{1}(y) d\sigma(y)+\int_{\partial D} G^e(x+z_2,y+z_2;p,\lambda) \varphi_{2}(y) d\sigma(y) =0, \quad x\in \partial D.
\end{aligned}\right.
\end{equation*}

Note that $G^e(x+t\bm{e}_1,y+t\bm{e}_1;p,\lambda)=G^e(x,y;p,\lambda)$ for $t\in\mathbf{R}$ and that $z_2-z_1 = \frac{1}{2}\bm{e}_1$. The above  system is equivalent to
\begin{equation} \label{eq29}
    T(p,\lambda)\bm{\varphi}:=
    \begin{pmatrix}
    T_{11}(p,\lambda) & T_{12}(p,\lambda) \\
    T_{21}(p,\lambda) & T_{22}(p,\lambda) 
    \end{pmatrix}
    \bm{\varphi}=0,
\end{equation}
where $\bm{\varphi}=(\varphi_1\enspace \varphi_2)^T$
and the operators $T_{ij}(p, \lambda)$'s are defined by
\begin{equation*} \label{eq30}
\begin{aligned}
        &T_{11}(p,\lambda)=T_{22}(p,\lambda):\phi(x)\mapsto\int_{\partial D} G^e(x,y;p,\lambda) \phi(y) d\sigma(y), \\
        &T_{12}(p,\lambda):\phi(x)\mapsto\int_{\partial D} G^e(x,y+\frac{1}{2}\bm{e}_1;p,\lambda) \phi(y) d\sigma(y), \\
        &T_{21}(p,\lambda):\phi(x)\mapsto\int_{\partial D} G^e(x+\frac{1}{2}\bm{e}_1,y;p,\lambda) \phi(y) d\sigma(y).
    \end{aligned}
\end{equation*}
Therefore, the spectrum of the operator $\mathcal{L}(p)$ can be obtained from the characteristic values of $T(p,\lambda)$.  In addition, from the solution of \eqref{eq29}, one can recover the corresponding Bloch eigenfunction by the layer potential \eqref{eq27}. We remark that
using the standard layer potential theory, one can show that 
$T_{ij}(p,\lambda) \in \mathcal{B}\left( H^{-\frac{1}{2}}(\partial D), H^{\frac{1}{2}}(\partial D)\right)$ for $1\leq i, j \leq 2$. 

For $\delta\neq0$, following a similar procedure, the eigenvalue problem \eqref{eq31} can be formulated using the following boundary integral equations:
\begin{equation} \label{eq32}
T_\delta(p,\lambda)\bm{\varphi}_\delta:=\begin{pmatrix}
    T_{11}(p,\lambda) & T_{12,\delta}(p,\lambda) \\
    T_{21,\delta}(p,\lambda) & T_{22}(p,\lambda) 
    \end{pmatrix}
    \bm{\varphi}_\delta=0, 
\end{equation}
where the operators $T_{12,\delta}(p,\lambda)$ and $T_{12,\delta}(p,\lambda)$ are defined by
\begin{equation*} \label{eq33}
\left\{
    \begin{aligned}
        &T_{12,\delta}:\phi(x)\mapsto\int_{\partial D} G^e(x,y+(\frac{1}{2}+2\delta)\bm{e}_1;p,\lambda) \phi(y) d\sigma(y) ,\\
        &T_{21,\delta}:\phi(x)\mapsto\int_{\partial D} G^e(x+(\frac{1}{2}-2\delta)\bm{e}_1,y;p,\lambda) \phi(y) d\sigma(y).
    \end{aligned}
\right.
\end{equation*}
% We note that by letting $\delta<0$ in the equation \eqref{eq32}, one obtains the boundary integral equation formulation  for the perturbed structure in Figure 2(a). 

% Finally, we present a lemma that will be used later.

\begin{lemma} \label{particle integral operator at the Dirac point}
The operator $T(p_*,\lambda_*): H^{-\frac{1}{2}}(\partial D)\times H^{-\frac{1}{2}}(\partial D)  \to  H^{\frac{1}{2}}(\partial D)\times H^{\frac{1}{2}}(\partial D)$ is a Fredholm operator with zero index. Moreover, $\text{Ker}\, T(p_*,\lambda_*)=\text{span}\{ \bm{\varphi_1},\bm{\varphi_2}\}$, where $\bm{\varphi_1}$ and $\bm{\varphi_2}$ are defined in Corollary \ref{even odd mode and root function}.
\end{lemma}
\begin{proof}
In the following, we prove that $T(p_*,\lambda_*)$ can be decomposed as the sum of an operator with a bounded inverse and a compact operator; hence it is a Fredholm operator with zero index.

To construct such a decomposition, we first identify each $\bm{\varphi}=(\varphi_1,\varphi_2)\in H^{-\frac{1}{2}}(\partial D)\times H^{-\frac{1}{2}}(\partial D)$ with $\tilde{\varphi}\in H^{-\frac{1}{2}}(\partial D_1\cup \partial D_2)$ by
\begin{equation*}
    \tilde{\varphi}(x)|_{\partial D_1}=\varphi_1(x-z_1),\quad
    \tilde{\varphi}(x)|_{\partial D_2}=\varphi_2(x-z_2).
\end{equation*}
Observe that for $\bm{\psi},\bm{\varphi}\in H^{-\frac{1}{2}}(\partial D)\times H^{-\frac{1}{2}}(\partial D)$,
\begin{equation*} \label{eq34_0}
\begin{aligned}
\langle\overline{\bm{\psi}},&T(p_*,\lambda_*)\bm{\varphi}\rangle \\
=&\int_{\partial D\times \partial D}\left(G^e(x,y;p_*,\lambda_*)\varphi_1(y)+G^e(x,y+\frac{1}{2}\bm{e}_1;p_*,\lambda_*)\varphi_2(y)\right)\overline{\psi_1(x)}d\sigma(y)d\sigma(x) \\
&+\int_{\partial D\times \partial D}\left(G^e(x+\frac{1}{2}\bm{e}_1,y;p_*,\lambda_*)\varphi_1(y)+G^e(x+\frac{1}{2}\bm{e}_1,y+\frac{1}{2}\bm{e}_1;p_*,\lambda_*)\varphi_2(y)\right)\overline{\psi_2(x)}d\sigma(y)d\sigma(x)
\\
=&\int_{\partial D\times \partial D}\Big(G^e(x+z_1,y+z_1;p_*,\lambda_*)\varphi_1(y)+G^e(x+z_1,y+z_2;p_*,\lambda_*)\varphi_2(y)\Big)\overline{\psi_1(x)}d\sigma(y)d\sigma(x) \\
&+\int_{\partial D\times \partial D}\Big(G^e(x+z_2,y+z_1;p_*,\lambda_*)\varphi_1(y)+G^e(x+z_2,y+z_2;p_*,\lambda_*)\varphi_2(y)\Big)\overline{\psi_2(x)}d\sigma(y)d\sigma(x) \\
=&\int_{\partial D_1}\int_{\partial D_1}G^e(x,y;p_*,\lambda_*)\tilde{\varphi}(y)\overline{\tilde{\psi}(x)}d\sigma(y)d\sigma(x)
+\int_{\partial D_1}\int_{\partial D_2}G^e(x,y;p_*,\lambda_*)\tilde{\varphi}(y)\overline{\tilde{\psi}(x)}d\sigma(y)d\sigma(x) \\
&+\int_{\partial D_2}\int_{\partial D_1}G^e(x,y;p_*,\lambda_*)\tilde{\varphi}(y)\overline{\tilde{\psi}(x)}d\sigma(y)d\sigma(x)
+\int_{\partial D_2}\int_{\partial D_2}G^e(x,y;p_*,\lambda_*)\tilde{\varphi}(y)\overline{\tilde{\psi}(x)}d\sigma(y)d\sigma(x) \\
=&\int_{\partial D_1\cup \partial D_2}\int_{\partial D_1\cup \partial D_2}G^e(x,y;p_*,\lambda_*)\tilde{\varphi}(y)\overline{\tilde{\psi}(x)}d\sigma(y)d\sigma(x). 
\end{aligned}
\end{equation*}
Denote $c_n(\tilde{\varphi})=\int_{\partial D_1\cup \partial D_2}\tilde{\varphi(x)}\overline{u_n(x;p_*)}d\sigma(x)$. By \eqref{eq26_1}, we have
\[
\langle\overline{\bm{\psi}},T(p_*,\lambda_*)\bm{\varphi}\rangle=a(\tilde{\varphi},\tilde{\psi}):=\sum_{n\geq 1} \frac{c_n(\tilde{\varphi})\cdot\overline{c_n(\tilde{\psi})}}{\lambda_*-\lambda_n^e(p_*)}. 
\]
%for all $\bm{\psi},\bm{\varphi}\in H^{-\frac{1}{2}}(\partial D)\times H^{-\frac{1}{2}}(\partial D)$.

%Define the following sesquilinear form on $H^{-\frac{1}{2}}(\partial D_1\cup \partial D_2)$
%$$
%a(\tilde{\varphi},\tilde{\psi}).
%$$
 We further introduce the following two sesquilinear forms
\begin{equation} \label{eq34}
\begin{aligned}
&a^{(0)}(\tilde{\varphi},\tilde{\psi}):=
\sum_{\{n:\lambda_n^e(p_*)>\lambda_*\}} \frac{c_n(\tilde{\varphi})\cdot\overline{c_n(\tilde{\psi})}}{\lambda_*-\lambda_n^e(p_*)}
-\sum_{\{n:\lambda_n^e(p_*)<\lambda_*\}} c_n(\tilde{\varphi})\cdot\overline{c_n(\tilde{\psi})}, \\
&a^{(1)}(\tilde{\varphi},\tilde{\psi}):=
\sum_{\{n:\lambda_n^e(p_*)<\lambda_*\}} \frac{c_n(\tilde{\varphi})\cdot\overline{c_n(\tilde{\psi})}}{\lambda_*-\lambda_n^e(p_*)}
+\sum_{\{n:\lambda_n^e(p_*)<\lambda_*\}} c_n(\tilde{\varphi})\cdot\overline{c_n(\tilde{\psi})},
\end{aligned}
\end{equation}
and their associated operators $\tilde{T}^{(0)},\tilde{T}^{(1)}\in\mathcal{B}(H^{-\frac{1}{2}}(\partial D_1\cup \partial D_2),H^{\frac{1}{2}}(\partial D_1\cup \partial D_2))$. Then we have $a(\tilde{\varphi},\tilde{\psi})=a^{(0)}(\tilde{\varphi},\tilde{\psi})+a^{(1)}(\tilde{\varphi},\tilde{\psi})$; thus, it is sufficient to show that $\tilde{T}^{(0)}$ is invertible and $\tilde{T}^{(1)}$ is compact to conclude that $T(p_*,\lambda_*)$ is a Fredholm operator with zero index. Since $\{n:\lambda_n^e(p_*)<\lambda_*\}$ is a finite set, it is clear that $\tilde{T}^{(1)}$ has finite rank and thus is compact. So we only need to prove the invertibility of $\tilde{T}^{(0)}$.
To this end, we shall show that the sesquilinear form $a^{(0)}$ is coercive and symmetric in the following two steps.

\medskip

Step 1. We prove the following inequality
\begin{equation} \label{eq_A_2}
|a^{(0)}(\tilde{\varphi},\tilde{\varphi})|\gtrsim \|\tilde{\varphi}\|^2,
\end{equation}
which ensures that $\tilde{T}^{(0)}$ is injectivity and that its range space is closed. First, \eqref{eq34} shows that
\begin{equation} \label{eq36}
|a^{(0)}(\tilde{\varphi},\tilde{\varphi})|
=\sum_{\{n:\lambda_n^e(p_*)>\lambda_*\}} \frac{|c_n(\tilde{\varphi})|^2}{|\lambda_*-\lambda_n^e(p_*)|}
+\sum_{\{n:\lambda_n^e(p_*)<\lambda_*\}} |c_n(\tilde{\varphi})|^2
\gtrsim \sum_{n\geq 1} \frac{|c_n(\tilde{\varphi})|^2}{|\lambda_*-\lambda_n^e(p_*)|}.
\end{equation}
Next, let $\text{Tr}:H^1_{p_*}(Y)\to H^{\frac{1}{2}}(\partial D_1\cup \partial D_2)$ be the trace operator ($H^1_{p_*}(Y)$ consists of the functions in $H^1(Y)$ that satisfy the $p_*-$quasi-periodic boundary condition in $x_1-$direction) and $E:H^{\frac{1}{2}}(\partial D_1\cup \partial D_2)\to H^1_{p_*}(Y)$ be the Sobolev extension operator such that $\text{Tr}\circ E=id|_{H^{\frac{1}{2}}(\partial D_1\cup \partial D_2)}$. Since $\{u_n^e(x;p_*)\}$ forms an orthogonal basis of $H^1_{p_*}(Y)$, we can expand $E\phi$ for any $\phi\in H^{\frac{1}{2}}(\partial D_1\cup \partial D_2)$ as $(E\phi)(x)\equiv \sum_{n\geq 1}a_n u_n^e(x;p_*)$. Similar to \eqref{eq87_5}, the following identity holds:
$$
\|E\phi\|^2_{H^1(Y)}=\sum_{n\geq 1}(1+\lambda_n^e(p_*))|a_n|^2.
$$
As a result,
\begin{equation*}
\begin{aligned}
\Big|\langle \tilde{\varphi}, \overline{\phi}\rangle_{H^{-\frac{1}{2}}(\partial D_1\cup \partial D_2)\times H^{\frac{1}{2}}(\partial D_1\cup \partial D_2)}\Big|
&=\Big|\langle \tilde{\varphi}, \overline{\text{Tr}(E\phi)}\rangle\Big|
=\Big|\sum_{n\geq 1}\overline{a_n}\cdot c_n(\tilde{\varphi})\Big| \\
&\leq(\sum_{n\geq 1}|\lambda_*-\lambda_n^e(p_*)||a_n|^2)^{\frac{1}{2}}
\cdot (\sum_{n\geq 1}\frac{|c_n(\tilde{\varphi})|^2}{|\lambda_*-\lambda_n^e(p_*)|})^{\frac{1}{2}}
\end{aligned}
\end{equation*}
Since $\lim_{n\to +\infty}\lambda_n^e(p_*)=+\infty$, $\sum_{n\geq 1}|\lambda_*-\lambda_n^e(p_*)||a_n|^2 \lesssim \sum_{n\geq 1}(1+\lambda_n^e(p_*))|a_n|^2$. Thus we have
\begin{equation*}
\begin{aligned}
\Big|\langle \tilde{\varphi}, \overline{\phi}&\rangle_{H^{-\frac{1}{2}}(\partial D_1\cup \partial D_2)\times H^{\frac{1}{2}}(\partial D_1\cup \partial D_2)}\Big|
\lesssim
(\sum_{n\geq 1}(1+\lambda_n^e(p_*))|a_n|^2)^{\frac{1}{2}}
\cdot (\sum_{n\geq 1}\frac{|c_n(\tilde{\varphi})|^2}{|\lambda_*-\lambda_n^e(p_*)|})^{\frac{1}{2}} \\
&=\|E\phi\|_{H^1(Y)}\left(\sum_{n\geq 1}\frac{|c_n(\tilde{\varphi})|^2}{|\lambda_*-\lambda_n^e(p_*)|}\right)^{\frac{1}{2}}
\lesssim \|\phi\|_{H^{\frac{1}{2}}(\partial D_1\cup \partial D_2)}\left(\sum_{n\geq 1}\frac{|c_n(\tilde{\varphi})|^2}{|\lambda_*-\lambda_n^e(p_*)|}\right)^{\frac{1}{2}}.
\end{aligned}
\end{equation*}
Therefore
\begin{equation} \label{eq36_1}
    \|\tilde{\varphi}\|\leq (\sum_{n\geq 1}\frac{|c_n(\tilde{\varphi})|^2}{|\lambda_*-\lambda_n^e(p_*)|})^{\frac{1}{2}}.
\end{equation}
Then the desired estimate \eqref{eq_A_2} follows directly from \eqref{eq36} and \eqref{eq36_1}.

\medskip
Step 2. We prove that $a^{(0)}(\cdot,\cdot)$ is symmetric. Then the injectivity, which is proved in Step 1, implies that $\text{Ran}\, T^{(0)}$ is dense. In fact, from the definition of  $a^{(0)}(\cdot,\cdot)$ in \eqref{eq34}, it follows that
\begin{equation} \label{eq_A_3}
\begin{aligned}
a^{(0)}(\tilde{\varphi},\tilde{\psi})&:=
\sum_{\{n:\lambda_n^e(p_*)>\lambda_*\}} \frac{c_n(\tilde{\varphi})\cdot\overline{c_n(\tilde{\psi})}}{\lambda_*-\lambda_n^e(p_*)}
-\sum_{\{n:\lambda_n^e(p_*)<\lambda_*\}} c_n(\tilde{\varphi})\cdot\overline{c_n(\tilde{\psi})} \\
&=\sum_{\{n:\lambda_n^e(p_*)>\lambda_*\}} \overline{\frac{\overline{c_n(\tilde{\varphi})}\cdot c_n(\tilde{\psi})}{\lambda_*-\lambda_n^e(p_*)}}
-\sum_{\{n:\lambda_n^e(p_*)<\lambda_*\}} \overline{\overline{c_n(\tilde{\varphi})}\cdot c_n(\tilde{\psi})} \\
&=\overline{a^{(0)}(\tilde{\psi},\tilde{\varphi})}.
\end{aligned}
\end{equation}

We conclude from Step 1 and 2 that $\tilde{T}^{(0)}$ is bijective. By the open mapping theorem,  $\tilde{T}^{(0)}$ is also  invertible with a bounded inverse. Consequently, $T(p_*,\lambda_*)$ is a Fredholm operator with zero index.

Finally, it is clear that any $\bm{\varphi}\in \text{Ker}\, T(p_*,\lambda_*)$ corresponds to the Neumann trace of a Bloch eigenfunction on $\partial D_1\cup \partial D_2$ at the Dirac point $(p_*, \lambda_*)$. Since the eigenspace at the Dirac point is 2-dimensional as stated in Proposition \ref{Existence of Dirac points of the period-1 structure}, Corollary \ref{even odd mode and root function} shows that $\text{Ker}\, T(p_*,\lambda_*)=\text{span}\{ \bm{\varphi_1},\bm{\varphi_2}\}$.
\end{proof}

\subsection{Band-gap opening and the asymptotic expansions}
We investigate band-gap opening at the Dirac point $(p_*, \lambda_*)$ of $\mathcal{L}$ for the perturbed periodic structures using the integral equation formulation presented in the previous subsection. 
First note that the quasi-periodic Green function $G^e(x,y;p,\lambda)$ in \eqref{eq26} is analytic with respect to $\lambda$ for $\lambda$ near $\lambda_*$  (by Assumption \ref{hypo_singular frequency}), the following property holds:

\begin{lemma} \label{taylor approximation of T_delta}
For $|p-p_*|,|\lambda-\lambda_*|,|\delta|\ll 1$, $T_\delta(p,\lambda)$ is analytic in $p,\lambda$. Moreover, there exist $T_p,T_\lambda,S\in \mathcal{B}\left(H^{-\frac{1}{2}}(\partial D)\times H^{-\frac{1}{2}}(\partial D), H^{\frac{1}{2}}(\partial D)\times H^{\frac{1}{2}}(\partial D)\right)$ such that when $|p-p_*|,|\lambda-\lambda_*|,|\delta|\ll 1$, there holds
\begin{equation} \label{eq37}
    T_\delta(p,\lambda)=T(p_*,\lambda_*)+(p-p_*)T_p+(\lambda-\lambda_*)T_\lambda+\delta S+\mathcal{O}((p-p_*)^2,(\lambda-\lambda_*)^2,\delta^2).
\end{equation}
\end{lemma}
 
Properties of the operators $T_p$, $T_\lambda$, and $S$ above are given in Proposition \ref{non-degeneracy 1} and \ref{non-degeneracy 2} below, wherein $\langle \cdot,\cdot\rangle$ denote the dual pair between $H^{-\frac{1}{2}}(\partial D)\times H^{-\frac{1}{2}}(\partial D)$ and $H^{\frac{1}{2}}(\partial D)\times H^{\frac{1}{2}}(\partial D)$.

\begin{proposition} \label{non-degeneracy 1}
Let $\bm{\varphi}_1,\bm{\varphi}_2$ be the boundary potentials at the Dirac point $(p_*, \lambda_*)$ defined in \eqref{eq24-1}-\eqref{eq24-2}. Then there exist real numbers $\gamma_*,\theta_*$ such that for $i,j=1,2$, 
\begin{equation} \label{eq38}
    \begin{aligned}
    & \langle \bm{\varphi_j},T_\lambda \bm{\varphi_i} \rangle=\gamma_*\delta_{ij},\\
    & \langle \bm{\varphi_j},T_p \bm{\varphi_i} \rangle=\sqrt{-1}\theta_*(1-\delta_{ij})(-1)^{i-1}.
     \end{aligned}
\end{equation}

%and the operators $T_\lambda$ and $T_p$ are defined in Lemma \ref{taylor approximation of T_delta}. 
\end{proposition}

\begin{proposition}{\label{prop-alpha}}
Let $\theta_*$ and $\gamma_*$ be defined in (\ref{eq38}), and $\alpha_*$ be defined in Assumption \ref{assump0}, we have 
$$
\alpha_*=\left|\frac{\theta_*}{\gamma_*}\right|. 
$$
\end{proposition}

\begin{proposition} \label{non-degeneracy 2}
Let $\bm{\varphi}_1$ and $\bm{\varphi}_2$ be the same as in Proposition \ref{non-degeneracy 1}. Then there exist $t_*\in\mathbf{R}$ such that
\begin{equation}\label{eq39}
\langle \bm{\varphi_j},S\bm{\varphi_i} \rangle=t_*(-1)^{i-1}\delta_{ij}. 
\end{equation}
\end{proposition}
The proofs of Proposition \ref{non-degeneracy 1} and \ref{non-degeneracy 2} contain technical calculations and are given in Appendix B and C, respectively. On the other hand, Proposition \ref{prop-alpha} is a consequence of Theorem \ref{dispersion relation near the dirac point} below and will be proved afterward. Now we are ready to state the main result of this section on the asymptotics of the perturbed dispersion relations and eigenmodes near the Dirac point. 

\begin{theorem}[Dispersion relation and eigenmodes near the Dirac point] \label{dispersion relation near the dirac point}
Under Assumption \ref{assump0}, \ref{hypo_reflection symmetry}, \ref{hypo_singular frequency}, and the assumption that $t_*$ defined in \eqref{eq39} is nonzero, there exists $\delta_0>0$ such that for all $|\delta|<\delta_0$ and quasi-momentum p with $|p-p_*|\ll 1$, there are two branches of dispersion curves $(p,\lambda_{2,\delta}(p))$ and $(p,\lambda_{1,\delta}(p))$ which are the characteristic values of \eqref{eq32}. In addition, they admit the following expansions:
\begin{equation} \label{eq40}
    \begin{aligned}
    &\lambda_{2,\delta}(p)
    =\lambda_*+\frac{1}{|\gamma_*|}\sqrt{\delta^2t^2_*+\theta^2_*(p-p_*)^2}\left(1+\mathcal{O}(p-p_*,\delta)\right), \\
    &\lambda_{1,\delta}(p)
    =\lambda_*-\frac{1}{|\gamma_*|}\sqrt{\delta^2t^2_*+\theta^2_*(p-p_*)^2}\left(1+\mathcal{O}(p-p_*,\delta)\right).  \\
    \end{aligned}
\end{equation}
Moreover, when $\delta>0$ and assume further that $\gamma_*,\theta_*,t_* >0$, the Floquet-Bloch eigenmodes defined in \eqref{eq31} admit the asymptotic expansion for $|p-p_*|\ll 1$:
\begin{equation} \label{eq41}
    \begin{aligned}
    &u_{2,\delta}(x;p):=u_{\delta}(x;p,\lambda_{2,\delta}(p))=\frac{\sqrt{-1}\theta_*(p-p_*)}{\delta t_*+\sqrt{\delta^2 t_*^2+\theta_*^2 (p-p_*)^2}}\phi_1(x)+\phi_2(x)+\mathcal{O}(p-p_*,\delta),\\
    &u_{1,\delta}(x;p):=u_{\delta}(x;p,\lambda_{1,\delta}(p))=\phi_1(x)+\frac{\sqrt{-1}\theta_*(p-p_*)}{\delta t_*+\sqrt{\delta^2 t_*^2+\theta_*^2 (p-p_*)^2}}\phi_2(x)+\mathcal{O}(p-p_*,\delta),
    \end{aligned}
\end{equation}
where $\phi_1$ and $\phi_2$ are the eigenmodes at the Dirac point $(\pi,\lambda_*)$ of $\mathcal{L}$.
Similarly,
\begin{equation} \label{eq42}
    \begin{aligned}
    &\lambda_{2,-\delta}(p)
    =\lambda_*+\frac{1}{|\gamma_*|}\sqrt{\delta^2t^2_*+\theta^2_*(p-p_*)^2}\left(1+\mathcal{O}(p-p_*,\delta)\right), \\
    &\lambda_{1,-\delta}(p)
    =\lambda_*-\frac{1}{|\gamma_*|}\sqrt{\delta^2t^2_*+\theta^2_*(p-p_*)^2}\left(1+\mathcal{O}(p-p_*,\delta)\right),  \\
    \end{aligned}   
\end{equation}
and
\begin{equation} \label{eq43}
    \begin{aligned}
    &u_{2,-\delta}(x;p):=u_{-\delta}(x;p,\lambda_{2,-\delta}(p))=\phi_1(x)-\frac{\sqrt{-1}\theta_*(p-p_*)}{\delta t_*+\sqrt{\delta^2 t_*^2+\theta_*^2 (p-p_*)^2}}\phi_2(x)+\mathcal{O}(p-p_*,\delta),\\
    &u_{1,-\delta}(x;p):=u_{-\delta}(x;p,\lambda_{1,-\delta}(p))=-\frac{\sqrt{-1}\theta_*(p-p_*)}{\delta t_*+\sqrt{\delta^2 t_*^2+\theta_*^2 (p-p_*)^2}}\phi_1(x)+\phi_2(x)+\mathcal{O}(p-p_*,\delta).
    \end{aligned}
\end{equation}
\end{theorem}

\begin{remark} \label{assump on positive constants}
In the above theorem, 
the assumption that $\gamma_*,\theta_*>0$ is made for the ease of the presentation. Without it, similar asymptotic expansions of Floquet-Bloch eigenvalues and eigenmodes, see \eqref{eq40}-\eqref{eq43}, can still be derived using the same arguments.
\end{remark}

\begin{remark}
It can be observed from \eqref{eq40} that the perturbed dispersion curves are locally parabolic when $t_*\neq 0$.  In addition, Theorem \ref{dispersion relation near the dirac point} implies that
\begin{equation*} \label{eq44}
    \begin{aligned}
    &\lim_{\delta\to 0+}\lambda_{2,\delta}(p_*)=
    \lim_{\delta\to 0+}\lambda_{1,\delta}(p_*)=
    \lim_{\delta\to 0+}\lambda_{2,-\delta}(p_*)=
    \lim_{\delta\to 0+}\lambda_{1,-\delta}(p_*)=\lambda_*,\\
    &\lim_{\delta\to 0+}u_{2,\delta}(x,p_*)=\lim_{\delta\to 0+}u_{1,-\delta}(x,p_*)=\phi_2(x),\\
    &\lim_{\delta\to 0+}u_{1,\delta}(x,p_*)=\lim_{\delta\to 0+}u_{2,-\delta}(x,p_*)=\phi_1(x).
    \end{aligned}
\end{equation*}
In other words, when the perturbation is introduced, the degeneracy at the Dirac point is lifted; See Figure 4 for such an illustration. Very importantly, for the two perturbations with $\delta>0$ and $\delta<0$, although the eigenvalues near the Dirac point are the same, the corresponding eigenspaces are swapped. This demonstrates the topological phase transition of the periodic structure at the Dirac point when $\delta=0$.
\end{remark}

\begin{figure}
    \centering
    \includegraphics[scale=0.3]{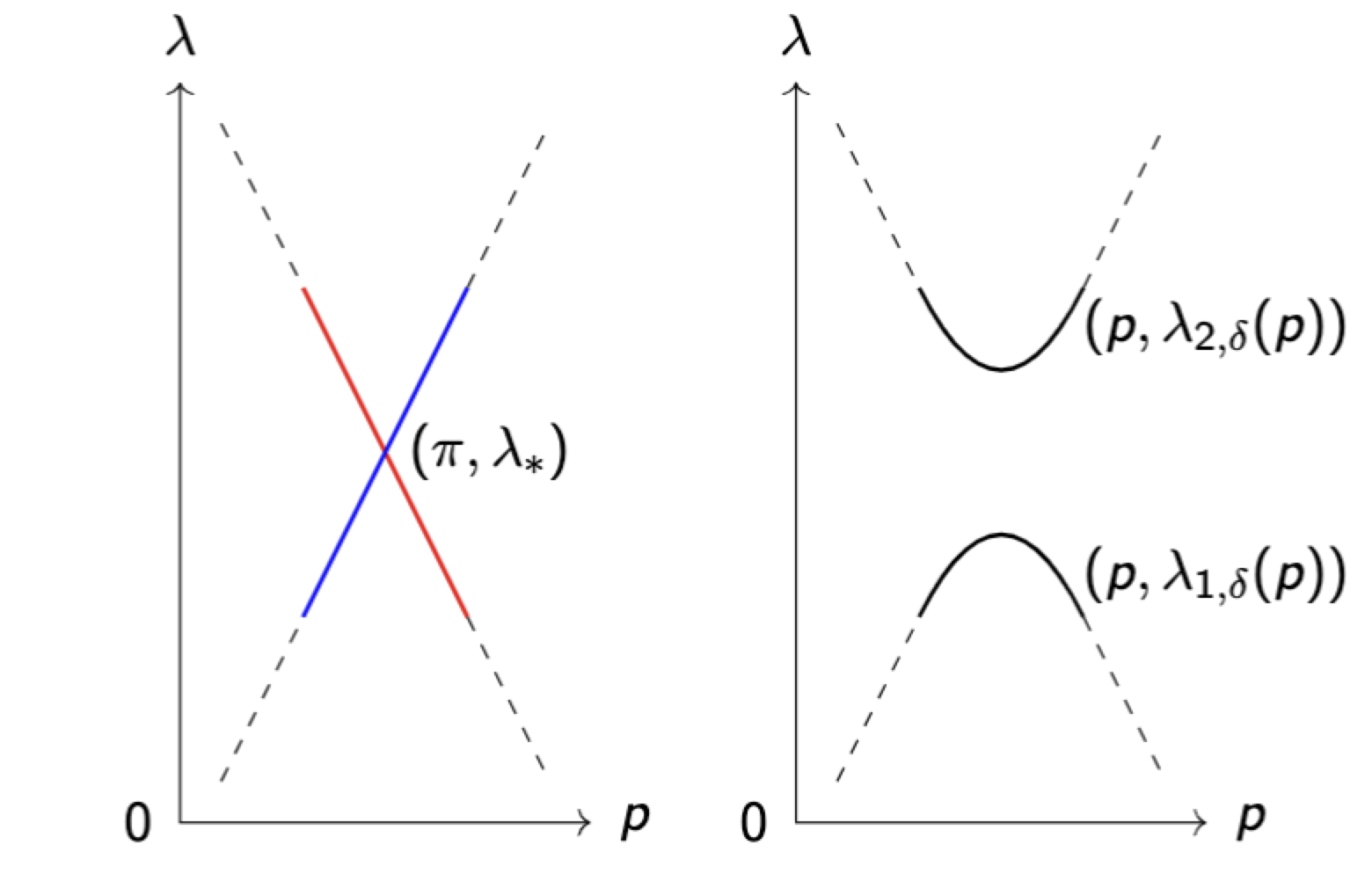}
    \caption{Lifting of the degeneracy near the Dirac point.}
\end{figure}

The following corollary states the existence of the common band gap for $\sigma (\mathcal{L}_\delta)$ and $\sigma (\mathcal{L}_{-\delta})$:
\begin{corollary}[The common band gap of $\sigma (\mathcal{L}_\delta)$ and $\sigma (\mathcal{L}_{-\delta})$] \label{common band gap}
Let $\delta>0$ and $0<c<1$ be a constant.  We follow all the notations and assumptions in Theorem \ref{dispersion relation near the dirac point} and define the real interval
\begin{equation} \label{eq45}
I_\delta=(E_{1,\delta},E_{2,\delta}):=\left(\lambda_*-c\delta|\beta_*|,\lambda_*+c\delta|\beta_*|\right),
\end{equation}
where $\beta_*=\frac{t_*}{\gamma_*}$. Then there exists
$\delta_0>0$, such that for all $0<\delta<\delta_0$, there holds
$$
I_\delta\bigcap\sigma(\mathcal{L}_\delta)=I_\delta\bigcap\sigma(\mathcal{L}_{-\delta})=\emptyset .
$$
\end{corollary}

\begin{proof}
We only consider $\mathcal{L}_\delta$, and the proof for $\mathcal{L}_{-\delta}$ is identical. From Theorem \ref{dispersion relation near the dirac point}, for any $c\in (0,1)$, there exists $\delta_0>0$ such that for all $0<\delta<\delta_0$, the dispersion relations $\lambda_{1,\delta}(p)$ and $\lambda_{2,\delta}(p)$ satisfy
$$
\max_{|p-p_*|\leq\delta}\lambda_{1,\delta}(p)\leq\lambda_*-c\delta\beta_*,\quad
\min_{|p-p_*|\leq\delta}\lambda_{2,\delta}(p)\geq\lambda_*+c\delta\beta_*.
$$
Thus to show that $I_\delta$ is indeed a band gap, it suffices to show that for $p\in\mathcal{B}\bigcap\left\{p:|p-\pi|\geq\delta\right\}$,  $\lambda_{1,\delta}(p),\lambda_{2,\delta}(p)\notin I_{\delta}$, and $\lambda_{n,\delta}(p)\notin I_\delta$ for all $n\geq 3$ and $p\in\mathcal{B}$. Indeed, this follows from Assumption \ref{assump0} that only the first and second dispersion curves touch at the Dirac point $(p_*,\lambda_*)$, and they are away from $\lambda=\lambda_*$ when $p\neq p_*$.
%Then the proof is complete.
\end{proof}

\subsection{Proof of Proposition \ref{prop-alpha} and Theorem \ref{dispersion relation near the dirac point}}

\begin{proof}[Proof of Theorem \ref{dispersion relation near the dirac point}]

For ease of presentation, we assume that the constants $\theta_*,\gamma_*$ and $t_*$ defined in Proposition \ref{non-degeneracy 1} and \ref{non-degeneracy 2} are all positive. Recall Remark \ref{assump on positive constants} that such an assumption is not essential.

We first show that for each $p$ near $\pi$ and small $\delta$, the operator $T_\delta(\lambda;p):=T_\delta(p,\lambda)$ has two characteristic values (counted with multiplicity) 
with $|\lambda-\lambda_*|\ll 1$. To this end, it is sufficient to consider $p=\pi$, while the cases of $p\neq\pi$ can be treated similarly.

Note the characteristic values of $T(\lambda;p_*)$ correspond to the Bloch eigenvalues of the problem \eqref{eq31} at $p=p_*$ and $\delta =0$. There is a small neighborhood $U(\lambda_*)$ of $\lambda=\lambda_*$ in the complex plane such that $\lambda_*$ is the only characteristic value of $T(\lambda;p_*)$ inside $\overline{U(\lambda_*)}$ and its multiplicity is two. The analyticity of $T(\lambda;p_*)$ in the variable $\lambda$ (See Lemma \ref{taylor approximation of T_delta}) implies that $T(\lambda;p_*)$ is normal with respect to $\partial U(\lambda_*)$, so does $T_{\delta}(\lambda;p_*)$. On the other hand, by Lemma \ref{taylor approximation of T_delta}, $\|T(\lambda;p_*)^{-1}(T(\lambda;p_*)-T_{\delta}(\lambda;p_*))\|< 1$ for $\lambda\in \partial U(\lambda_*)$ for $\delta$ sufficiently small. Then Theorem \ref{generalized rouche theorem} implies that $T_\delta(\lambda;p_*)$ attains two characteristic values (counted with multiplicity) for $\lambda\in U(\lambda_*)$.

Next, we calculate the asymptotic expansion of the characteristic values of $T_{\delta}(\lambda;p)$ and their associated eigenvectors by a perturbation argument.

Step 1.  We first set up the framework to conduct the perturbation. In the vicinity of the Dirac point $(p_*, \lambda_*)$, we write the quasi-momentum as
\begin{equation} \label{eq46}
    p=p_*+p^{(1)} ,
\end{equation}
where $|p^{(1)}|\ll 1$. We seek a solution to \eqref{eq32} of the form
\begin{equation} \label{eq47}
    \begin{aligned}
    &\qquad\qquad\qquad\qquad\qquad\qquad\qquad
    \lambda(p_*+p^{(1)})=\lambda_*+\lambda^{(1)}, \\
    &\bm{\varphi}(x;p_*+p^{(1)})=\bm{\varphi}^{(0)}+\bm{\varphi}^{(1)},\quad \bm{\varphi}^{(0)}=a\bm{\varphi_1}+b\bm{\varphi_2}
    \in \text{Ker}\, T(p_*,\lambda_*),\quad \bm{\varphi}^{(1)}\in (\text{Ker}\, T(p_*,\lambda_*))^\perp ,
    \end{aligned}
\end{equation}
where $|\lambda^{(1)}|\ll 1$, and $(\text{Ker}\, T(p_*,\lambda_*))^\perp$ is the orthogonal complement of $\text{Ker}\, T(p_*,\lambda_*)$ in $H^{-\frac{1}{2}}(\partial D)\times H^{-\frac{1}{2}}(\partial D)$ . Here $\bm{\varphi_1}$ and $\bm{\varphi_2}$ are defined in Corollary \ref{even odd mode and root function}.  

Substituting \eqref{eq37}, \eqref{eq46} and \eqref{eq47} into \eqref{eq32}, we obtain the following equation for $\bm{\varphi}^{(1)}$:
\begin{equation} \label{eq48}
\begin{aligned}
T(p_*,\lambda_*)\bm{\varphi}^{(1)}=&-\left(p^{(1)} T_p+ \lambda^{(1)} T_\lambda +\delta S+\mathcal{O}((p^{(1)})^2,(\lambda^{(1)})^2,\delta^2)\right)\bm{\varphi}^{(0)} \\
&-\left(p^{(1)} T_p+ \lambda^{(1)} T_\lambda +\delta S+\mathcal{O}((p^{(1)})^2,(\lambda^{(1)})^2,\delta^2)\right)\bm{\varphi}^{(1)}.
\end{aligned}
\end{equation}

\medskip

Step 2. We solve \eqref{eq48} by following a Lyapunov-Schmidt reduction argument. Since $\text{Ran}\, (T(p_*,\lambda_*))$ is closed in $H^{\frac{1}{2}}(\partial D)\times H^{\frac{1}{2}}(\partial D)$, we introduce the orthogonal projection $Q:H^{\frac{1}{2}}(\partial D)\times H^{\frac{1}{2}}(\partial D)\to \text{Ran}\, (T(p_*,\lambda_*))$. By applying $Q$ to \eqref{eq48}, we obtain
\begin{equation} \label{eq51}
\centering
\begin{aligned}
T(p_*,\lambda_*)\bm{\varphi}^{(1)}=&-Q\left(p^{(1)} T_p+ \lambda^{(1)} T_\lambda +\delta S+\mathcal{O}((p^{(1)})^2,(\lambda^{(1)})^2,\delta^2)\right)\bm{\varphi}^{(0)} \\
&-Q\left(p^{(1)} T_p+ \lambda^{(1)} T_\lambda +\delta S+\mathcal{O}((p^{(1)})^2,(\lambda^{(1)})^2,\delta^2)\right)\bm{\varphi}^{(1)}.
\end{aligned}
\end{equation}
%where \eqref{eq50_1} is used to yield that $T(p_*,\lambda_*)\bm{\varphi}^{(1)}\in R(Q_{\perp})$. 
 By Lemma \ref{particle integral operator at the Dirac point}, $T^{-1}(p_*,\lambda_*)\in \mathcal{B}(\text{Ran}\, (T(p_*,\lambda_*)),(\text{Ker}\, T(p_*,\lambda_*))^\perp)$. Then \eqref{eq51} can be rewritten as
\begin{equation} \label{eq53}
    (I+A)\bm{\varphi}^{(1)}=-A\bm{\varphi}^{(0)},
\end{equation}
where the map $A$ is defined as
\begin{equation*} \label{eq54}
    \bm{f}\mapsto A\bm{f}:=T^{-1} (p_*,\lambda_*)Q\left(p^{(1)} T_p+ \lambda^{(1)} T_\lambda +\delta S+\mathcal{O}((p^{(1)})^2,(\lambda^{(1)})^2,\delta^2)\right)\bm{f}.
\end{equation*}
Thus, for $p^{(1)},\lambda^{(1)},\delta$ sufficiently small, $(I+A)^{-1}$ exists, which implies that \eqref{eq53} is uniquely solvable:
\begin{equation} \label{eq55}
    \bm{\varphi}^{(1)}=-(I+A)^{-1}A\bm{\varphi}^{(0)}.
\end{equation}
With \eqref{eq47}, we may rewrite \eqref{eq55} as 
\begin{equation*} \label{eq56}
    \bm{\varphi}^{(1)}=\bm{\varphi}^{(1)}(x;p^{(1)},\lambda^{(1)},\delta)
    =\bm{g_1}(x;p^{(1)},\lambda^{(1)},\delta)a+\bm{g_1}(x;p^{(1)},\lambda^{(1)},\delta)b,
\end{equation*}
where the map $(p^{(1)},\lambda^{(1)},\delta)\mapsto \bm{g_{i}}(x;p^{(1)},\lambda^{(1)},\delta)$ ($i=1,2$) is smooth from a neighborhood of $(0,0,0)$ to $(\text{Ker}\, T(p_*,\lambda_*))^\perp$ with the following estimate:
\begin{equation*} \label{eq57}
   \left\| \bm{g_{i}}(x;p^{(1)},\lambda^{(1)},\delta)\right\|
   \lesssim |p^{(1)}|+|\lambda^{(1)}|+|\delta|.
\end{equation*}
\medskip
Step 3. We take dual pair with $\bm{\varphi}_{i}$ ($i=1,2$) on both sides of \eqref{eq48}. From the identity $\langle \overline{\bm{\psi}},T(p_*,\lambda_*)\bm{\varphi}\rangle =\overline{\langle \overline{\bm{\varphi}},T(p_*,\lambda_*)\bm{\psi}\rangle}$ for any $\bm{\varphi},\bm{\psi}$, we obtain the following equations for $(a,b)^T$:
\begin{equation} \label{eq59}
    \mathcal{M}(p^{(1)},\lambda^{(1)},\delta)
    \begin{pmatrix}
    a \\b
    \end{pmatrix}
    =0,
\end{equation}
with
\begin{equation} \label{eq60}
\begin{aligned}
\mathcal{M}(p^{(1)},\lambda^{(1)},\delta)&:= 
    \begin{pmatrix}
    \gamma_* \lambda^{(1)} +t_*\delta & -i\theta_*p^{(1)} \\
    i\theta_*p^{(1)} & \gamma_* \lambda^{(1)} -t_*\delta \\
    \end{pmatrix}
    \quad+\mathcal{O}((\lambda^{(1)})^2,(p^{(1)})^2,\delta^2).
\end{aligned}
\end{equation}
Note that the higher-order term in \eqref{eq60} is smooth in $(p^{(1)},\lambda^{(1)},\delta)$ near $(0,0,0)$. Thus $(p,\lambda)$ with $\lambda=\lambda_*+\lambda^{(1)}(p^{(1)},\delta)$ is a characteristic value of $T_\delta (p,\lambda)$ defined in \eqref{eq32} if and only if $\lambda^{(1)}=\lambda^{(1)}(p^{(1)},\delta)$ solves the following equation
\begin{equation} \label{eq61}
    F(p^{(1)},\lambda^{(1)},\delta)\equiv \det\mathcal{M}(p^{(1)},\lambda^{(1)},\delta)=\gamma_*^2(\lambda^{(1)})^2-t_*^2\delta^2-\theta_*^2(p^{(1)})^2+\rho(p^{(1)},\lambda^{(1)},\delta)=0,
\end{equation}
where $\rho(p^{(1)},\lambda^{(1)},\delta)$ is smooth near $(0,0,0)$ and satisfies
\begin{equation} \label{eq62}
   \left\| \rho(p^{(1)},\lambda^{(1)},\delta)\right\|
   = \mathcal{O}((\lambda^{(1)})^3,(p^{(1)})^3,\delta^3).
\end{equation}

\medskip

Step 4. We solve $\lambda^{(1)}=\lambda^{(1)}(p^{(1)},\delta)$ from \eqref{eq61} for each $p^{(1)}$ and $\delta$. We first note that $\pm\frac{1}{|\gamma_*|}\sqrt{t^2_*\delta^2+\theta^2_*(p^{(1)})^2}$ give two branches of solutions if we drop the remainder term $\rho$. Thus, we seek a solution to \eqref{eq61} in the following form
\begin{equation} \label{eq63}
    \lambda^{(1)}(p^{(1)},\delta)=\frac{x}{|\gamma_*|}\sqrt{t^2_*\delta^2+\theta^2_*(p^{(1)})^2}
\end{equation}
with $|x|$ close to 1. It is clear that $x$ depends on $p^{(1)}$ and $\delta$. By substituting \eqref{eq63} into \eqref{eq61}, we obtain the following equation of $x$, with $p^{(1)}$ and $\delta\neq 0$ being viewed as two parameters:
\begin{equation} \label{eq63_1}
    \begin{aligned}
    H(x;p^{(1)},\delta)
    &:=\frac{1}{t^2_*\delta^2+\theta^2_*(p^{(1)})^2}F(p^{(1)},\frac{x}{|\gamma_*|}\sqrt{t^2_*\delta^2+\theta^2_*(p^{(1)})^2},\delta) \\
    &=x^2-1+\rho_1(x;p^{(1)},\delta)=0,
    \end{aligned}
\end{equation}
where $\rho_1(x;p^{(1)},\delta):=\frac{\rho(\frac{x}{|\gamma_*|}\sqrt{t^2_*\delta^2+\theta^2_*(p^{(1)})^2},p^{(1)},\delta)}{t^2_*\delta^2+\theta^2_*(p^{(1)})^2}$. Now we consider the upper branch of solution to (\ref{eq63_1}) with $|x-1|\ll 1$. Note that the following estimates hold uniformly in $x$ by \eqref{eq62}:
\begin{equation} \label{eq63_2}
\left|\rho_1(x;p^{(1)},\delta)\right|=\mathcal{O}(p^{(1)},\delta).
\end{equation}
We conclude that there exists a unique solution $x_s(p^{(1)},\delta)$ to the equation \eqref{eq63_1} with the estimate 
$x_s(p^{(1)},\delta)=1+\mathcal{O}(p^{1},\delta)$
for $\delta,|p^{(1)}|\ll 1$. It follows from \eqref{eq63} that there exists a unique solution $\lambda^{(1)}_{+}(p^{(1)},\delta)$ to the equation \eqref{eq61} near $\frac{1}{|\gamma_*|}\sqrt{t^2_*\delta^2+\theta^2_*(p^{(1)})^2}$.  Moreover, 
\begin{equation*}
    \lambda^{(1)}_{+}(p^{(1)},\delta)=\frac{x_s(p^{(1)},\delta)}{|\gamma_*|}\sqrt{t^2_*\delta^2+\theta^2_*(p^{(1)})^2})=\frac{1}{|\gamma_*|}\sqrt{t^2_*\delta^2+\theta^2_*(p^{(1)})^2})\left(1+\mathcal{O}(p^{(1)},\delta)\right).
\end{equation*}
Similarly, we can derive that
\begin{equation*}
    \lambda^{(1)}_{-}(p^{(1)},\delta)=-\frac{1}{|\gamma_*|}\sqrt{t^2_*\delta^2+\theta^2_*(p^{(1)})^2})\left(1+\mathcal{O}(p^{(1)},\delta)\right).
\end{equation*}
Note that $\lambda_{1,\delta}(p)=\lambda_*+\lambda^{(1)}_{+}(p-p_*,\delta)$, $\lambda_{2,\delta}(p)=\lambda_*+\lambda^{(1)}_{-}(p-p_*,\delta)$, whence \eqref{eq40} follows. 
\medskip

Step 5. Finally, by substituting $\lambda^{(1)}(p)=\lambda_{i,\delta}(p)-\lambda_*$ in \eqref{eq59} for $i=1,2$ respectively, one obtains the following solutions accordingly, 
$$
\begin{pmatrix}
    a_1 \\ b_1
\end{pmatrix}
=
\begin{pmatrix}
    1 \\
    \frac{\sqrt{-1}\theta_*(p-p_*)}{\delta t_*+\sqrt{\delta^2 t_*^2+\theta_*^2 (p-p_*)^2}}
\end{pmatrix}
,\quad
\begin{pmatrix}
    a_2 \\ b_2
\end{pmatrix}
=
\begin{pmatrix}
    \frac{\sqrt{-1}\theta_*(p-p_*)}{\delta t_*+\sqrt{\delta^2 t_*^2+\theta_*^2 (p-p_*)^2}} \\
    1
\end{pmatrix}.
$$
The asymptotic expansions of the eigenmodes $u_{i,\delta}(\cdot;p)$ in \eqref{eq41} are obtained by substituting $(a_i,b_i)^T$ above into \eqref{eq47} and then using the layer potentials in \eqref{eq27}. The proof for \eqref{eq42} and \eqref{eq43} follows in a similar manner.

\end{proof}

\medskip

\begin{proof}[Proof of Proposition \ref{prop-alpha}] The proposition is a consequence of Theorem \ref{dispersion relation near the dirac point}.  By letting $\delta=0$ in the asymptotic formula (\ref{eq40}), we see that the first two dispersion functions of the unperturbed structure near $p=p^*$ admit the following expansions:
\begin{equation*} \label{eq66_B}
    \begin{aligned}
        &\lambda_{1}(p)=\lambda_*+\lambda^{(1)}_{-}(p-p_*)
    =\lambda_*-\Big|\frac{\theta_
    *}{\gamma_*}\Big||p-p_*|+\mathcal{O}((p-p_*)^2),\\
        &\lambda_{2}(p)=\lambda_*+\lambda^{(1)}_{+}(p-p_*)
    =\lambda_*+\Big|\frac{\theta_
    *}{\gamma_*}\Big||p-p_*|+\mathcal{O}((p-p_*)^2).
    \end{aligned}
\end{equation*}
Therefore, the slope of the dispersion curve at the intersection point (which is the Dirac point) is $|\frac{\theta_*}{\gamma_*}|$. The slope is consistent with the one proposed in Proposition \ref{Existence of Dirac points of the period-1 structure} and Assumption \ref{assump0}. Thus we have
\begin{equation*}
    \alpha_*=\Big|\frac{\theta_*}{\gamma_*}\Big|.
\end{equation*}
\end{proof}

\section{Interface mode for the waveguide with perturbations}
In this section, we prove the existence of an interface mode for the waveguide in Figure 2(c), as stated in Theorem \ref{main result}. In Section 4.1, we reformulate the eigenvalue problem \eqref{eq7} as a boundary integral equation. The asymptotic expansions of the related boundary integral operators are derived with respect to the perturbation parameter $\delta$ in Section 4.2. Finally, we prove Theorem \ref{main result} in Section 4.3 by investigating the characteristic values of the associated boundary integral operator.

%Throughout the paper, when we consider integral operators defined on $\Gamma$, we will work in the spaces $\tilde{H}^{-\frac{1}{2}}(\Gamma)$ and $H^{\frac{1}{2}}(\Gamma)$. }

\subsection{Boundary-integral formulation for the joint system with two semi-infinite perturbed media}
In this subsection, we reformulate the eigenvalue problem \eqref{eq7} for the joint system in Figure 2(c) by using a boundary integral equation. The idea is to match the wave fields on both sides of the waveguide over the interface $\Gamma$. To proceed, we first introduce some notations. Recall that $\Gamma:=\{0\}\times (0,\frac{1}{2})$.  Let $\tilde{\Gamma}:=\Gamma\cup \left([0,+\infty)\times\{0\}\right)\cup \left([0,+\infty)\times\{\frac{1}{2}\}\right)$ be the boundary of the semi-infinite waveguide $(0,+\infty)\times (0,\frac{1}{2})$.
We define 
$$
H^{\frac{1}{2}}(\Gamma):=\{u=U|_{\Gamma}:U\in H^{\frac{1}{2}}(\tilde{\Gamma})\},
$$
and 
$$
\tilde{H}^{-\frac{1}{2}}(\Gamma):=\{u=U|_{\Gamma}:U\in H^{-\frac{1}{2}}(\tilde{\Gamma})\text{ and }supp (U)\subset \overline{\Gamma}\}.
$$ 
Then ${H}^{\frac{1}{2}}(\Gamma)$ is the dual space of $\tilde H^{-\frac{1}{2}}(\Gamma)$ and vice versa.  
We also denote the left and right domain in Figure 2(c) by $\tilde{\Omega}_\delta^-:=\tilde{\Omega}_\delta\bigcap\left((-\infty,0)\times(0,\frac{1}{2})\right)$ and $\tilde{\Omega}_\delta^+:=\tilde{\Omega}_\delta\bigcap\left((0,+\infty)\times(0,\frac{1}{2})\right)$, respectively. 

Suppose that $u(x;\lambda) \in L^2(\tilde{\Omega}_\delta) $ is a solution of the eigenvalue problem \eqref{eq7}. We express $u(x;\lambda)$ as
\begin{equation} \label{eq102}
    u(x;\lambda)=\left\{
    \begin{aligned}
        &u^+(x),\quad x\in\tilde{\Omega}_{\delta}^+, \\
        &u^-(x),\quad x\in\tilde{\Omega}_{\delta}^-,
    \end{aligned}
    \right.
\end{equation}
where $u^+$ and $u^-$ satisfy respectively the following equations
\begin{equation*} \label{eq99}
    \left\{
    \begin{aligned}
        &(\Delta_x +\lambda)u^+(x;\lambda)=0,\quad x \in \tilde{\Omega}_{\delta}^+,\\
        &u^+(x;\lambda)=0,\quad x \in \partial \tilde{D}_{n,\delta}\enspace(n=1,2,\cdots),\\
        &\frac{\partial}{\partial x_2} u^+(x;\lambda)=0,\quad x\in \Gamma_{-}\bigcup\Gamma_{+}.
    \end{aligned}
    \right.
    \quad
    \left\{
    \begin{aligned}
        &(\Delta_x +\lambda)u^-(x;\lambda)=0,\quad x \in \tilde{\Omega}_{\delta}^-,\\
        &u^-(x;\lambda)=0,\quad x \in \partial \tilde{D}_{n,\delta}\enspace(n=-1,-2,\cdots),\\
        &\frac{\partial}{\partial x_2} u^-(x;\lambda)=0,\quad x\in \Gamma_{-}\bigcup\Gamma_{+}.
    \end{aligned}
    \right.
\end{equation*}
For $\lambda\in I_{\delta}$, a common band gap for the left and the right periodic structures near the Dirac point $(p_*, \lambda_*)$, it is known that (see \cite{fliss2016solutions}) $u^+$ and $u^-$  decay exponentially away from the interface $\Gamma$ as $|x_1|\to \infty$ in $\tilde{\Omega}_{\delta}^+$ and $\tilde{\Omega}_{\delta}^-$ respectively. Moreover, 
the following interface conditions hold:
\begin{equation} \label{eq103}
    u(0^-,x_2;\lambda)=u(0^+,x_2;\lambda),
\end{equation}
\begin{equation} \label{eq104}
    \frac{\partial u}{\partial x_1}(0^-,x_2;\lambda)=\frac{\partial u}{\partial x_1}(0^+,x_2;\lambda).
\end{equation}
We have the following representation formulas for $u^\pm$. 
\begin{lemma} \label{representation formula and jump formula in the band-gap}
Let $\lambda\in I_{\delta}$, then
\begin{equation}\label{eq-u_plus}
    u^+(x;\lambda)=2\int_{\Gamma}G_\delta (x,y;\lambda)\phi^+(y;\lambda)d\sigma(y),\quad
    \phi^+=\frac{\partial u^+}{\partial x_1}\Big|_{\Gamma},
\end{equation}
\begin{equation}\label{eq-u_mius}
    u^-(x;\lambda)=-2\int_{\Gamma}G_{-\delta} (x,y;\lambda)\phi^-(y;\lambda)d\sigma(y),\quad
    \phi^-=\frac{\partial u^-}{\partial x_1}\Big|_{\Gamma},
\end{equation}
where $G_{\pm\delta} (x,y;\lambda)$ is the Green's function defined in \eqref{eq85}. Moreover, for each $\phi\in \tilde{H}^{-\frac{1}{2}}(\Gamma)$, the following identity holds:
\begin{equation}\label{eq_int_eqn_phi}
    \left(\frac{\partial}{\partial x_1}\int_{\Gamma}G_\delta (x,y;\lambda)\phi(y)d\sigma(y)\right)(0^+,x_2)=\frac{\phi(x_2)}{2}.
\end{equation}
\end{lemma}
\begin{proof}
We first prove the representation formula for $u^+$. The formula for $u^-$ can be proved similarly.
Let $G_\delta^{\Gamma}$ be Green's function in the perturbed semi-infinity waveguide with the Neumann boundary condition on $\Gamma$:
\begin{equation*}
    \left\{
    \begin{aligned}
        &(\Delta_x +\lambda)G^\Gamma_\delta(x,y;\lambda)=\tilde{\delta}(x-y),\quad x,y \in \Omega_\delta,\\
        &G^\Gamma_\delta(x,y;\lambda)=0,\quad x \in \cup_{n\geq 1}\partial D_{n,\delta}, \\
        &\frac{\partial}{\partial {x_2}} G^\Gamma_\delta(x,y;\lambda)=0,\quad x\in \Gamma_{-}\bigcup\Gamma_{+}, \\
        &\frac{\partial}{\partial {x_1}} G^\Gamma_\delta(x,y;\lambda)=0,\quad x\in \Gamma.
    \end{aligned}
    \right.
\end{equation*}

Since $\Omega_{\delta}$ attains reflection symmetry, for $y\in\Gamma$, there holds
\[
    G^\Gamma_\delta(x,y;\lambda)=\lim_{y_1\to 0^-}G_\delta(x,y;\lambda)+\lim_{y_1\to 0^+}G_\delta(x,y;\lambda)=2G_\delta(x,y;\lambda),
\]
where $G_\delta (x,y;\lambda)$ is the Green's function for the perturbed periodic structure defined in \eqref{eq85}.
Therefore, for $\lambda\in I_{\delta}$, $G^\Gamma_\delta(x,y;\lambda)$ decays exponentially for $|x_1|\to \infty$. Then an integration by parts yields
\begin{equation*}
u^+(x;\lambda)=\int_{\Gamma}G^{\Gamma}_\delta (x,y;\lambda)\phi^+(y;\lambda)d\sigma(y)=2\int_{\Gamma}G_\delta (x,y;\lambda)\phi^+(y;\lambda)d\sigma(y)
,\quad
\phi^+=\frac{\partial u^+}{\partial x_1}\Big|_{\Gamma}, 
\end{equation*}
whence the representation formula for $u^+$ follows.

Next, we show \eqref{eq_int_eqn_phi}. Let $\left(\frac{\partial}{\partial x_1}\int_{\Gamma}G_\delta (x,y;\lambda)\phi(y)d\sigma(y)\right)(0^+,x_2)=\frac{\psi(x_2)}{2}$.  
Then \eqref{eq-u_plus} gives
\begin{equation} \label{eq169}
    \int_{\Gamma}G_\delta (x,y;\lambda)\phi(y)d\sigma(y)
    =\int_{\Gamma}G_\delta (x,y;\lambda)\psi(y)d\sigma(y)\quad \mbox{for} \; x\in \Omega_{\delta}^+.
\end{equation}
Define
\begin{equation} \label{eq170}
    u(x) := \int_{\Gamma}G_\delta (x,y;\lambda)(\phi(y)-\psi(y))d\sigma(y)\quad \mbox{for} \; x\in \Omega_{\delta}^+.
\end{equation}
Then \eqref{eq169} implies that $u(x)\equiv 0$ for $x\in \Omega_{\delta}^+$. By the reflection symmetry of the Green's function $G_\delta (x,y;\lambda)$ (which follows from the same arguments as in Lemma \ref{commutativity and parity of the Green function}), $u(x)$ in \eqref{eq170} can be naturally extended to $x\in \Omega_{\delta}^-$.  This gives an even extension of $u$, i.e. $u(x_1,x_2):=u(-x_1,x_2)$ for $x_1<0$. Note that the extended $u(x)$ also vanishes in $\Omega_\delta$. On the other hand, if we rewrite $u$ as
\begin{equation*}
\begin{aligned}
u(x)= \int_{\Gamma}G_\delta (x,y;\lambda)(\phi(y)-\psi(y))d\sigma(y)
=\int_{\Omega_\delta}G_\delta (x,y;\lambda)(\tilde{\phi}-\tilde{\psi})(y)dy,
\end{aligned}
\end{equation*}
where $(\tilde{\phi}-\tilde{\psi})(y):=\phi(y_2)\tilde{\delta}(y_1)-\psi(y_2)\tilde{\delta}(y_1)$ (here $\tilde{\delta}$ denotes the delta function), then 
\begin{equation*}
(\Delta+\lambda)u=\tilde{\phi}-\tilde{\psi}.
\end{equation*}
However, $(\Delta_x +\lambda)u\equiv 0$ for $u(x)\equiv 0$ when $x\in \Omega_\delta$. Hence we conclude that $\phi=\psi$ and Lemma \ref{representation formula and jump formula in the band-gap} is proved.
\end{proof}
By Lemma \ref{representation formula and jump formula in the band-gap} and \eqref{eq104}, $u^{\pm}$ can be expressed as 
\begin{equation} \label{eq100}
    u^+(x;\lambda)=2\int_{\Gamma}G_\delta (x,y;\lambda)\phi(y;\lambda)d\sigma(y),
\end{equation}
\begin{equation} \label{eq101}
    u^-(x;\lambda)=-2\int_{\Gamma}G_{-\delta} (x,y;\lambda)\phi(y;\lambda)d\sigma(y),
\end{equation}
for some $\phi\in \tilde{H}^{-\frac{1}{2}}(\Gamma)$. By substituting \eqref{eq100} and \eqref{eq101} into the interface condition \eqref{eq103}, we obtain  
\begin{equation} \label{eq106}
    2\int_{\Gamma} \left (G_{\delta} ((0,x_2),(0,y_2);\lambda)+G_{-\delta} ((0,x_2),(0,y_2);\lambda)\right)\phi((0,y_2);\lambda)dy_2=0.
\end{equation}

Let us introduce the following single-layer boundary integral operators
\begin{equation} \label{eq89}
        \mathbb{G}_{\pm \delta}(\lambda): \tilde{H}^{-\frac{1}{2}}(\Gamma)\to H^{\frac{1}{2}}(\Gamma),\quad \phi(y) \mapsto 2\int_{\Gamma}G_{\pm \delta}(x,y;\lambda)\phi(y;\lambda)dy, 
\end{equation}
and set
\begin{equation*} \label{eq108}
    \tilde{\mathbb{G}}_\delta(\lambda):\tilde{H}^{-\frac{1}{2}}(\Gamma)\to H^{\frac{1}{2}}(\Gamma),\quad
\tilde{\mathbb{G}}_\delta(\lambda)=\mathbb{G}_\delta(\lambda)+\mathbb{G}_{-\delta}(\lambda).
\end{equation*}
Then \eqref{eq106} is equivalent to the following boundary integral equation
\begin{equation} \label{eq109}
\tilde{\mathbb{G}}_\delta(\lambda)\phi=0,\quad\mbox{for}\; \phi\in \tilde{H}^{-\frac{1}{2}}(\Gamma) .
\end{equation}
In summary, \eqref{eq7} attains an interface mode $u(x;\lambda)$, if and only if $\lambda\in I_{\delta}$ is a characteristic value of $\tilde{\mathbb{G}}_\delta(\lambda)$.

\subsection{Properties of boundary integral operators}
Here and henceforth, for each $\lambda\in I_\delta$, we parameterize $\lambda$ as $\lambda:=\lambda_*+\delta\cdot h$ for $h\in J:=(\frac{E_{1,\delta}-\lambda_*}{\delta},\frac{E_{2,\delta}-\lambda_*}{\delta})$. The complex neighborhood  of $J$ is denoted by $\tilde{J}:=\{h\in \mathbb{C}:|h|<|\frac{E_{1,\delta}-E_{2,\delta}}{2\delta}|\}$.

We investigate the boundary integral operator  $\tilde{\mathbb{G}}_\delta(\lambda_*+\delta\cdot h)=\mathbb{G}_\delta(\lambda_*+\delta\cdot h)+\mathbb{G}_{-\delta}(\lambda_*+\delta\cdot h)$ for $h\in\tilde{J}$. The results obtained will pave the way for applying the Gohberg and Sigal theory to prove Theorem \ref{main result}.

To this end,  we first define $\mathbb{T}_{\pm\delta}(\lambda_*+\delta\cdot h): \tilde{H}^{-\frac{1}{2}}(\Gamma)\to H^{\frac{1}{2}}(\Gamma)$ by
\begin{equation} \label{eq93}
    \begin{aligned}
        \mathbb{T}_{\pm \delta}(\lambda_*+\delta\cdot h)=&\int_{0}^{2\pi}\sum_{n\geq 3}\frac{\langle \cdot,\overline{u_{n,\pm\delta}(y;p)}\rangle}{\lambda_*+\delta\cdot h-\lambda_{n,\pm\delta}(p)}u_{n,\pm\delta}(x;p)dp \\ 
        &+\int_{[0,\pi-\delta^{\frac{1}{3}}]\bigcup[\pi+\delta^{\frac{1}{3}},2\pi]}\frac{\langle \cdot,\overline{u_{2,\pm\delta}(y;p)}\rangle}{\lambda_*+\delta\cdot h-\lambda_{2,\pm\delta}(p)}u_{2,\pm\delta}(x;p)dp \\
        &+\int_{[0,\pi-\delta^{\frac{1}{3}}]\bigcup[\pi+\delta^{\frac{1}{3}},2\pi]}\frac{\langle \cdot,\overline{u_{1,\pm\delta}(y;p)}
        \rangle}{\lambda_*+\delta\cdot h-\lambda_{1,\pm\delta}(p)}u_{1,\pm\delta}(x;p)dp. 
    \end{aligned}
\end{equation}
Here and henceforth, $\langle\cdot,\cdot\rangle$ denotes the duality pair on ${\tilde{H}^{-\frac{1}{2}}(\Gamma)\times H^{\frac{1}{2}}(\Gamma)}$.

Recall that $\phi_1, \phi_2$ are the two Dirac eigenmodes, with $\phi_1$ being odd and $\phi_2$ being even (see Corollary \ref{even odd mode and root function}). Let 
$\mathbb{P}$ be the projection operator defined by
\begin{equation*}
    \mathbb{P}:\tilde{H}^{-\frac{1}{2}}(\Gamma)\to H^{\frac{1}{2}}(\Gamma),\quad
    \psi\mapsto \langle\psi ,\overline{\phi}_2\rangle\phi_2.
\end{equation*}
We further introduce the following four functions:
\begin{equation*} \label{eq91}
    \begin{aligned}
     &f_1(h;\delta):=
     \frac{1}{2\pi}\int_{\pi-\delta^{\frac{1}{3}}}^{\pi+\delta^{\frac{1}{3}}}\frac{1}{\delta\cdot h+\sqrt{\delta^2 \beta_*^2+\alpha_*^2 (p-p_*)^2}}\frac{L(p;\delta)}{1+L(p;\delta)}dp, \\
     &f_2(h;\delta):=\frac{1}{2\pi}\int_{\pi-\delta^{\frac{1}{3}}}^{\pi+\delta^{\frac{1}{3}}}\frac{1}{\delta\cdot h-\sqrt{\delta^2 \beta_*^2+\alpha_*^2 (p-p_*)^2}}\frac{1}{1+L(p;\delta)}dp, \\
     &\tilde{f}_1(h;\delta):=\frac{1}{2\pi}\int_{\pi-\delta^{\frac{1}{3}}}^{\pi+\delta^{\frac{1}{3}}}\frac{1}{\delta\cdot h+\sqrt{\delta^2 \beta_*^2+\alpha_*^2 (p-p_*)^2}}\frac{1}{1+L(p;\delta)}dp, \\
     &\tilde{f}_2(h;\delta):=
     \frac{1}{2\pi}\int_{\pi-\delta^{\frac{1}{3}}}^{\pi+\delta^{\frac{1}{3}}}\frac{1}{\delta\cdot h-\sqrt{\delta^2 \beta_*^2+\alpha_*^2 (p-p_*)^2}}\frac{L(p;\delta)}{1+L(p;\delta)}dp,
    \end{aligned}
\end{equation*}
where 
\[
\alpha_*=\frac{\theta_*}{\gamma_*}, \quad \beta_*=\frac{t_*}{\gamma_*}, \quad L(p;\delta):=\frac{\alpha_*^2(p-p_*)^2}{(\delta \beta_*+\sqrt{\delta^2 \beta_*^2+\alpha_*^2 (p-p_*)^2})^2}. 
\]

We have the following asymptotic expansions of the operators $\mathbb{G}_{\pm\delta}$.

\begin{proposition} \label{Asymptotic formula of Green operators}
There exists $\delta_0>0$ such that for $0<\delta<\delta_0$, the operator $\mathbb{G}_{\delta}$ defined in \eqref{eq89} admits the following asymptotic expansion
\begin{equation} \label{eq90}
    \begin{aligned}
        \mathbb{G}_{\delta}(\lambda_*+\delta\cdot h)=\mathbb{T}_{\delta}(\lambda_*+\delta\cdot h)+\Big(f_1(h;\delta)+f_2(h;\delta)\Big)\mathbb{P}+\Big(\mathbb{R}_1(h;\delta)+\mathbb{R}_2(h;\delta)\Big),
    \end{aligned}
\end{equation}
where all functions and operators in \eqref{eq90} are analytic with respect to $h$ in $\tilde{J}$. Moreover, the following hold uniformly for $h\in \tilde{J}$:
\begin{equation} \label{eq92}
        \lim_{\delta\to 0}\|\mathbb{R}_1(h;\delta)\|_{\mathcal{B}(\tilde{H}^{-\frac{1}{2}}(\Gamma),H^{\frac{1}{2}}(\Gamma))}=0,\quad
        \lim_{\delta\to 0}\|\mathbb{R}_2(h;\delta)\|_{\mathcal{B}(\tilde{H}^{-\frac{1}{2}}(\Gamma),H^{\frac{1}{2}}(\Gamma))}=0.
\end{equation}
Similarly,
\begin{equation} \label{eq94}
    \begin{aligned}
        \mathbb{G}_{-\delta}(\lambda_*+\delta\cdot h)=\mathbb{T}_{-\delta}(\lambda_*+\delta\cdot h)+\Big(\tilde{f}_1(h;\delta)+\tilde{f}_2(h;\delta)\Big)\mathbb{P}
        +\Big(\tilde{\mathbb{R}}_1(h;\delta)+\tilde{\mathbb{R}}_2(h;\delta)\Big),
    \end{aligned}
\end{equation}
with
\begin{equation}\label{eq92-1}
    \lim_{\delta\to 0}\|\tilde{\mathbb{R}}_1(h;\delta)\|_{\mathcal{B}(\tilde{H}^{-\frac{1}{2}}(\Gamma),H^{\frac{1}{2}}(\Gamma))}=0,\quad
    \lim_{\delta\to 0}\|\tilde{\mathbb{R}}_2(h;\delta)\|_{\mathcal{B}(\tilde{H}^{-\frac{1}{2}}(\Gamma),H^{\frac{1}{2}}(\Gamma))}=0.
\end{equation}
\end{proposition}
\begin{proof}
    See Appendix D.
\end{proof}

\medskip
We next investigate the limiting behavior of the operator $\mathbb{T}_{\pm\delta}$ as $\delta\to 0$. To this end, we define the following integral operator 
$\mathbb{T}_{0}\in \mathcal{B}(\tilde{H}^{-\frac{1}{2}}(\Gamma),H^{\frac{1}{2}}(\Gamma))$: 
\begin{equation*} \label{eq95}
        \mathbb{T}_{0}=\frac{1}{2\pi}\int_{0}^{2\pi}\sum_{n\geq 3}\frac{\langle \cdot,\overline{u_{n}(y;p)}\rangle}{\lambda_*-\lambda_{n}(p)}u_{n}(x;p)dp 
        + \frac{1}{2\pi}p.v.\int_{0}^{2\pi}\sum_{n=1,2}\frac{\langle \cdot,\overline{v_{n}(y;p)}\rangle}{\lambda_*-\mu_{n}(p)}v_{n}(x;p)dp.    
\end{equation*}
Using (\ref{eq-G}), the kernel function of the operator $\mathbb{T}_0$ (denoted by $T_0(x,y;\lambda_*)$) is related to the Green's function $G(x,y;\lambda_*)$ by
\begin{equation} \label{eq97}
     G(x,y;\lambda_*)=
   T_0(x,y;\lambda_*) -\frac{i}{2}\frac{v_1(x,\pi)\overline{v_1(y,\pi)}}{\alpha_*}-\frac{i}{2}\frac{v_2(x,\pi)\overline{v_2(y,\pi)}}{\alpha_*}. 
\end{equation}
%One is referred to Remark \ref{link between two types of modes} for the relationship between $v_i$ ($i=1,2$) in \eqref{eq97} and $\phi_i$ ($i=1,2$) in Proposition \ref{Asymptotic formula of Green operators}. 
The following two propositions give the properties of the operator $\mathbb{T}_0$.
\begin{proposition} \label{lem-t0}
The kernel of $\mathbb{T}_0$ is given by
\begin{equation} \label{eq-kernel-T_0}
\text{Ker}\,\mathbb{T}_0=\text{span}\left\{\frac{\partial v_1}{\partial x_1}\big|_{\Gamma}\right\}.
\end{equation}
\end{proposition}

\begin{proposition} \label{lem-t1}
$\mathbb{T}_0\in\mathcal{B}(\tilde{H}^{-\frac{1}{2}}(\Gamma),H^{\frac{1}{2}}(\Gamma))$ is a Fredholm operator of index zero. 
\end{proposition}
\begin{proof}
    See Appendix E.
\end{proof}
We are now ready to investigate the limit of the operator $\tilde{\mathbb{G}}_\delta(\lambda_*+\delta\cdot h)$. 
\begin{proposition}\label{limit of the sum operator}
The following holds uniformly for $h\in \tilde{J}$:
\begin{equation*} \label{eq98}
    \lim_{\delta\to 0}\Bigg\|\tilde{\mathbb{G}}_\delta(\lambda_*+\delta\cdot h)-\Big(2\mathbb{T}_0+\beta(h)\mathbb{P}\Big)\Bigg\|_{\mathcal{B}(\tilde{H}^{-\frac{1}{2}}(\Gamma),H^{\frac{1}{2}}(\Gamma))}=0, 
\end{equation*}
where
\begin{equation*} \label{eq96}
    \beta(h)=-\frac{1}{\beta_*\alpha_*}\frac{h}{\sqrt{1-(\frac{h}{\beta_*})^2}}.
\end{equation*}
\end{proposition}
\begin{proof}
See Appendix F.
\end{proof} 
Denote the limiting operator above by $\tilde{\mathbb{G}}_0(h):=2\mathbb{T}_0+\beta(h)\mathbb{P}$. We have
\begin{proposition}\label{property of the limiting operator}
Let $h\in\tilde{J}$. Then $\tilde{\mathbb{G}}_0( h)$ is a Fredholm operator with index zero, analytic for $h\in\tilde{J}$, and continuous for $h\in\partial\tilde{J}$. As a function of $h$, it attains a unique characteristic value $h=0$ in $\tilde{J}$, whose null multiplicity is one. Moreover, $\tilde{\mathbb{G}}_0( h)$ is invertible for any $h\in\tilde{J}$ with $h\neq 0$.
\end{proposition}

\begin{proof}[Proof]
First, by the boundedness of $\mathbb{T}_0$ and $\mathbb{P}$, $\tilde{\mathbb{G}}_0(h)\in\mathcal{B}(\tilde{H}^{-\frac{1}{2}}(\Gamma),H^{\frac{1}{2}}(\Gamma))$. 
By selecting the principal branch of the square root on $\mathbb{C}\backslash(-\infty,-\beta_*)\bigcup(\beta_*,+\infty)$, it is clear that $\beta(h)$ is analytic on $\tilde{J}$, which implies the analyticity of $\tilde{\mathbb{G}}_0( h)$ as seen by its definition; then its continuity for $h\in\partial \tilde{J}$ is also clear. By Proposition \ref{lem-t0}, we deduce that $\mathbb{T}_0\in\mathcal{B}(\tilde{H}^{-\frac{1}{2}}(\Gamma),H^{\frac{1}{2}}(\Gamma))$ is a Fredholm operator with zero index and $\text{Ker}\,\mathbb{T}_0=\text{span}\left\{\frac{\partial v_1}{\partial x_1}\big|_{\Gamma}\right\}$.
Thus $\tilde{\mathbb{G}}_0( h)$, which is obtained by perturbing $\mathbb{T}_0$ with the rank-1 projection operator $\mathbb{P}$, is a Fredholm operator with zero index.

Now we show that, as a function of $h$, $\tilde{\mathbb{G}}_0( h)=2\mathbb{T}_0+\beta(h)\mathbb{P}$ attains a unique characteristic value $h=0$ in $\Tilde{J}$. In another word, the following equation attains a nontrivial solution $w\in \tilde{H}^{-\frac{1}{2}}(\Gamma)$ if and only if $h=0$:
\begin{equation} \label{eq112}
\tilde{\mathbb{G}}_0(h)w=0.
\end{equation}
To this end, we apply $\frac{\partial\overline{v_1}}{\partial x_1}\Big|_{\Gamma}$ to both hand sides of the equation above to get
\begin{equation*}
\begin{aligned}
0&=2\left\langle \frac{\partial\overline{v_1}}{\partial x_1}\Big|_{\Gamma}, \mathbb{T}_0 w\right\rangle
+\langle w,\overline{\phi_2}\rangle\beta(h)\left\langle \frac{\partial\overline{v_1}}{\partial x_1}\Big|_{\Gamma},\phi_2\right\rangle \\
&=2\overline{\left\langle \overline{w}, \mathbb{T}_0 \Big(\frac{\partial v_1}{\partial x_1}\Big)\Big|_{\Gamma}\right\rangle}
+\langle w,\overline{\phi_2}\rangle\beta(h)\left\langle \frac{\partial\overline{v_1}}{\partial x_1}\Big|_{\Gamma},\phi_2\right\rangle \\
&=\langle w,\overline{\phi_2}\rangle\beta(h)\left\langle \frac{\partial\overline{v_1}}{\partial x_1}\Big|_{\Gamma},\phi_2\right\rangle .
\end{aligned}
\end{equation*}
where the second identity follows from the definition of $\mathbb{T}_0$. Recall the fact that $\phi_2|_{\Gamma}$ is proportional to $v_1|_{\Gamma}$ (See Remark \ref{link between two types of modes}). As a result,
\begin{equation} \label{eq113}
    \langle w,\overline{\phi_2}\rangle\beta(h)\left\langle \frac{\partial\overline{v_1}}{\partial x_1}\Big|_{\Gamma},v_1\right\rangle =0.
\end{equation}
By Lemma \ref{orthogonality of the propagating mode}, $\Big\langle \frac{\partial\overline{v_1}}{\partial x_1}\Big|_{\Gamma},v_1\Big\rangle\neq 0$; thus \eqref{eq113} implies $\langle w,\overline{\phi_2}\rangle\beta(h)=0$, or equivalently
$$
\langle w,\overline{v_1}\rangle\beta(h)=0.
$$
If $\beta(h)\neq 0$, we have $\langle w,\overline{v_1}\rangle=0$ and $\mathbb{T}_0 w=0$, which imply that $w=0$ by Proposition \ref{lem-t0}. But this contradicts the assumption that $w\neq 0$. Hence, we deduce that the characteristic value problem \eqref{eq112} attains solutions only when $\beta(h)=0$. Solving $\beta(h)=0$, we obtain a simple root $h=0$. Then Proposition \ref{lem-t0} implies that $h=0$ is a characteristic value of multiplicity one with its associated eigenvector $\frac{\partial v_1}{\partial x_1}\Big|_{\Gamma}$. 

%To show the null multiplicity of the characteristic value $h=0$ for the operator $\tilde{\mathbb{G}}_0(h)=2\mathbb{T}_0+\beta(h)\mathbb{P}$ is one, it is sufficient to note that the space of root functions of $2\mathbb{T}_0+\beta(0)\mathbb{P}$ is exactly $\text{Ker}\,\mathbb{T}_0$, which is a one-dimensional space by Proposition \ref{lem-t0}.

Finally, we prove the invertibility of $\tilde{\mathbb{G}}_0( h)$ for $h\neq 0$. Indeed, since $\tilde{\mathbb{G}}_0( h)$ has a unique characteristic value $h=0$, it is injective for $h\neq 0$; thus, $\tilde{\mathbb{G}}_0( h)$ is invertible by noting that it is a Fredholm operator of zero index.
\end{proof}

%First, we claim that the operator $\tilde{G}_\delta(\lambda)$ defined in (4.24) satisfies that:
\begin{proposition}\label{finitely meormophic of the sum operator}
Let $h\in\tilde{J}$. Then $\tilde{\mathbb{G}}_\delta(\lambda_*+\delta\cdot h)$ is a Fredholm operator, and it is analytic for $h\in\tilde{J}$.
%and continuous for $h\in\partial\tilde{J}$.
\end{proposition}

\begin{proof}[Proof]
The analyticity of $\tilde{\mathbb{G}}_\delta(\lambda_*+\delta\cdot h)$ follows from the analyticity of the Green's function $G_{\delta}(x, y; \lambda)$ in $\lambda$, which holds for any $\lambda=\lambda_*+\delta\cdot h$ with $h\in\tilde{J}$.
Using Proposition \ref{limit of the sum operator} and \ref{property of the limiting operator}, and the fact that Fredholm index is stable under small perturbation (see Proposition \ref{stability of index under norm perturbation}), we conclude that 
$\tilde{G}_\delta(\lambda_*+\delta\cdot h)$ is a Fredholm operator with zero index for $h\in\tilde{J}$ when $\delta>0$ is sufficiently small. 
\end{proof}

\subsection{Proof of Theorem \ref{main result}}
By Propositions \ref{limit of the sum operator} and \ref{property of the limiting operator}, for sufficiently small $\delta>0$ and $h\in\partial \tilde{J}$, we have
\begin{equation*} \label{eq114}
\Bigg\|\tilde{\mathbb{G}}^{-1}_0(h)\Big(\tilde{\mathbb{G}}_\delta(\lambda_*+\delta\cdot h)-\tilde{\mathbb{G}}_0(h)\Big) \Bigg\|_{\mathcal{B}(\tilde{H}^{-\frac{1}{2}}(\Gamma))}<1,
\end{equation*}
where we have used the fact that $\tilde{\mathbb{G}}^{-1}_0(h)$ is uniformly bounded in norm for $h\in\partial\tilde{J}$ (it is a direct consequence of the continuity and invertibility of $\tilde{\mathbb{G}}_0(h)$ for $h\in\partial\tilde{J}$). Now, with Proposition \ref{property of the limiting operator} and \ref{finitely meormophic of the sum operator}, an application of Theorem \ref{generalized rouche theorem} shows that, for sufficiently small $\delta>0$, \eqref{eq109} attains a unique characteristic value $\lambda^\star :=\lambda_*+h^\star$ with $h^\star\in \tilde{J}$.
Let $\phi^\star$ be the associated eigenvector, then we construct a solution to \eqref{eq7} by setting
\begin{equation*} \label{eq115}
    \lambda=\lambda^\star,\quad
    u^\star=u((x_1,x_2);\lambda^\star)=
    \left\{
    \begin{aligned}
        &\int_{\Gamma}G_\delta ((x_1,x_2),y;\lambda)\phi^\star(y;\lambda)d\sigma(y),\quad x_1\geq 0, \\
        &-\int_{\Gamma}G_{-\delta} ((x_1,x_2),y;\lambda)\phi^\star(y;\lambda)d\sigma(y),\quad x_1< 0.
    \end{aligned}
    \right.
\end{equation*}
Meanwhile, we claim that $\lambda^\star \in\mathbf{R}$. Otherwise,
$(\overline{\lambda^\star},\overline{u^\star})$ gives another solution to the eigenvalue problem \eqref{eq7} and thus, $\overline{h}=\frac{\overline{\lambda^\star}-\lambda_*}{\delta}$ is another characteristic value of $\tilde{\mathbb{G}}_\delta(\lambda_*+\delta\cdot h)$ for $h\in \tilde{J}$, which is different from $h^\star$. But this contradicts the uniqueness of the characteristic value of $\tilde{\mathbb{G}}_\delta(\lambda_*+\delta\cdot h)$ for $h\in \tilde{J}$.

Finally, the assertion that $u^\star$ decays exponentially away from $\Gamma$ follows from the radiation condition of the Green's function $G_{\pm\delta}$ (See Section 2.3). This completes the proof of the theorem.

%\newpage
\appendix
\setcounter{secnumdepth}{0}
\section{Appendix}

\subsection{Appendix A: Proof of Corollary \ref{even odd mode and root function}}
\setcounter{equation}{0}
\setcounter{subsection}{0}
\renewcommand{\theequation}{A.\arabic{equation}}
\renewcommand{\thesubsection}{A.\arabic{subsection}}

\begin{proof}[Proof of Corollary \ref{even odd mode and root function}]
%\begin{equation*}
%    \left\{
 %   \begin{aligned}
%        &(\mathcal{L}-\lambda_*)v_1(x)=0,\quad x \in \Omega,\\
%        &v_1(x)=0,\quad x \in \partial D_n ,\\
 %       &\frac{\partial}{\partial x_2} %%v_1(x)=0,\quad x\in \Gamma_{-}\bigcup\Gamma_{+} ,\\
 %       &v_1(x_1+\frac{1}{2},x_2)=iv(x_1,x_2).
 %   \end{aligned}
 %   \right.
%\end{equation*}
Let $v_1(x)=\text{Re }v_1(x)+i\,\text{Im }v_1(x)$. We first claim that $\text{Re }v_1(x)$ and $\text{Im }v_1(x)$ are linearly independent. Otherwise, there exists $c\in\mathbf{R}$ such that $\text{Re }v_1(x)=c \, \text{Im }v_1(x)$. Then the quasi-periodic boundary condition $v_1(x+\frac{1}{2}\bm{e_1})=iv_1$ gives
\begin{equation*}
    \begin{aligned}
    &v_1(x+\frac{1}{2}\bm{e_1})=\text{Re }v_1(x+\frac{1}{2}\bm{e_1})+i\text{Im }v_1(x+\frac{1}{2}\bm{e_1})=(c+i)\text{Im }v_1(x+\frac{1}{2}\bm{e_1}), \\
    &v_1(x+\frac{1}{2}\bm{e_1})=iv_1(x)=i(\text{Re }v_1(x)+i\text{Im }v_1(x))=i(c+i)\text{Im }v_1(x),
    \end{aligned}    
\end{equation*}
which implies that $i\,\text{Im }v_1(x)=\text{Im }v_1(x+\frac{1}{2}\bm{e_1})$, and thus $\text{Im }v_1(x)=\text{Re }v_1(x)=0$. This contradiction proves the claim. 

We next show that the desired even/odd eigenmodes at the Dirac point can be constructed by using linear combinations of $\text{Re }v_1$ and $\text{Im }v_1$. It is clear that both $\text{Re }v_1$ and $\text{Im }v_1$ are eigenmodes at the Dirac point defined in Proposition \ref{Existence of Dirac points of the period-1 structure}. Note that both $\overline{v_1(x_1,x_2)}$ and $v_1(-x_1,x_2)$ are Bloch eigenmodes with the same quasi-periodic boundary condition in the period-$\frac{1}{2}$ structure. Since the dimension of the eigenspace is one by Assumption \ref{assump0}, there exists $\theta\in [0,2\pi)$ such that $\overline{v_1(x_1, x_2)}=e^{i\theta}v_1(-x_1, x_2)$. Set 
$R:=\begin{pmatrix}
\cos\theta & -\sin\theta \\ -\sin\theta & -\cos\theta
\end{pmatrix}$. Then we have
\begin{equation*}
    \begin{pmatrix}
    \text{Re }v_1(x_1,x_2) \\ \text{Im }v_1(x_1,x_2)
    \end{pmatrix}
    =R
    \begin{pmatrix}
     \text{Re }v_1(-x_1,x_2) \\ \text{Im }v_1(-x_1,x_2)
    \end{pmatrix}.
\end{equation*}
The real matrix $R$ has two eigenvalues, $-1$ and $1$. Let $(a_1,b_1)^T,(a_2,b_2)^T\in \mathbf{R}^2$ be the corresponding real eigenvectors, i.e.,
\begin{equation*}
    R(a_1,b_1)^T=-(a_1,b_1)^T,\quad R(a_2,b_2)^T=(a_2,b_2)^T. 
\end{equation*}
We define
\begin{equation*}
    \phi_1:=a_1 \text{Re }v_1 +b_1\text{Im }v_1,\quad \phi_2:=a_2 \text{Re }v_1 +b_2\text{Im }v_1.
\end{equation*}
Then one can check that $\phi_1(-x_1,x_2)=-\phi_1(x_1,x_2)$, $\phi_2(-x_1,x_2)=\phi_2(x_1,x_2)$. This gives the desired construction. 

Finally, note that $\phi_1$ is real-valued. Thus, the quasi-periodicity of $\phi_1$ and reflection-symmetry allows us to set the corresponding boundary potential $\bm{\varphi_1}$ as
\begin{equation*} \label{eq116}
\bm{\varphi_1}=(\varphi_{ref},\varphi)^T,
\end{equation*}
where $\varphi\in H^{-\frac{1}{2}}(\partial D)$ is real-valued. Taking the reflection image of $\phi_1$ with respect to the straight line $x_1=\frac{1}{4}$, we obtain an even mode $\phi_2$ with the boundary potential
\begin{equation*} \label{eq117}
\bm{\varphi_2}=(\varphi,-\varphi_{ref})^T.
\end{equation*}
Then the proof is completed.
\end{proof}
\subsection{Appendix B: Proof of Proposition \ref{non-degeneracy 1}}
\setcounter{equation}{0}
\setcounter{subsection}{0}
\renewcommand{\theequation}{B.\arabic{equation}}
\renewcommand{\thesubsection}{B.\arabic{subsection}}
In Appendix B-C, the bracket $\langle\cdot,\cdot\rangle$ denotes the dual pair between $H^{-\frac{1}{2}}(\partial D)\times H^{-\frac{1}{2}}(\partial D)$ and $H^{\frac{1}{2}}(\partial D)\times H^{\frac{1}{2}}(\partial D)$.

% As in Appendix A, the bracket $\langle\cdot,\cdot\rangle$ denotes the duality pair between $(H^{-\frac{1}{2}}(\partial D))^2$ and $(H^{\frac{1}{2}}(\partial D))^2$ if there's no further notification.
\begin{proof}
Step 1. First we show that for $(p,\lambda)$ near $(\pi,\lambda_*)$,
\begin{equation} \label{eq119}
    \langle \bm{\varphi},T(p,\lambda)\bm{\varphi}\rangle=\overline{\langle \bm{\varphi},T(p,\lambda)\bm{\varphi}\rangle}
    \text{  for real-valued $\bm{\varphi}$}.
\end{equation}
Note that the quasi-periodic Green's function for the empty waveguide defined in \eqref{eq26} has the following properties: 
\begin{equation} \label{eq120}
    \begin{aligned}
    &G^e (x,y;\pi+h_p,\lambda_*)=\overline{G^e(x,y;\pi-h_p,\lambda_*)}, \\
    &G^e (x,y;\pi+h_p,\lambda_*+h_\lambda)=\overline{G^e(y,x;\pi+h_p,\lambda_*+h_\lambda)},
    \end{aligned}
\end{equation}
for sufficiently small constants $h_p,h_\lambda$.
Using the polar coordinate for the boundary $\partial D$ and Corollary \ref{even odd mode and root function}, we have
\begin{equation*} \label{eq121_1}
    \langle \bm{\varphi},T(p,\lambda)\bm{\varphi}\rangle = V_1(p,\lambda)+V_2(p,\lambda)+V_3(p,\lambda)+V_4(p,\lambda),
\end{equation*}
where 
\begin{equation*} \label{eq121}
    \begin{aligned}
    V_1(p,\lambda)&=\int_{[0,2\pi]\times[0,2\pi]}G^e(\bm{r}(\theta_1),\bm{r}(\theta_2);p,\lambda)\cdot\varphi_1(\theta_2)\varphi_1(\theta_1)\cdot|\bm{r}'(\theta_2)||\bm{r}'(\theta_1)|d\theta_2d\theta_1, \\
    V_2(p,\lambda)&=\int_{[0,2\pi]\times[0,2\pi]}G^e(\bm{r}(\theta_1),\frac{1}{2}\bm{e}_1+\bm{r}(\theta_2);p,\lambda)\cdot\varphi_2(\theta_2)\varphi_1(\theta_1)\cdot|\bm{r}'(\theta_2)||\bm{r}'(\theta_1)|d\theta_2d\theta_1 ,\\
    V_3(p,\lambda)&=\int_{[0,2\pi]\times[0,2\pi]}G^e(\frac{1}{2}\bm{e}_1+\bm{r}(\theta_1),\bm{r}(\theta_2);p,\lambda)\cdot\varphi_1(\theta_2)\varphi_2(\theta_1)\cdot|\bm{r}'(\theta_2)||\bm{r}'(\theta_1)|d\theta_2d\theta_1 ,\\
    V_4(p,\lambda)&=  \int_{[0,2\pi]\times[0,2\pi]}G^e(\bm{r}(\theta_1),\bm{r}(\theta_2);p,\lambda)\cdot\varphi_2(\theta_2)\varphi_2(\theta_1)\cdot|\bm{r}'(\theta_2)||\bm{r}'(\theta_1)|d\theta_2d\theta_1 .
    \end{aligned}
\end{equation*}
We claim that $V_1, V_4$, and $V_2+V_3$ are all real numbers. For $V_1$, a change of variable yields
\begin{equation*} \label{eq122}
    \begin{aligned}
2V_1(p,\lambda)=&\int_{[0,2\pi]\times[0,2\pi]}G^e(\bm{r}(\theta_1),\bm{r}(\theta_2);p,\lambda)\cdot\varphi_1(\theta_2)\varphi_1(\theta_1)\cdot|\bm{r}'(\theta_2)||\bm{r}'(\theta_1)|d\theta_2d\theta_1 \\
    &+\int_{[0,2\pi]\times[0,2\pi]}G^e(\bm{r}(\theta_2),\bm{r}(\theta_1);p,\lambda)\cdot\varphi_1(\theta_1)\varphi_1(\theta_2)\cdot|\bm{r}'(\theta_1)||\bm{r}'(\theta_2)|d\theta_1d\theta_2 \\
    =&\int_{[0,2\pi]\times[0,2\pi]}\Big(G^e(\bm{r}(\theta_1),\bm{r}(\theta_2))+G^e(\bm{r}(\theta_2),\bm{r}(\theta_1))\Big)\varphi_1(\theta_2)\varphi_1(\theta_1)|\bm{r}'(\theta_2)||\bm{r}'(\theta_1)|d\theta_2d\theta_1 .
    \end{aligned}
\end{equation*}
Thus \eqref{eq120} shows that the integrand of $V_1(p,\lambda)$ is real, which implies that $V_1(p,\lambda)\in \mathbf{R}$. Similarly, it can be proved that $V_4(p,\lambda)\in\mathbf{R}$. Besides,
\begin{equation*} \label{eq123}
\begin{aligned}
   &V_2(p,\lambda)+V_3(p,\lambda)=
    \int_{[0,2\pi]\times[0,2\pi]}G^e(\bm{r}(\theta_1),\frac{1}{2}\bm{e}_1+\bm{r}(\theta_2);p,\lambda)\cdot\varphi_2(\theta_2)\varphi_1(\theta_1)\cdot|\bm{r}'(\theta_2)||\bm{r}'(\theta_1)|d\theta_2d\theta_1 \\
&\qquad\qquad\qquad\qquad\quad+\int_{[0,2\pi]\times[0,2\pi]}G^e(\frac{1}{2}\bm{e}_1+\bm{r}(\theta_2),\bm{r}(\theta_1);p,\lambda)\cdot\varphi_1(\theta_1)\varphi_2(\theta_2)\cdot|\bm{r}'(\theta_1)||\bm{r}'(\theta_2)|d\theta_1d\theta_2 \\
&=\int_{[0,2\pi]\times[0,2\pi]}\Big(G^e(\bm{r}(\theta_1),\frac{1}{2}\bm{e}_1+\bm{r}(\theta_2))+G^e(\frac{1}{2}\bm{e}_1+\bm{r}(\theta_2),\bm{r}(\theta_1))\Big)\varphi_1(\theta_1)\varphi_2(\theta_2)|\bm{r}'(\theta_1)||\bm{r}'(\theta_2)|d\theta_1d\theta_2, 
    \end{aligned}
\end{equation*}
which is also real by invoking \eqref{eq120}. Thus, $\langle \bm{\varphi},T(p,\lambda)\bm{\varphi}\rangle$ is real for any real-valued $\bm{\varphi}$.

\medskip

Step 2. We prove that $\langle \bm{\varphi_1},T_p\bm{\varphi_1}\rangle=0$, and in a similar way $\langle\bm{\varphi_2},T_p\bm{\varphi_2}\rangle=0$. 
From  \eqref{eq120}, we have
\begin{equation*}
    \langle \bm{\varphi_1},T(p_*+h,\lambda_*)\bm{\varphi_1}\rangle
    =\langle \bm{\varphi_1},\overline{T(p_*-h,\lambda_*)}\bm{\varphi_1}\rangle .
\end{equation*}
Moreover, since $\bm{\varphi}_1$ is real-valued, \eqref{eq119} gives that
$$
\langle \bm{\varphi_1},\overline{T(p_*-h,\lambda_*)}\bm{\varphi_1}\rangle
=\overline{\langle \bm{\varphi_1},T(p_*-h,\lambda_*)\bm{\varphi_1}\rangle}
=\langle \bm{\varphi_1},T(p_*-h,\lambda_*)\bm{\varphi_1}\rangle .
$$
In conclusion,
$$
\langle \bm{\varphi_1},T(p_*+h,\lambda_*)\bm{\varphi_1}\rangle
=\langle \bm{\varphi_1},T(p_*-h,\lambda_*)\bm{\varphi_1}\rangle .
$$
Then by the smoothness of $\langle \bm{\varphi_1},T(p,\lambda_*)\bm{\varphi_1}\rangle$ in $p$ near $p=\pi$, we obtain $\langle \bm{\varphi_1},T_p\bm{\varphi_1}\rangle=0$.

\medskip

Step 3. We show that $\langle \bm{\varphi_1},T_p\bm{\varphi_2}\rangle=-\langle \bm{\varphi_2},T_p\bm{\varphi_1}\rangle$ is a pure imaginary number. First,
\begin{equation} \label{eq125}
    \begin{aligned}
    \langle \bm{\varphi_2},T(p,\lambda)\bm{\varphi_1}\rangle&=\int_{[0,2\pi]\times[0,2\pi]}G^e(\bm{r}(\theta_1),\bm{r}(\theta_2);p,\lambda)\cdot\varphi_{ref}(\theta_2)\varphi(\theta_1)\cdot|\bm{r}'(\theta_2)||\bm{r}'(\theta_1)|d\theta_2d\theta_1 \\
    &\quad+\int_{[0,2\pi]\times[0,2\pi]}G^e(\bm{r}(\theta_1),\frac{1}{2}\bm{e}_1+\bm{r}(\theta_2);p,\lambda)\cdot\varphi(\theta_2)\varphi(\theta_1)\cdot|\bm{r}'(\theta_2)||\bm{r}'(\theta_1)|d\theta_2d\theta_1 \\
    &\quad-\int_{[0,2\pi]\times[0,2\pi]}G^e(\frac{1}{2}\bm{e}_1+\bm{r}(\theta_1),\bm{r}(\theta_2);p,\lambda)\cdot\varphi_{ref}(\theta_2)\varphi_{ref}(\theta_1)\cdot|\bm{r}'(\theta_2)||\bm{r}'(\theta_1)|d\theta_2d\theta_1 \\
    &\quad-\int_{[0,2\pi]\times[0,2\pi]}G^e(\bm{r}(\theta_1),\bm{r}(\theta_2);p,\lambda)\cdot\varphi(\theta_2)\varphi_{ref}(\theta_1)\cdot|\bm{r}'(\theta_2)||\bm{r}'(\theta_1)|d\theta_2d\theta_1 \\
    &=:U_1(p,\lambda)+U_2(p,\lambda)+U_3(p,\lambda)+U_4(p,\lambda).
    \end{aligned}
\end{equation}

Note that the Green's function $G^e$ can be written in the following form
\begin{equation} \label{eq126}
G^e(x,y;p,\lambda)=\sum_{n\in\mathbb{Z}}f(\lambda,|p_n|,x_2,y_2)e^{ip_n(x_1-y_1)},
\end{equation}
where $p_n:=p+2n\pi$ ($n\in\mathbb{Z}$) and $f:\mathbf{R}^4\to\mathbf{R}$ is real-valued when $\lambda$ near $\lambda_*$.
It follows that
\begin{equation*}
    \begin{aligned}
    &G^e(\bm{r}(\theta_1),\bm{r}(\pi-\theta_2);p,\lambda)-
    G^e(\bm{r}(\pi-\theta_1),\bm{r}(\theta_2);p,\lambda) \\
    &=\sum_{n\in\mathbb{Z}}f\Big(\lambda,|p_n|,r(\theta_1)sin(\theta_1),r(\pi-\theta_2)sin(\pi-\theta_2)\Big)e^{ip_n\Big(r(\theta_1)cos(\theta_1)-r(\pi-\theta_2)cos(\pi-\theta_2)\Big)}\\
    &\quad-\sum_{n\in\mathbb{Z}}f\Big(\lambda,|p_n|,r(\pi-\theta_1)sin(\pi-\theta_1),r(\theta_2)sin(\theta_2)\Big)e^{ip_n\Big(r(\pi-\theta_1)cos(\pi-\theta_1)-r(\theta_2)cos(\theta_2)\Big)} \\
    &=2i\sum_{n\in\mathbb{Z}}f\Big(\lambda,|p_n|,r(\theta_1)sin\theta_1,r(\theta_2)sin\theta_2\Big)sin\Big[p_n(r(\theta_1)cos\theta_1+r(\theta_2)cos\theta_2)\Big],
    \end{aligned}
\end{equation*}
where $r(\pi-\theta)=r(\theta)$ is used in the last equality above. Therefore,
\begin{equation} \label{eq127}
    \begin{aligned}
    &U_1(p,\lambda)+U_4(p,\lambda) \\
    &=2i\sum_{n\in\mathbb{Z}}\int_{[0,2\pi]\times[0,2\pi]}
    f\Big(\lambda,|p_n|,r(\theta_1)sin\theta_1,r(\theta_2)sin\theta_2\Big)sin\Big[p_n(r(\theta_1)cos\theta_1+r(\theta_2)cos\theta_2)\Big] \\
    &\qquad\qquad\qquad\qquad\cdot\varphi(\theta_2)\varphi(\theta_1)\cdot|\bm{r}'(\theta_2)||\bm{r}'(\theta_1)|d\theta_2d\theta_1 
    \end{aligned}
\end{equation}
which is purely imaginary for $p$ near $p_*=\pi$ and $\lambda=\lambda_*$. Similarly, we can show that
$
U_2(p,\lambda)+U_3(p,\lambda)
$
is also purely imaginary. Therefore
\[
\langle \bm{\varphi_2},T_p\bm{\varphi_1}\rangle=
\frac{\partial}{\partial p}\Big(U_1+U_2+U_3+U_4  \Big)(\pi,\lambda_*)
\]
is also purely imaginary. We denote $\langle \bm{\varphi_2},T_p\bm{\varphi_1}\rangle =i\theta_* $
for some real number $\theta_*$. 

\medskip

Step 4. We prove $\langle \bm{\varphi_2},T_\lambda\bm{\varphi_1}\rangle=\langle \bm{\varphi_1},T_\lambda\bm{\varphi_2}\rangle=0$. Observe that the $n$-th and the $(-n-1)$-th term in \eqref{eq127} cancel out at $p=\pi$. Therefore $U_1(\pi,\lambda)+U_4(\pi,\lambda)\equiv 0$ for  $\lambda$ near $\lambda_*$. Similarly, there holds $U_2(\pi,\lambda)+U_3(\pi,\lambda)\equiv 0$. It follows that $\langle \bm{\varphi_2},T(\pi,\lambda)\bm{\varphi_1}\rangle\equiv 0$ for $\lambda$ near $\lambda_*$, which gives $\langle \bm{\varphi_2},T_\lambda\bm{\varphi_1}\rangle=0$. The equality $\langle \bm{\varphi_1},T_\lambda\bm{\varphi_2}\rangle=0$ can be proved similarly.

\medskip

Step 5. Finally, the equality $\langle \bm{\varphi_1}, T_\lambda\bm{\varphi_1} \rangle=\langle \bm{\varphi_2}, T_\lambda \bm{\varphi_2} \rangle=\gamma_*$ ($i=1,2$) for some real number $\gamma_*$ can be proved in the same way as $\langle \bm{\varphi_2}, T_p \bm{\varphi_1} \rangle$. 
\end{proof}

\subsection{Appendix C: Proof of Proposition \ref{non-degeneracy 2}}
\setcounter{equation}{0}
\setcounter{subsection}{0}
\renewcommand{\theequation}{C.\arabic{equation}}
\renewcommand{\thesubsection}{C.\arabic{subsection}}

\begin{proof}[Proof]
Notice that the perturbation of the periodic media is introduced by shifting the obstacles. Hence we only need to consider the off-diagonal terms  $\langle \bm{\varphi_j},T_\delta\bm{\varphi_i} \rangle$ for $i\neq j$. In particular, when $j=1$ and $i=2$, 
\begin{equation} \label{eq130}
    \begin{aligned}
    &\Big\langle
    \bm{\varphi_1},
    \begin{pmatrix}
    0 & T_{12,\delta}(p_*,\lambda_*) \\
    T_{21,\delta}(p_*,\lambda_*) & 0
    \end{pmatrix}
    \bm{\varphi_2} \Big\rangle\\
    &=-\int_{[0,2\pi]\times[0,2\pi]}G^e(\bm{r}(\theta_1),\frac{1+2\delta}{2}\bm{e}_1+\bm{r}(\theta_2);p_*,\lambda_*)\cdot\varphi_{ref}(\theta_2)\varphi_{ref}(\theta_1)\cdot|\bm{r}'(\theta_2)||\bm{r}'(\theta_1)|d\theta_2d\theta_1 \\
    &\quad+\int_{[0,2\pi]\times[0,2\pi]}G^e(\frac{1+2\delta}{2}\bm{e}_1+\bm{r}(\theta_1),\bm{r}(\theta_2);p_*,\lambda_*)\cdot\varphi(\theta_2)\varphi(\theta_1)\cdot|\bm{r}'(\theta_2)||\bm{r}'(\theta_1)|d\theta_2d\theta_1 \\
    &=\int_{[0,2\pi]\times[0,2\pi]}
    \Big[
    G^e(\frac{1+2\delta}{2}\bm{e}_1+\bm{r}(\theta_1),\bm{r}(\theta_2);p_*,\lambda_*)-
    G^e(\bm{r}(\pi-\theta_1),\frac{1+2\delta}{2}\bm{e}_1+\bm{r}(\pi-\theta_2);p_*,\lambda_*)
    \Big] \\
    &\qquad\qquad\qquad\qquad\cdot\varphi(\theta_2)\varphi(\theta_1)\cdot|\bm{r}'(\theta_2)||\bm{r}'(\theta_1)|d\theta_2d\theta_1 .
    \end{aligned}
\end{equation}
From \eqref{eq126},
\begin{equation*}
\begin{aligned}
&G^e(\frac{1+2\delta}{2}\bm{e}_1+\bm{r}(\theta_1),\bm{r}(\theta_2);p_*,\lambda_*)-
G^e(\bm{r}(\pi-\theta_1),\frac{1+2\delta}{2}\bm{e}_1+\bm{r}(\pi-\theta_2);p_*,\lambda_*) \\
&=\sum_{n\geq 0}f(\lambda_*,|p_n|,r(\theta_1)\sin\theta_1,r(\theta_2)\sin\theta_2)\cos \left((2n+1)\pi\Big(\frac{1+2\delta}{2}+r(\theta_1)\cos\theta_1-r(\theta_2)\cos\theta_2\Big)\right) \\
&\quad -
\sum_{n\geq 0}f(\lambda_*,|p_n|,r(\pi-\theta_1)\sin(\pi-\theta_1),r(\pi-\theta_2)\sin(\pi-\theta_2)) \\
&\quad\quad\quad\quad\quad\quad\quad\quad\cdot\cos \Bigg((2n+1)\pi\Big(r(\pi-\theta_1)\cos(\pi -\theta_1)-r(\pi-\theta_2)\cos(\pi-\theta_2)-\frac{1+2\delta}{2}\Big)\Bigg) \\
&=\sum_{n\geq 0}f(\lambda_*,|p_n|,r(\theta_1)\sin\theta_1,r(\theta_2)\sin\theta_2)\cos \left((2n+1)\pi(\frac{1+2\delta}{2}+r(\theta_1)\cos\theta_1-r(\theta_2)\cos\theta_2)\right) \\
&-\sum_{n\geq 0}f(\lambda_*,|p_n|,r(\theta_1)\sin\theta_1,r(\theta_2)\sin\theta_2)\cos \left(-(2n+1)\pi(\frac{1+2\delta}{2}+r(\theta_1)\cos\theta_1-r(\theta_2)\cos\theta_2)\right) \\
&=0.
\end{aligned}
\end{equation*}
Thus, for $\delta$ sufficiently small,
$$
\Big\langle
    \bm{\varphi_1},
    \begin{pmatrix}
    0 & T_{12,\delta}(p_*,\lambda_*) \\
    T_{21,\delta}(p_*,\lambda_*) & 0
    \end{pmatrix}
    \bm{\varphi_2} \Big\rangle \equiv 0.
$$
A similar calculation yields 
$$
\Big\langle\bm{\varphi_2},
\begin{pmatrix}
    0 & T_{12,\delta}(p_*,\lambda_*) \\
    T_{21,\delta}(p_*,\lambda_*) & 0
    \end{pmatrix}
\bm{\varphi_1} \Big\rangle\equiv 0.
$$
Hence $\langle \bm{\varphi_2},S \bm{\varphi_1}\rangle=\langle \bm{\varphi_1},S \bm{\varphi_2}\rangle=0$ by recalling that $S=\frac{\partial T_\delta}{\partial \delta}(p_*,\lambda_*;\delta)\Big|_{\delta=0}$. On the other hand, \eqref{eq126} also implies that
\begin{equation*} \label{eq131}
    \begin{aligned}
    &G^e(\bm{r}(\pi-\theta_1),\frac{1+2\delta}{2}\bm{e}_1+\bm{r}(\theta_2);p_*,\lambda_*)+G^e(\frac{1+2\delta}{2}\bm{e}_1+\bm{r}(\theta_1),\bm{r}(\pi-\theta_2);p_*,\lambda_*) \\
    &=4\sum_{n\geq 0}f\Big(\lambda_*,(2n+1)\pi,r(\theta_1)\sin\theta_1,r(\theta_2)\sin\theta_2\Big)
    \cos\left((2n+1)\pi\Big(
    \frac{1}{2}+\delta+r(\theta_1)\cos\theta_1+r(\theta_2)\cos\theta_2\Big)\right),
    \end{aligned}
\end{equation*}
\begin{equation*} \label{eq132}
    \begin{aligned}
    &G^e(\bm{r}(\theta_1),\frac{1+2\delta}{2}\bm{e}_1+\bm{r}(\pi-\theta_2);p_*,\lambda_*)+G^e(\frac{1+2\delta}{2}\bm{e}_1+\bm{r}(\pi-\theta_1),\bm{r}(\theta_2);p_*,\lambda_*) \\
    &=4\sum_{n\geq 0}f\Big(\lambda_*,(2n+1)\pi,r(\theta_1)\sin\theta_1,r(\theta_2)\sin\theta_2\Big)
    \cos\left((2n+1)\pi\Big(
    \frac{1}{2}+\delta-r(\theta_1)\cos\theta_1-r(\theta_2)\cos\theta_2\Big)\right).
    \end{aligned}
\end{equation*}
It follows that $\langle \bm{\varphi_1},T_\delta(p_*,\lambda_*)\bm{\varphi_1}\rangle =-\langle \bm{\varphi_2},T_\delta(p_*,\lambda_*)\bm{\varphi_2}\rangle\in\mathbf{R}$, by following the same lines as in \eqref{eq130}. Therefore, we conclude that
$\langle \bm{\varphi_1},S \bm{\varphi_1}\rangle
    =-\langle \bm{\varphi_2},S \bm{\varphi_2}\rangle\in\mathbf{R}$.
\end{proof}

\subsection{Appendix D: Proof of Proposition \ref{Asymptotic formula of Green operators}}
\setcounter{equation}{0}
\setcounter{subsection}{0}
\renewcommand{\theequation}{D.\arabic{equation}}
\renewcommand{\thesubsection}{D.\arabic{subsection}}
\begin{proof} [Proof]
We only prove \eqref{eq90} here. The proof of \eqref{eq94} is identical. The idea is to split the integral expression of the Green's function \eqref{eq88} into different parts and apply asymptotic expansion to each part. We start with the following decomposition:
$$
 G_\delta(x, y;\lambda) = T_\delta(x,y;\lambda)+U_1(x,y;\lambda,\delta)+U_2(x,y;\lambda,\delta),
$$
where
\begin{equation*} \label{eq135}
\begin{aligned}
     &T_\delta(x,y;\lambda)=\frac{1}{2\pi}\bigg(\int_{0}^{2\pi}\sum_{n\geq 3}\frac{u_{n,\delta}(x;p)\overline{u_{n,\delta}(y;p)}}{\lambda-\lambda_{n,\delta}(p)}dp \\
     &\qquad\qquad+\int_{[0,\pi-\delta^{\frac{1}{3}}]\bigcup[\pi+\delta^{\frac{1}{3}},2\pi]}\frac{u_{1,\delta}(x;p)\overline{u_{1,\delta}(y;p)}}{\lambda-\lambda_{1,\delta}(p)}dp+\int_{[0,\pi-\delta^{\frac{1}{3}}]\bigcup[\pi+\delta^{\frac{1}{3}},2\pi]}\frac{u_{2,\delta}(x;p)\overline{u_{2,\delta}(y;p)}}{\lambda-\lambda_{2,\delta}(p)}dp\bigg), \\
     &U_1(x,y;\lambda,\delta)=
        \frac{1}{2\pi}\int_{[\pi -\delta^{\frac{1}{3}},\pi+ \delta^{\frac{1}{3}}]}\frac{u_{1,\delta}(x;p)\overline{u_{1,\delta}(y;p)}}{\lambda-\lambda_{1,\delta}(p)}dp, \\
     &U_2(x,y;\lambda,\delta)=\frac{1}{2\pi}\int_{[\pi -\delta^{\frac{1}{3}},\pi +\delta^{\frac{1}{3}}]}\frac{u_{2,\delta}(x;p)\overline{u_{2,\delta}(y;p)}}{\lambda-\lambda_{2,\delta}(p)}dp. 
    \end{aligned}
\end{equation*}
Note that $T_\delta(x,y;\lambda)$ is the kernel function of the integral operator $\mathbb{T}_\delta(\lambda)$ defined in \eqref{eq93}. We only need to consider the functions $U_1$ and $U_2$. We first study the asymptotics of $U_2$. Set $\alpha_*=\frac{\theta_*}{\gamma_*}$ and $\beta_*=\frac{t_*}{\gamma_*}$. By Theorem \ref{dispersion relation near the dirac point}, for $p\in(\pi-\delta^{\frac{1}{3}},\pi+\delta^{\frac{1}{3}})$ we have
\begin{equation} \label{eq136}
    \left\{
    \begin{aligned}
    &\lambda_{1,\delta}(p)
    =\lambda_*-\sqrt{\delta^2\beta^2_*+\alpha^2_*(p-p_*)^2}\left(1+\mathcal{O}(p-p_*,\delta)\right), \\
    &u_{1,\delta}(x;p)=N(p;\delta)\left(\frac{i\alpha_*(p-p_*)}{\delta \beta_*+\sqrt{\delta^2 \beta_*^2+\alpha_*^2 (p-p_*)^2}}\phi_2(x)+r_1(x;p,\delta)\right),
    \enspace\|r_1(x;p,\delta)\|_{H^{\frac{1}{2}}(\Gamma)}=\mathcal{O}(\delta^{\frac{1}{3}})
    \end{aligned}
    \right.
\end{equation}
\begin{equation} \label{eq137}
    \left\{
    \begin{aligned}
    &\lambda_{2,\delta}(p)
    =\lambda_*+\sqrt{\delta^2\beta^2_*+\alpha^2_*(p-p_*)^2}\left(1+\mathcal{O}(p-p_*,\delta)\right), \\
    &u_{2,\delta}(x;p)=N(p;\delta)\left(\phi_2(x)+r_2(x;p,\delta)\right),\quad \|r_2(x;p,\delta)\|_{H^{\frac{1}{2}}(\Gamma)}=\mathcal{O}(\delta^{\frac{1}{3}}),
    \end{aligned}
    \right.
\end{equation}
where we have used the fact that the odd Dirac eigenmode $\phi_1$ vanishes on the interface $\Gamma$. The normalization factor $N(p)$ admits the following expansion:
\begin{equation*} \label{eq138}
    N(p;\delta)=\left(1+L(p;\delta)+\mathcal{O}(\delta^{\frac{1}{3}})\right)^{-\frac{1}{2}},
\end{equation*}
where $L(p;\delta)=\frac{\alpha_*^2(p-p_*)^2}{(\delta \beta_*+\sqrt{\delta^2 \beta_*^2+\alpha_*^2 (p-p_*)^2})^2}$. By substituting \eqref{eq137} into the integral $U_2(x,y;\lambda,\delta)$ and setting $h:=\frac{\lambda-\lambda_*}{\delta}$, we have
\begin{equation} \label{eq143}
    \begin{aligned}
    &U_2(x,y;\lambda,\delta) \\
    &=\frac{1}{2\pi}\int_{\pi-\delta^{\frac{1}{3}}}^{\pi+\delta^{\frac{1}{3}}}\frac{\phi_2(x)\overline{\phi_2(y)}+
    \phi_2(x)\overline{r_2(y;p,\delta)}
    +r_2(x;p,\delta)\overline{\phi_2(y)}
    +r_2(x;p,\delta)\overline{r_2(y;p,\delta)}
    }{\delta\cdot h-\sqrt{\delta^2 \beta_*^2+\alpha_*^2 (p-p_*)^2}\left(1+\mathcal{O}(p-p_*,\delta)\right)}N^2(p;\delta)dp.
    \end{aligned}
\end{equation}
Observe that the following estimates hold uniformly for $|p-p^*|\leq\delta^{\frac{1}{3}}$ and $h\in \tilde{J}$:
\begin{equation} \label{eq139}
    \begin{aligned}
    &\delta\cdot h-\sqrt{\delta^2 \beta_*^2+\alpha_*^2 (p-p_*)^2}\left(1+\mathcal{O}(p-p_*,\delta)\right)=\left(\delta\cdot h-\sqrt{\delta^2 \beta_*^2+\alpha_*^2 (p-p_*)^2}\right)(1+\mathcal{O}(\delta^{\frac{1}{3}})), \\
    &\quad\qquad\qquad\quad N(p;\delta)
    =\left(1+L(p;\delta)\right)^{-\frac{1}{2}}(1+\mathcal{O}(\delta^{\frac{1}{3}})),\quad \|r_2(x;p,\delta)\|_{H^{\frac{1}{2}}(\Gamma)}=\mathcal{O}(\delta^{\frac{1}{3}}).
    \end{aligned}
\end{equation}

We can show that the leading order term of $U_2$ is given by $f_2(h;\delta)\phi_2(x)\overline{\phi_2(y)}$, where
\begin{equation*}
    f_2(h;\delta)=\frac{1}{2\pi}\int_{\pi-\delta^{\frac{1}{3}}}^{\pi+\delta^{\frac{1}{3}}}\frac{1}{\delta\cdot h-\sqrt{\delta^2 \beta_*^2+\alpha_*^2 (p-p_*)^2}}\frac{1}{1+L(p;\delta)}dp. 
\end{equation*}
The remainder term is denoted as $R_2(x, y; h, \delta):=U_2-f_2\phi_2(x)\overline{\phi_2(y)}$. 
Let $\mathbb{R}_2(h;\delta)$ be the integral operator with kernel $R_2(x, y; h, \delta)$. Note that \eqref{eq139} gives that
\begin{equation*}
\begin{aligned}
R_2(x, y; h, \delta) 
&= \frac{1}{2\pi}\int_{\pi-\delta^{\frac{1}{3}}}^{\pi+\delta^{\frac{1}{3}}}\frac{\phi_2(x)\overline{r_2(y;p,\delta)}
    +r_2(x;p,\delta)\overline{\phi_2(y)}
    +r_2(x;p,\delta)\overline{r_2(y;p,\delta)}}{\delta\cdot h-\sqrt{\delta^2 \beta_*^2+\alpha_*^2 (p-p_*)^2}}\frac{1+\mathcal{O}(\delta^{\frac{1}{3}})}{1+L(p;\delta)}dp \\
&\quad +\frac{1}{2\pi}\int_{\pi-\delta^{\frac{1}{3}}}^{\pi+\delta^{\frac{1}{3}}}\frac{\phi_2(x)\overline{\phi_2(y)}}{\delta\cdot h-\sqrt{\delta^2 \beta_*^2+\alpha_*^2 (p-p_*)^2}}\frac{\mathcal{O}(\delta^{\frac{1}{3}})}{1+L(p;\delta)}dp.
\end{aligned}
\end{equation*}
By \eqref{eq139}, we have $\max_{|p-p_*|<\delta^{\frac{1}{3}}}\|r_2(x;p,\delta)\|_{H^{\frac{1}{2}}(\Gamma)}=\mathcal{O}(\delta^{\frac{1}{3}})$. 
On the other hand, 
\begin{equation*}
 \begin{aligned}
&\int_{\pi-\delta^{\frac{1}{3}}}^{\pi+\delta^{\frac{1}{3}}}\Big|\frac{1}{\delta\cdot h-\sqrt{\delta^2 \beta_*^2+\alpha_*^2 (p-p_*)^2}}\Big| |\frac{1}{1+L(p;\delta)} |dp \\
&\lesssim \int_{\pi-\delta^{\frac{1}{3}}}^{\pi+\delta^{\frac{1}{3}}}\Big|\frac{1}{\delta\cdot h-\sqrt{\delta^2 \beta_*^2+\alpha_*^2 (p-p_*)^2}}\Big|  dp \\
&=\int_{-tan^{-1}(\frac{\alpha_*}{\beta_*}\delta^{-\frac{2}{3}})}^{tan^{-1} (\frac{\alpha_*}{\beta_*}\delta^{-\frac{2}{3}})}\frac{\beta_*}{\alpha_*}\cdot\Big|\frac{\sec^2(\theta)}{h-|\beta_*|\sec(\theta)}\Big|d\theta
=\mathcal{O}(\log (\delta)).
\end{aligned}
\end{equation*}
Therefore, we can conclude that $\|\mathbb{R}_2(h;\delta)\|_{\mathcal{B}(\tilde{H}^{-\frac{1}{2}}(\Gamma),H^{\frac{1}{2}}(\Gamma))}=\mathcal{O}(\delta^{\frac{1}{3}}\log (\delta))$, whence the second equality in \eqref{eq92} follows.

The analysis of $U_1$ is similar. Set its leading-order term as
\begin{equation*}
f_1(h;\delta)\phi_2(x)\overline{\phi_2(y)}:=\left(\frac{1}{2\pi}\int_{\pi-\delta^{\frac{1}{3}}}^{\pi+\delta^{\frac{1}{3}}}\frac{1}{\delta\cdot h+\sqrt{\delta^2 \beta_*^2+\alpha_*^2 (p-p_*)^2}}\frac{L(p;\delta)}{1+L(p;\delta)}dp\right)\phi_2(x)\overline{\phi_2(y)}
\end{equation*}
and the remainder term $R_1(x,y;h,\delta):=U_1(x,y;\lambda,\delta)-f_1(h;\delta)\phi_2(x)\overline{\phi_2(y)}$. We denote the operator associated with $f_1(h;\delta)\phi_2(x)\overline{\phi_2(y)}$ and $R_1(x,y;h,\delta)$ by $f_1(h;\delta)\mathbb{P}$ and $\mathbb{R}_1(h;\delta)$, respectively. Then the first equality in \eqref{eq92} follows directly by repeating the work of estimating $f_2$ and $\mathbb{R}_2$ on $f_1$ and $\mathbb{R}_1$, and we omit it for brevity.  

Combining all the results above, we arrive at \eqref{eq90}. Finally, we point out that there is no essential difference between the analysis of $\mathbb{G}_{\delta}$ and $\mathbb{G}_{-\delta}$: by replacing \eqref{eq136} and \eqref{eq137} with \eqref{eq42} and \eqref{eq43}, we can deduce \eqref{eq94} by following the same line of argument.

\end{proof}

\subsection{Appendix E: Proof of Proposition \ref{lem-t0} and \ref{lem-t1}}
\setcounter{equation}{0}
\setcounter{subsection}{0}
\renewcommand{\theequation}{E.\arabic{equation}}
\renewcommand{\thesubsection}{E.\arab ic{subsection}}
In Appendix E-F, the duality pair between $\tilde{H}^{-\frac{1}{2}}(\Gamma)$ and $H^{\frac{1}{2}}(\Gamma)$ is denoted by $\langle \cdot,\cdot\rangle$.

\begin{proof} [Proof of Proposition \ref{lem-t0}]
We first prove (\ref{eq-kernel-T_0}). 
Observe that for $y\in\Gamma$, it follows from (\ref{eq97}) that
$$
     \mathbb{T}_0\left(\frac{\partial v_1}{\partial x_1}\Big|_\Gamma\right)(y)
{=}\int_{\Gamma}G(y,x;\lambda_*)\frac{\partial v_1}{\partial x_1}dx_2
    +\frac{i}{2\alpha_*}\int_{\Gamma}\left(
v_1(y)\overline{v_1(x)}+v_2(y)\overline{v_2(x)}
    \right)\frac{\partial v_1}{\partial x_1}(x)dx_2.  
$$
By Proposition \ref{representation formula}, we have $\int_{\Gamma}G(y,x;\lambda_*)\frac{\partial v_1}{\partial x_1}dx_2=\frac{1}{2}v_1(y)$; on the other hand, Lemma \ref{conjugate vs reflection} implies that $\frac{i}{2\alpha_*}\int_{\Gamma}\left(
v_1(y)\overline{v_1(x)}+v_2(y)\overline{v_2(x)}
    \right)\frac{\partial v_1}{\partial x_1}(x)dx_2=\frac{i}{\alpha_*}\int_{\Gamma}
v_1(y)\overline{v_1(x)}\frac{\partial v_1}{\partial x_1}(x)dx_2$. Thus
$$
\mathbb{T}_0\left(\frac{\partial v_1}{\partial x_1}\Big|_\Gamma\right)(y)   {=}\frac{1}{2}v_1(y)+\frac{i}{\alpha_*}v_1(y)\int_{\Gamma}\overline{v_1(x)}\frac{\partial v_1}{\partial x_1}(x)dx_2.
$$
By Lemma \ref{orthogonality of the propagating mode}, we obtain
$$
\mathbb{T}_0\left(\frac{\partial v_1}{\partial x_1}\Big|_\Gamma\right)(y)=\frac{1}{2}v_1(y)-\frac{1}{2}v_1(y)=0.
$$
Thus $\text{span}\big\{\frac{\partial v_1}{\partial x_1}\big|_{\Gamma}\big\}\subset\text{Ker}\mathbb{T}_0$. Conversely, suppose $\psi\in \text{Ker}\mathbb{T}_0$ and $\langle \psi,\overline{v_1}\rangle=0$. We aim to show $\psi=0$. Note that \eqref{eq-73} and \eqref{eq97} lead to
\begin{equation*} \label{eq148}
    \begin{aligned}
        0&=2\mathbb{T}_0 \psi
        =\left(2\int_{\Gamma}\Big(G_{0}^{+}(x,y;\lambda_*)- \frac{i}{2\alpha_*}v_1(x)\overline{v_1(y)}+\frac{i}{2\alpha_*}v_2(x)\overline{v_2(y)}\Big)\psi(y)dy_2\right)\Big|_{\Gamma} \\
        &=\left(2\int_{\Gamma}G_{0}^{+}(x,y;\lambda_*)\psi(y)dy_2\right)\Big|_{\Gamma}-\frac{i}{\alpha_*}v_1(x)\Big|_\Gamma\cdot \int_{\Gamma}\overline{v_1(y)}\psi(y)dy_2
        +\frac{i}{\alpha_*}v_2(x)\Big|_\Gamma\cdot \int_{\Gamma}\overline{v_2(y)}\psi(y)dy_2
    \end{aligned}
\end{equation*}
Then Lemma \ref{conjugate vs reflection} gives that
$$
\begin{aligned}
0&{=}\left(2\int_{\Gamma}G_{0}^{+}(x,y;\lambda_*)\psi(y)dy_2\right)\Big|_{\Gamma}-\frac{i}{\alpha_*}v_1(x)\Big|_\Gamma\cdot \int_{\Gamma}\overline{v_1(y)}\psi(y)dy_2
        +\frac{i}{\alpha_*}v_1(x)\Big|_\Gamma\cdot \int_{\Gamma}\overline{v_1(y)}\psi(y)dy_2 \\
        &=\left(2\int_{\Gamma}G_{0}^{+}(x,y;\lambda_*)\psi(y)dy_2\right)\Big|_{\Gamma}.
\end{aligned}
$$
Thus, if we define $v(x):=2\int_{\Gamma}G_{0}^{+}(x,y;\lambda_*)\psi(y)dy_2$ for $x\in\Omega^+$, then we have $v|_{\Gamma}=0$. Moreover, using \eqref{eq-73} and the fact that $\langle \psi,\overline{v_1}\rangle=0$, we have another expression of $v(x)$
\begin{equation} \label{eq150}
v(x)=2\int_{\Gamma}G(x,y;\lambda_*)\psi(y)dy_2,\quad x\in\Omega^+.
\end{equation}
We next show that
$v(x)\equiv 0$ for $x\in\Omega^+$, which shall lead to $\psi=0$ as claimed. For this purpose, we consider the following odd extension of $v$:
\begin{equation*}
    \tilde{v}(x_1,x_2)=\left\{
    \begin{aligned}
    &v(x_1,x_2),\quad x_1\geq 0, \\
    &-v(-x_1,x_2),\quad x_1<0,
    \end{aligned}
    \right.
\end{equation*}
Since $v|_{\Gamma}=0$, we have $\tilde{v}\in\mathcal{V}(\lambda_*)\bigcap L^2(\Omega)$ by noting that $G_{0}^+$ introduced in \eqref{eq-73}
decays exponentially. Thus $\tilde{v}$ gives a $L^2$-eigenmode for the 
unperturbed periodic structure $\Omega$. But by Assumption \ref{assump0}, $\lambda_*$ is not an embedded eigenvalue. Therefore $\tilde{v}\equiv 0$ and hence $v(x)\equiv 0$ for $x\in\Omega^+$.
On the other hand,  using Lemma \ref{commutativity and parity of the Green function}, \eqref{eq150} implies that $v(x)$ can be extended to the whole space $\Omega$. Moreover, the extended function is identically zero in $\Omega$. It follows that
\begin{equation*}
    \begin{aligned}
        0=(\Delta_x+\lambda_*)v(x)&=2(\Delta_x+\lambda_*)\int_{\Gamma}G(x,y;\lambda_*)\psi(y_2)dy_2 \\
        &=2(\Delta_x+\lambda_*)\int_{\Omega}G(x,y;\lambda_*)\left(\psi(y_2)\tilde{\delta}(y_1)\right)dy=\psi(y_2)\tilde{\delta}(y_1).
    \end{aligned}
\end{equation*}
Therefore $\psi=0$ in $\tilde{H}^{-\frac{1}{2}}(\Gamma)$. 
In conclusion, $\text{Ker}\mathbb{T}_0$ is at most one-dimensional. Finally, since  $\text{Ker}\mathbb{T}_0\supset\text{span}(\frac{\partial v_1}{\partial x_1}\big|_{\Gamma})$ and $\langle \frac{\partial v_1}{\partial x_1},\overline{v_1}\rangle\neq 0$, we conclude that $\text{Ker}\mathbb{T}_0=\text{span}\{\frac{\partial v_1}{\partial x_1}\big|_{\Gamma}\}$.
\end{proof} 

\bigskip

\begin{proof}[Proof of Proposition \ref{lem-t1}]
The proof here follows the same lines as in Lemma \ref{particle integral operator at the Dirac point}. We point out the major difference in the proof and skip the analogous steps.

Let $\text{Tr}:H^1(\Omega)\to H^{\frac{1}{2}}(\Gamma),\quad f\mapsto f|_{\Gamma}$ be the trace operator and $E:H^{\frac{1}{2}}(\Gamma)\to H^{1}(\Omega)$ be the Sobolev extension operator such that $\text{Tr}\circ E=id|_{H^{\frac{1}{2}}(\Gamma)}$. For each $\psi\in \tilde{H}^{-\frac{1}{2}}(\Gamma)$, let $c_n(\psi;p):=\langle \psi,\overline{u_n(\cdot;p)}\rangle$, where $\{u_n(\cdot;p)\}_{n\geq 1}$ are the Bloch eigenmodes associated with the waveguide in Figure 1 for the eigenvalue $\lambda_n(p)$. Now we decompose $\mathbb{T}_0$ as $\mathbb{T}_0=\mathbb{T}_0^{(1)}+\mathbb{T}_0^{(2)}$, where for any $\psi\in \tilde{H}^{-\frac{1}{2}}(\Gamma)$,
\begin{equation*}
\left\{
\begin{aligned}
&\mathbb{T}_0^{(1)}\psi:=
\text{Tr}\left(\frac{1}{2\pi}\int_{0}^{2\pi}\sum_{n\geq 3}\frac{c_n(\psi;p)}{\lambda_*-\lambda_n(p)}u_n(\cdot;p)dp
-\frac{1}{2\pi}\int_{0}^{2\pi}\sum_{n=1,2}c_n(\psi;p)u_n(\cdot;p)dp\right), \\
&\mathbb{T}_0^{(2)}\psi:=
\text{Tr}\left(\frac{1}{2\pi}\text{p.v.}\int_{0}^{2\pi}\sum_{n=1,2}\frac{c_n(\psi;p)}{\lambda_*-\mu_n(p)}v_n(\cdot;p)dp
+\frac{1}{2\pi}\int_{0}^{2\pi}\sum_{n=1,2}c_n(\psi;p)u_n(\cdot;p)dp\right).
\end{aligned}
\right. 
\end{equation*}
Similar to the proof of Lemma \ref{particle integral operator at the Dirac point}, we shall show that $\mathbb{T}_0^{(1)}$ is invertible while $\mathbb{T}_0^{(2)}$ is compact, which then implies that $\mathbb{T}_0$ is a Fredholm operator of zero index.

To show the invertibility of $\mathbb{T}_0^{(1)}$, we only need to: (1) establish an estimate analogous to \eqref{eq_A_2}, which implies the injectivity of $\mathbb{T}_0^{(1)}$ and closedness of $\text{Ran}\,(\mathbb{T}_0^{(1)})$; (2) obtain an identity parallel to \eqref{eq_A_3} which proves that $\text{Ran}\,(\mathbb{T}_0^{(1)})$ is dense in $H^{\frac{1}{2}}(\Gamma)$. Note that \eqref{eq_A_3} for $\mathbb{T}_0^{(1)}$ stills holds when the dual pair is taken on $\tilde{H}^{-\frac{1}{2}}(\Gamma)\times H^{\frac{1}{2}}(\Gamma)$, thus (2) is proved. Now we prove (1). Note that the following inequality, which is the counterpart of \eqref{eq36}, is straightforward
\begin{equation} \label{eq_F_1}
|\langle\overline{\psi}, \mathbb{T}_0^{(1)}\psi\rangle |\gtrsim \int_{0}^{2\pi}\sum_{n\geq 1}\frac{|c_n(\psi;p)|^2}{|\lambda_*-\lambda_n(p)|}dp.
\end{equation}
By the Floquet-Bloch theory, we can write $E\phi=\int_{0}^{2\pi}\sum_{n\geq 1}a_n(p)u_n(x;p)dp$ for each $\phi\in H^{\frac{1}{2}}(\Gamma)$. It follows that
\begin{equation*}
\begin{aligned}
\Big|\langle \psi, \overline{\phi}\rangle\Big|
&=\Big|\langle \psi, \overline{\text{Tr}(E\phi)}\rangle \Big|
=\Big|\int_{0}^{2\pi}\sum_{n\geq 1}\overline{a_n(p)}\cdot c_n(\psi;p)dp\Big| \\
&\leq\Big|\int_{0}^{2\pi}(\sum_{n\geq 1}|\lambda_*-\lambda_n(p)||a_n(p)|^2)^{\frac{1}{2}}
\cdot (\sum_{n\geq 1}\frac{|c_n(\psi;p)|^2}{|\lambda_*-\lambda_n(p)|})^{\frac{1}{2}}dp\Big| \\
&\leq 
\left(\int_{0}^{2\pi}\sum_{n\geq 1}|\lambda_*-\lambda_n(p)||a_n(p)|^2 dp\right)^{\frac{1}{2}}
\left(\int_{0}^{2\pi}\sum_{n\geq 1}\frac{|c_n(\psi;p)|^2}{|\lambda_*-\lambda_n(p)|} dp\right)^{\frac{1}{2}} \\
&\lesssim
\left(\int_{0}^{2\pi}\sum_{n\geq 1}(1+\lambda_n(p))|a_n(p)|^2 dp\right)^{\frac{1}{2}}
\left(\int_{0}^{2\pi}\sum_{n\geq 1}\frac{|c_n(\psi;p)|^2}{|\lambda_*-\lambda_n(p)|} dp\right)^{\frac{1}{2}}. 
\end{aligned}
\end{equation*}
Using \eqref{eq87_5}, we further obtain
\[
\Big|\langle \psi, \overline{\phi}\rangle\Big| \lesssim \|E\phi\|_{H^1(\Omega)}\left(\int_{0}^{2\pi}\sum_{n\geq 1}\frac{|c_n(\psi;p)|^2}{|\lambda_*-\lambda_n(p)|} dp\right)^{\frac{1}{2}}.
\]
Since $\|E\phi\|_{H^1(\Omega)} \lesssim  \|\phi\|_{H^\frac{1}{2}(\Gamma)}$, we can conclude that 
\begin{equation} \label{eq_F_2}
    \|\psi\|_{\tilde{H}^{-\frac{1}{2}}(\Gamma)}\lesssim \left(\int_{0}^{2\pi}\sum_{n\geq 1}\frac{|c_n(\psi;p)|^2}{|\lambda_*-\lambda_n(p)|} dp\right)^{\frac{1}{2}}.
\end{equation}
Then the desired estimate $|\langle\overline{\psi}, \mathbb{T}_0^{(1)}\psi\rangle |
\gtrsim \|\psi\|_{\tilde{H}^{-\frac{1}{2}}(\Gamma)}$ follows from \eqref{eq_F_1} and \eqref{eq_F_2}. 

We next show that $\mathbb{T}_0^{(2)}:\tilde{H}^{-\frac{1}{2}}(\Gamma)\to H^{\frac{1}{2}}(\Gamma)$ is compact. To this end, we fix a smooth domain $O$ such that $\Gamma \subset O \subset \Omega$. Then, by using the compactness of the embedding of $H^{2}(O)$ into $H^{1}(O)$ and the boundedness of the restriction operator from $H^{1}(O)$ to $H^{\frac{1}{2}}(\Gamma)$, it is sufficient to show that the natural extension of $(\mathbb{T}_0^{(2)}\psi)(x)$ ($x\in O$) is uniformly bounded in $H^{2}(O)$-norm for $\|\psi\|_{\tilde{H}^{-\frac{1}{2}}(\Gamma)}\leq 1$. Here we only estimate the $H^{2}(O)$-norm of $u_{\psi}(x):=\frac{1}{2\pi}\text{p.v.}\int_{0}^{2\pi}\frac{c_1(\psi;p)}{\lambda_*-\mu_1(p)}v_1(x;p)dp$, while the other terms in $\mathbb{T}_0^{(2)}\psi$ can be estimated similarly. To proceed, recall that both $\mu_1(p)$ and $v_1(x;p)$ are analytic in $p$ within a complex neighborhood of $[0,2\pi]$, which implies the following inequality by the principal value estimate of Banach-valued function\cite{CLEMENT1991453},
\begin{equation} \label{eq_F_3}
    \|u_{\psi}\|_{H^2(O)}
    \lesssim \max_{0\leq p\leq 2\pi}\left(\|v_1(\cdot;p)\|_{H^2(O)}+\|\partial_p v_1(\cdot;p)\|_{H^2(O)}\right)\|\psi\|_{\tilde{H}^{-\frac{1}{2}}(\Gamma)}.
\end{equation}
For the first term in \eqref{eq_F_3}, note that the regularity of Laplacian eigenfunctions implies that $\|v_1(\cdot;p)\|_{H^2(O)}<\infty$ for each $p\in[0,2\pi]$. Hence the analyticity gives $\max_{0\leq p\leq 2\pi}\|v_1(\cdot;p)\|_{H^2(O)}<\infty$. Second, note that the partial derivative $\partial_p v_1(\cdot;p)$ solves the following equation
\begin{equation*}
(\Delta_x +\mu_{1}(p))\partial_p v_1(x;p)=\mu_1' v_1(x;p),\quad x\in O.
\end{equation*}
Then a standard regularity theory of elliptic equations implies that $\max_{0\leq p\leq 2\pi}\|\partial_p v_1(\cdot;p)\|_{H^2(O)}<\infty$. Thus the uniformly boundedness of $\|u_{\psi}\|_{H^2(O)}$ follows from \eqref{eq_F_3}. This completes the proof of the proposition. 
\end{proof}

\subsection{Appendix F: Proof of Proposition \ref{limit of the sum operator}}
\setcounter{equation}{0}
\setcounter{subsection}{0}
\renewcommand{\theequation}{F.\arabic{equation}}
\renewcommand{\thesubsection}{F.\arabic{subsection}}
% As in Appendix E, the bracket $\langle \cdot,\cdot\rangle$ is referred to the dual pair on $H^{-\frac{1}{2}}(\Gamma)\times H^{\frac{1}{2}}(\Gamma)$.
\begin{proof}[Proof]
Step 1.  By \eqref{eq90} and \eqref{eq94},
\begin{equation} \label{eq157}
    \begin{aligned}
    \big(\mathbb{G}_{\delta}+\mathbb{G}_{-\delta}\big)&(\lambda_*+\delta\cdot h)
    =\big(\mathbb{T}_{\delta}+\mathbb{T}_{-\delta}\big)(\lambda_*+\delta\cdot h) \\
    &+\left[(f_1+f_2+\tilde{f}_1+\tilde{f}_2)(h;\delta)\mathbb{P}+(\mathbb{R}_1+\mathbb{R}_2+\tilde{\mathbb{R}}_1+\tilde{\mathbb{R}}_2)(h;\delta)\right].
    \end{aligned}
\end{equation}
By Proposition \ref{Asymptotic formula of Green operators},  $\lim_{\delta\to 0}\|\mathbb{R}_1+\mathbb{R}_2+\tilde{\mathbb{R}}_1+\tilde{\mathbb{R}}_2\|=0$. 
Moreover, a direct calculation shows that
\begin{equation*} \label{eq158}
    \begin{aligned}
    (f_1&+f_2+\tilde{f}_1+\tilde{f}_2)(h;\delta)
    =\frac{1}{2\pi}\int_{\pi-\delta^{\frac{1}{3}}}^{\pi+\delta^{\frac{1}{3}}}\left(\frac{1}{\delta\cdot h-\sqrt{\delta^2 \beta_*^2+\alpha_*^2 (p-p_*)^2}}+\frac{1}{\delta\cdot h+\sqrt{\delta^2 \beta_*^2+\alpha_*^2 (p-p_*)^2}}\right) dp \\
    &=-\frac{\delta\cdot h}{\pi}\int_{\pi-\delta^{\frac{1}{3}}}^{\pi+\delta^{\frac{1}{3}}}\frac{1}{\delta^2 (\beta_*^2-h^2)+\alpha_*^2 (p-p_*)^2}dp =-\frac{1}{\pi\beta_*\alpha_*}\frac{h}{\sqrt{1-(\frac{h}{\beta_*})^2}}\cdot 2\tan^{-1}(\frac{\alpha_*}{\sqrt{\beta_*^2-h^2}}\delta^{-\frac{2}{3}}).
    \end{aligned}
\end{equation*}
For $h\in\tilde{J}$, we can show that the following  convergence holds uniformly:
\begin{equation*} \label{eq159_1}
    \lim_{\delta\to 0}(f_1+f_2+\tilde{f}_1+\tilde{f}_2)(h;\delta)=\beta(h):=-\frac{1}{\beta_*\alpha_*}\frac{h}{\sqrt{1-(\frac{h}{\beta_*})^2}}.
\end{equation*}

We now investigate the limit of the term 
$\mathbb{T}_\delta(\lambda_*+\delta\cdot h)$ in the rest of the proof.

Step 2. We decompose respectively the operators $\mathbb{T}_\delta(\lambda)$ and $\mathbb{T}_0$ as
$$
\mathbb{T}_\delta(\lambda)= \mathbb{T}^{prop}_\delta(\lambda) + 
\mathbb{T}^{evan}_\delta(\lambda), \quad \mathbb{T}_0= \mathbb{T}^{prop}_0 + 
\mathbb{T}^{evan}_0,
$$
with
\begin{equation} \label{eq161}
    \left\{
    \begin{aligned}
        &\mathbb{T}^{prop}_\delta(\lambda)=\frac{1}{2\pi}\int_{[0,\pi-\delta^{\frac{1}{3}}]\bigcup[\pi+\delta^{\frac{1}{3}},2\pi]}\sum_{n=1,2}\frac{\langle \cdot,\overline{u_{n,\delta}(x;p)}\rangle}{\lambda-\lambda_{n,\delta}(p)}u_{n,\delta}(x;p)dp,\\
        &\mathbb{T}^{evan}_\delta(\lambda)=\frac{1}{2\pi}\int_{0}^{2\pi}\sum_{n\geq 3}\frac{\langle \cdot,\overline{u_{n,\delta}(x;p)}\rangle}{\lambda-\lambda_{n,\delta}(p)}u_{n,\delta}(x;p)dp,
    \end{aligned}
    \right.
\end{equation}
\begin{equation*} \label{eq162}
    \left\{
    \begin{aligned}
        &\mathbb{T}^{prop}_0=\frac{1}{2\pi}\lim_{\epsilon\to 0}\int_{[0,\pi-\epsilon)\bigcup(\pi+\epsilon,2\pi]}\sum_{n=1,2}\frac{\langle \cdot,\overline{v_{n}(x;p)}\rangle}{\lambda_*-\lambda_{n}(p)}v_{n}(x;p)dp, \\
        &\mathbb{T}^{evan}_0=\frac{1}{2\pi}\int_{0}^{2\pi}\sum_{n\geq 3}\frac{\langle \cdot,\overline{u_{n}(x;p)}\rangle}{\lambda_*-\lambda_{n}(p)}u_{n}(x;p)dp,
    \end{aligned}
    \right.
\end{equation*}
where $v_n$ and $\mu_n$ ($n=1,2$) are introduced in Proposition \ref{Existence of Dirac points of the period-1 structure}. We further introduce the following auxiliary operator
\begin{equation*}
\begin{aligned}
\mathbb{T}^{prop}_{0,\delta}&=\frac{1}{2\pi}\sum_{n=1,2}\int_{[0,\pi-\delta^{\frac{1}{3}}]\bigcup[\pi+\delta^{\frac{1}{3}},2\pi]}\frac{\langle \cdot,\overline{u_{n}(x;p)}\rangle}{\lambda_*-\lambda_{n}(p)}u_{n}(x;p)dp \\
&=\frac{1}{2\pi}\sum_{n=1,2}\int_{[0,\pi-\delta^{\frac{1}{3}}]\bigcup[\pi+\delta^{\frac{1}{3}},2\pi]}\frac{\langle \cdot,\overline{v_{n}(x;p)}\rangle}{\lambda_*-\mu_{n}(p)}v_{n}(x;p)dp.
\end{aligned}
\end{equation*}

Step 3. In this step, we show that 
\begin{equation} \label{eq166}
    \lim_{\delta\to 0}\|\mathbb{T}^{prop}_{\delta}(\lambda_*+\delta\cdot h)-\mathbb{T}^{prop}_{0}\|_{\mathcal{B}(\tilde{H}^{-\frac{1}{2}}(\Gamma),H^{\frac{1}{2}}(\Gamma))}=0.
\end{equation} 
First of all, we prove the following limit
\begin{equation} \label{eq-aaa}
    \lim_{\delta\to 0}\|\mathbb{T}^{prop}_{0,\delta}-\mathbb{T}_0^{prop}\|_{\mathcal{B}(\tilde{H}^{-\frac{1}{2}}(\Gamma),H^{\frac{1}{2}}(\Gamma))}=0,
\end{equation}
Indeed, the definitions of $\mathbb{T}^{prop}_{0,\delta}$ and $\mathbb{T}_0^{prop}$ give that
\begin{equation*}
    (\mathbb{T}^{prop}_{0,\delta}-\mathbb{T}_0^{prop})\psi
    =\frac{1}{2\pi}\sum_{n=1,2}\lim_{\epsilon\to 0}\int_{[\pi-\delta^{\frac{1}{3}},\pi-\epsilon)\bigcup(\pi+\epsilon,\pi+\delta^{\frac{1}{3}},2\pi]}\frac{\langle \psi(y),\overline{v_{n}(y;p)}\rangle}{\lambda_*-\mu_{n}(p)}v_{n}(x;p)dp
    ,\quad \psi\in \tilde{H}^{-\frac{1}{2}}(\Gamma).
\end{equation*}
Then the analyticity of $v_n$ and $\mu_n$ ($n=1,2$) give the following estimate
\begin{equation*}
\begin{aligned}
\|(\mathbb{T}^{prop}_{0,\delta}-\mathbb{T}_0^{prop})\psi\|_{H^{\frac{1}{2}}(\Gamma)}
&\lesssim \delta^{\frac{1}{3}}\sum_{n=1,2}
\left(\|v_{n}(x;\pi)\|_{H^{\frac{1}{2}}(\Gamma)}
+\|\partial_p v_{n}(x;\pi)\|_{H^{\frac{1}{2}}(\Gamma)}\right)\|\psi\|_{\tilde{H}^{-\frac{1}{2}}(\Gamma)}
\end{aligned}
\end{equation*}
which is similar to \eqref{eq_F_3}. Thus
\begin{equation*}
    \|\mathbb{T}^{prop}_{0,\delta}-\mathbb{T}_0^{prop}\|_{\mathcal{B}(\tilde{H}^{-\frac{1}{2}}(\Gamma),H^{\frac{1}{2}}(\Gamma))}
    \lesssim \delta^{\frac{1}{3}}\sum_{n=1,2}\left( \|v_{n}(x;\pi)\|_{H^{\frac{1}{2}}(\Gamma)}
+\|\partial_p v_{n}(x;\pi)\|_{H^{\frac{1}{2}}(\Gamma)} \right),
\end{equation*}
where \eqref{eq-aaa} follows.

Next, we prove that
\begin{equation} \label{eq163}
    \lim_{\delta \to 0}\|\mathbb{T}^{prop}_\delta(\lambda_*+\delta\cdot h)-\mathbb{T}^{prop}_\delta(\lambda_*)\|=0.
\end{equation}
Note that by Theorem \ref{dispersion relation near the dirac point} and Assumption \ref{assump0}, the following estimate holds uniformly for all $h\in\tilde{J}$ and $p\in [0,\pi-\delta^{\frac{1}{3}}]\bigcup [\pi+\delta^{\frac{1}{3}},2\pi ]$
\begin{equation*}
    |\lambda_*+\delta\cdot h-\lambda_{n,\delta}(p)|\gtrsim \delta^{\frac{1}{3}},
    \quad \|u_{n,\delta}(\cdot;p)\|_{H^{\frac{1}{2}}(\Gamma)}=\mathcal{O}(1)
    \quad (n=1,2).
\end{equation*}
Next, for each $\psi\in \tilde{H}^{-\frac{1}{2}}(\Gamma)$, let $v:=\left(\mathbb{T}^{prop}_\delta(\lambda_*+\delta\cdot h)-\mathbb{T}^{prop}_\delta(\lambda_*)\right)\psi\in H^{\frac{1}{2}}(\Gamma)$. Note that $v(x)$ can be extended to $x\in\Omega_\delta$ using \eqref{eq161}. We can show that
\begin{equation} \label{eq-T-estimate}
    \begin{aligned}
        \|v\|_{H^1(\Omega_\delta)}^2 
        &=\Big\|\frac{1}{2\pi}\sum_{n=1,2}\int_{[0,\pi-\delta^{\frac{1}{3}}]\bigcup[\pi+\delta^{\frac{1}{3}},2\pi]}\frac{\delta\cdot h\langle \psi,\overline{u_{n,\delta}(x;p)}\rangle}{(\lambda_*+\delta\cdot h-\lambda_{n,\delta}(p))(\lambda_*-\lambda_{n,\delta}(p))}u_{n,\delta}(x;p)dp\Big\|_{H^1(\Omega_\delta)}^2 \\
        &=\sum_{n=1,2}\int_{[0,\pi-\delta^{\frac{1}{3}}]\bigcup[\pi+\delta^{\frac{1}{3}},2\pi]}(1+\lambda_{n,\delta}(p))\Big|\frac{\delta\cdot h\langle \psi,\overline{u_{n,\delta}(x;p)}\rangle}{(\lambda_*+\delta\cdot h-\lambda_{n,\delta}(p))(\lambda_*-\lambda_{n,\delta}(p))}\Big|^2dp \\
        &\lesssim \delta^{\frac{2}{3}}|h|^2\sum_{n=1,2}\int_{[0,\pi-\delta^{\frac{1}{3}}]\bigcup[\pi+\delta^{\frac{1}{3}},2\pi]}|\langle \psi,\overline{u_{n,\delta}(x;p)}\rangle|^2dp \\
        &\lesssim \delta^{\frac{2}{3}}|h|^2 \sum_{n=1,2}\int_{[0,\pi-\delta^{\frac{1}{3}}]\bigcup[\pi+\delta^{\frac{1}{3}},2\pi]}\|\psi\|^2_{\tilde{H}^{-\frac{1}{2}}(\Gamma)}dp
        \lesssim\delta^{\frac{2}{3}}|h|^2\|\psi\|^2_{\tilde{H}^{-\frac{1}{2}}(\Gamma)}
    \end{aligned}
\end{equation}
where the second equality is derived from \eqref{eq87_5}. Thus we have
\begin{equation*}
     \|\left(\mathbb{T}^{prop}_\delta(\lambda_*+\delta\cdot h)-\mathbb{T}^{prop}_\delta(\lambda_*)\right)\psi\|_{H^{\frac{1}{2}}(\Gamma)}
    \lesssim
    \|v\|_{H^1(\Omega_\delta)}\lesssim \delta^{\frac{1}{3}}|h|\|\psi\|_{\tilde{H}^{-\frac{1}{2}}(\Gamma)},
\end{equation*}
whence \eqref{eq163} follows.

We then show that
\begin{equation} \label{estimate-T22}
    \|\mathbb{T}^{prop}_\delta(\lambda_*)-\mathbb{T}^{prop}_{0,\delta}\|_{\mathcal{B}(\tilde{H}^{-\frac{1}{2}}(\Gamma),H^{\frac{1}{2}}(\Gamma))}\lesssim\delta^{\frac{1}{3}}|h|.
\end{equation}
By a similar perturbation argument, as we did in the proof of Theorem \ref{dispersion relation near the dirac point}, there holds
\begin{equation*}
    |\lambda_{n,\delta}(p)-\lambda_{n}(p)|=\mathcal{O}(\delta),\quad
    \|u_{n,\delta}(\cdot;p)-u_{n}(\cdot;p)\|_{H^\frac{1}{2}(\Gamma)}=\mathcal{O}(\delta),
\end{equation*}
uniformly for $n=1,2$ and $p\in(0,\pi)\bigcup(\pi,2\pi)$. Again, for each $\phi\in \tilde{H}^{-\frac{1}{2}}(\Gamma)$, let $u:=\mathbb{T}^{prop}_\delta(\lambda_*)\phi-\mathbb{T}^{prop}_{0,\delta}\phi$ and extend $u$ as a function defined on $\Omega_\delta$. Then we have
\begin{equation*} \label{eq165_1}
    \begin{aligned}
    u=&\frac{1}{2\pi}\sum_{n=1,2}\int_{[0,\pi-\delta^{\frac{1}{3}}]\bigcup[\pi+\delta^{\frac{1}{3}},2\pi]}\frac{\langle \psi,\overline{u_{n,\delta}(x;p)}\rangle u_{n,\delta}(x;p)}{\lambda_*-\lambda_{n,\delta}(p)}-\frac{\langle \psi,\overline{u_{n,\delta}(x;p)}\rangle u_{n}(x;p)}{\lambda_*-\lambda_{n,\delta}(p)}dp \\
    &+
    \frac{1}{2\pi}\sum_{n=1,2}\int_{[0,\pi-\delta^{\frac{1}{3}}]\bigcup[\pi+\delta^{\frac{1}{3}},2\pi]}\frac{\langle \psi,\overline{u_{n,\delta}(x;p)}\rangle u_{n}(x;p)}{\lambda_*-\lambda_{n,\delta}(p)}-\frac{\langle \psi,\overline{u_{n}(x;p)}\rangle u_{n}(x;p)}{\lambda_*-\lambda_{n,\delta}(p)}dp \\
    &+
    \frac{1}{2\pi}\sum_{n=1,2}\int_{[0,\pi-\delta^{\frac{1}{3}}]\bigcup[\pi+\delta^{\frac{1}{3}},2\pi]}\frac{\langle \psi,\overline{u_{n}(x;p)}\rangle u_{n}(x;p)}{\lambda_*-\lambda_{n,\delta}(p)}-\frac{\langle \psi,\overline{u_{n}(x;p)}\rangle u_{n}(x;p)}{\lambda_*-\lambda_{n}(p)}dp.
  \end{aligned}
\end{equation*}  
By applying the same method of estimation for each term at the right-hand side above as in \eqref{eq-T-estimate}, we arrive at \eqref{estimate-T22}.

In conclusion, with \eqref{eq-aaa}, \eqref{eq163} and \eqref{estimate-T22}, we deduce \eqref{eq166}.

\medskip

Step 4. We show that 
\begin{equation} \label{eq168}
    \lim_{\delta\to 0}\|\mathbb{T}^{evan}_{\delta}(\lambda_*+\delta\cdot h)-\mathbb{T}^{evan}_{0}\|_{\mathcal{B}(\tilde{H}^{-\frac{1}{2}}(\Gamma),H^{\frac{1}{2}}(\Gamma))}=0. 
\end{equation}
To this end, define the following auxiliary operators in $\mathcal{B}(\tilde{H}^{-\frac{1}{2}}(\Gamma),H^{\frac{1}{2}}(\Gamma))$ for $p\in [0,2\pi]$
\begin{equation*}
\mathbb{T}^{evan}_{\delta}(\lambda;p):=\sum_{n\geq 3}\frac{\langle\cdot,\overline{u_{n,\delta}(x;p)}\rangle}{\lambda-\lambda_{n,\delta}(p)}u_{n,\delta}(x;p)
,\quad
\mathbb{T}^{evan}_{0}(\lambda;p):=\sum_{n\geq 3}\frac{\langle\cdot,\overline{u_{n}(x;p)}\rangle}{\lambda-\lambda_{n}(p)}u_{n}(x;p).
\end{equation*}
Then we have
\begin{equation}
\mathbb{T}^{evan}_{\delta}(\lambda)=\int_{0}^{2\pi}\mathbb{T}^{evan}_{\delta}(\lambda;p)dp,\quad
\mathbb{T}^{evan}_{0}(\lambda)=\int_{0}^{2\pi}\mathbb{T}^{evan}_{0}(\lambda;p)dp.
\end{equation}
We aim to show that (1)$\mathbb{T}^{evan}_{\delta}(\lambda_*+\delta \cdot h;p)$ has uniformly bounded norm for every $p\in[0,2\pi]$ and $\delta\ll 1$; (2)$\mathbb{T}^{evan}_{\delta}(\lambda_*+\delta\cdot h;p)$ is a continuous operator-valued function of $p$; (3)$\mathbb{T}^{evan}_{\delta}(\lambda_*+\delta\cdot h;p)$ converges to $\mathbb{T}^{evan}_{0}(\lambda_*;p)$ in operator norm for almost every $p\in [0,2\pi]$. Then \eqref{eq168} follows directly from the dominated convergence theorem. In what follows, after introducing some notations in Step 5, we shall prove (1), (2), and (3) in Step 6, Step 7, and Step 8 respectively. 

\medskip

Step 5. We fix some notations. Denote $C_\delta:=\Omega_{\delta}\cap Y$, and the closed subspace $V_{p,\delta}$ of $H^1(C_\delta)$ by $V_{p,\delta}:=\overline{\text{span}\{u_{n,\delta}(\cdot;p)\}_{n\geq 3}}\subset H^1(C_\delta)$. Then the trace operator can be defined on $V_{p,\delta}$, which is $\text{Tr}: V_{p,\delta}\to  H^\frac{1}{2}(\Gamma)$. Moreover, we use $M: \tilde{H}^{-\frac{1}{2}}(\Gamma)\to (V_{p,\delta})^*$ to represent the adjoint of $\text{Tr}$, i.e. $M:=Tr^*$. It is clear that the $\|\text{Tr}\|$ and $\|M\|$ is uniformly bounded for every $p\in[0,2\pi]$ and $\delta\ll 1$ since $D_{1,\delta}$ and $D_{2,\delta}$ are away from $\Gamma$. 
%The operators $\text{Tr}$ and $M$ are introduced to show the continuity of $\mathbb{T}^{evan}_{\delta}(\lambda;p)$ in $p$ and uniform boundedness of $\|\mathbb{T}^{evan}_{\delta}(\lambda;p)\|$. 

\medskip

Step 6. We prove the uniform boundedness of $\|\mathbb{T}^{evan}_{\delta_*}(\lambda+ \delta\cdot h;p)\|$ in $p$. 
For this purpose, we define the following sesquilinear form $a_{p,\delta}(\cdot,\cdot)$ on $V_{p,\delta}$ and its associated operator $A_{p,\delta}:V_{p,\delta}\to (V_{p,\delta})^*$ by
$$
\begin{aligned}
&a_{p,\delta}(u,v):=-\sum_{n\geq 3}\lambda_{n,\delta}(p)\langle u(x),\overline{u_{n,\delta}(x;p)}\rangle_{(V_{\delta,p})^{*}\times V_{\delta,p}}
\cdot \langle \overline{v(x)},u_{n,\delta}(x;p)\rangle_{(V_{\delta,p})^{*}\times V_{\delta,p}}, \\
&\qquad\qquad\qquad\qquad\qquad a_{p,\delta}(u,v)\equiv \langle A_{p,\delta}u,v\rangle_{(V_{\delta,p})^*\times V_{\delta,p}},
\end{aligned}
$$
then the resolvent $(\lambda+A_{p,\delta})^{-1}$ can be expanded in its spectral form when it is well-defined
$$
(\lambda+A_{p,\delta})^{-1}=\sum_{n\geq 3}\frac{1}{\lambda-\lambda_{n,\delta}(p)}\langle \cdot,\overline{u_{n,\delta}(x;p)}\rangle_{(V_{\delta,p})^*\times V_{\delta,p}}u_{n,\delta}(x;p).
$$
Thus $\mathbb{T}^{evan}_{\delta}(\lambda;p)$ admits the following factorization:
\begin{equation} \label{eq168_1}
    \mathbb{T}^{evan}_{\delta}(\lambda;p)=\text{Tr}\circ (\lambda+A_{p,\delta})^{-1}\circ M.
\end{equation}
We note that for $\lambda=\lambda_*+\delta\cdot h$ with $h\in \tilde{J}$, by \eqref{eq87_5}, there exists $c>0$, which is independent of both $\delta$ and $p$, such that 
$$
|a_{p,\delta}(u,u)+\lambda\langle u,u\rangle_{V_{\delta,p}\times V_{\delta,p}}|\geq c\|u\|^2
$$
for any $u\in V_{p,\delta}$. Thus, by Lax-Milgram theorem, we have $\|(\lambda+A_{p,\delta})^{-1}\|\leq c^{-1}$. Then the uniform boundedness of $\|\mathbb{T}^{evan}_{\delta}(\lambda_*+\delta\cdot h;p)\|$ follows by using \eqref{eq168_1}.

\medskip

Step 7. We prove the continuity of $\mathbb{T}^{evan}_{\delta}(\lambda;p)$ with respect to $p\in[0,2\pi]$. Actually, the definition of the sesquilinear $a_{p,\delta}(\cdot,\cdot)$ implies that $A_{p,\delta}=\Delta|_{V_{p,\delta}}$. Thus, the operator $e^{-ipx_1}A_{p,\delta}e^{ipx_1}=(\nabla +ip \bm{e}_1)^2|_{V_{0,\delta}}$ is analytic of type (A) in $p$ for each fixed $\delta$, in the sense that $\left(e^{-ipx_1}A_{p,\delta}e^{ipx_1}\right)u\in (V_{0,\delta})^*$ is analytic vector-valued function with respect to $p$ for each $u\in V_{0,\delta}$ (see Section 2, Chapter VII of \cite{kato2013perturbation}). As a consequence, the discussions in the same chapter show that $e^{-ipx_1}A_{p,\delta}e^{ipx_1}:V_{0,\delta}\to (V_{0,\delta})^*$ is analytic. Then, by Theorem 1.3 at p.367 of \cite{kato2013perturbation}, we derive that $e^{-ipx_1}(\lambda+A_{p,\delta})^{-1}e^{ipx_1}:(V_{0,\delta})^*\to V_{0,\delta}$ is an analytic operator-valued function with respect to $p$ (see the discussions in Chapter VII, \cite{kato2013perturbation}). On the other hand, we can factorize $\mathbb{T}^{evan}_{\delta}(\lambda;p)$ as
$$
\mathbb{T}^{evan}_{\delta}(\lambda;p)=(\text{Tr}\circ e^{ipx_1})\circ (e^{-ipx_1}(\lambda+A_{p,\delta})^{-1}e^{ipx_1})\circ (e^{-ipx_1}\circ M).
$$
Note that all the operators inside the brackets on the right-hand side above have domains independent of $p$ and are continuous in $p$. We conclude that $\mathbb{T}^{evan}_{\delta}(\lambda;p)$ is continuous with respect to $p\in[0,2\pi]$ for each fixed $\delta$. 

\medskip

Step 8. We prove that $\mathbb{T}^{evan}_{\delta}(\lambda_*+\delta\cdot h;p)$ converges to $\mathbb{T}^{evan}_{0}(\lambda_*;p)$ in operator norm for almost every $p\in [0,2\pi]$. To this end, we apply the resolvent identity of $(\lambda+A_{p,\delta})^{-1}$ to \eqref{eq168_1} to obtain
$$
\begin{aligned}
\mathbb{T}^{evan}_{\delta}(\lambda_*+\delta\cdot h;p)-\mathbb{T}^{evan}_{\delta}(\lambda_*;p)=-\delta\cdot h\left( \text{Tr}\circ(\lambda_*+\delta\cdot h+A_{p,\delta})^{-1}\circ (\lambda_*+A_{p,\delta})^{-1}\circ M\right).
\end{aligned}
$$
Then the uniform boundedness of operators on the right-hand side above yields 
\begin{equation} \label{eq168_2}
    \lim_{\delta\to 0}\|\mathbb{T}^{evan}_{\delta}(\lambda_*+\delta\cdot h;p)-\mathbb{T}^{evan}_{\delta}(\lambda_*;p)\|_{\mathcal{B}(\tilde{H}^{-\frac{1}{2}}(\Gamma),H^{\frac{1}{2}}(\Gamma))}=0,
\end{equation}
for each $p\in[0,2\pi]$. The desired conclusion follows if we can prove that for almost every $p\in[0,2\pi]$,
\begin{equation} \label{eq168_3}
    \lim_{\delta\to 0}\|\mathbb{T}^{evan}_{\delta}(\lambda_*;p)-\mathbb{T}^{evan}_{0}(\lambda_*;p)\|_{\mathcal{B}(\tilde{H}^{-\frac{1}{2}}(\Gamma),H^{\frac{1}{2}}(\Gamma))}=0. 
\end{equation}
The idea for the proof of \eqref{eq168_3} is to express $\mathbb{T}^{evan}_{\delta}(\lambda_*;p)$ as the composition of a sequence of operators, which all converge in operator norm for all $p\in[0,2\pi]$ except a finite set. 
%: for any $ \psi\in \tilde{H}^{-\frac{1}{2}}(\Gamma)$,  we extend the function $w_\delta(x;p):=\mathbb{T}^{evan}_{\delta}(\lambda_*;p)\psi$ to $x\in C_\delta$, and show that its Neumann trace on $\partial D_{n,\delta}$ ($n=1,2$) vary uniformly as $\delta$ changes for $\|\psi\|\leq 1$; then the convergence \eqref{eq168_3} follows if we express $w_\delta$ by the Green formula (similar as \eqref{eq27} except for an extra inhomogeneous term).
In more detail, for any $ \psi\in \tilde{H}^{-\frac{1}{2}}(\Gamma)$,  we extend the function $w_\delta(x;p):=\mathbb{T}^{evan}_{\delta}(\lambda_*;p)\psi$ to $x\in C_\delta$.  Note that $w_{\delta}$ satisfies the following equations
\begin{equation*}
\left\{
\begin{aligned}
&(\Delta_x +\lambda_*)w_{\delta}(x;p)=(f_\delta^0 \psi)(x):=\psi(x_2)\tilde{\delta}(x_1)-\sum_{n=1,2}\langle \psi,\overline{u_{n,\delta}(x;p)}\rangle u_{n,\delta}(x;p)\in H^{-1}(C_\delta) ,\\
&\frac{\partial }{\partial x_2}\Big|_{\Gamma^{\pm}}(w_{\delta}(x;p))=w_{\delta}(x;p)\Big|_{\partial D_{1,\delta}}=w_{\delta}(x;p)\Big|_{\partial D_{2,\delta}}=0, \\
&w_{\delta}(x;p)(x+\bm{e}_1;p)=e^{ip}w_{\delta}(x;p),\quad
\frac{\partial w_{\delta}}{\partial x_1}(x+\bm{e}_1;p)=e^{ip}\frac{\partial w_{\delta}}{\partial x_1}(x;p).
\end{aligned}
\right.
\end{equation*}
where $\tilde{\delta}(\cdot)$ denotes the Dirac delta function. Thus, for $p\notin \mathcal{S}(\lambda_*):=\{\pi\}\cup\{p\in[0,2\pi]:\lambda_*=\lambda_n^e(p)\text{ for some integer $n$}\}$ ($\lambda_n^e(p)$ are introduced in \eqref{eq26_1}), we can express $w_{\delta}(x;p)$ by the Green formula
\begin{equation} \label{eq168_4}
\begin{aligned}
w_{\delta}(x;p)=& 
\int_{C_\delta} G^e(x,y;p,\lambda_*) (f_\delta^0 \psi)(y) dy \\
&+\int_{\partial D} G^e(x,y+z_1-\delta\bm{e}_1;p,\lambda_*) \varphi_{1,\delta}(y) d\sigma(y)+\int_{\partial D} G^e(x,y+z_2+\delta\bm{e}_1;p,\lambda_*) \varphi_{2,\delta}(y) d\sigma(y),
\end{aligned}
\end{equation}
for some $\bm{\varphi}_\delta=(\varphi_{1,\delta},\varphi_{2,\delta})\in{H}^{-\frac{1}{2}}(\partial D)\times{H}^{-\frac{1}{2}}(\partial D)$. The Dirichlet conditions on $\partial D_{1,\delta}\cap\partial D_{2,\delta}$ require that
\begin{equation*}
T_\delta(p,\lambda_*)\bm{\varphi}_\delta=-
\begin{pmatrix}
\left(\int_{\partial D} G^e(x+z_1-\delta\bm{e}_1,y;p,\lambda_*) (f_\delta^0 \psi)(y) dy\right)\Big|_{\partial D} \\
\left(\int_{\partial D} G^e(x+z_2+\delta\bm{e}_1,y;p,\lambda_*) (f_\delta^0 \psi)(y) dy\right)\Big|_{\partial D}
\end{pmatrix},
\end{equation*}
where $T_\delta(p,\lambda_*)$ is introduced in \eqref{eq32}. 

We now show that $T_\delta(p,\lambda_*)$ is invertible for any $p\notin \mathcal{S}(\lambda_*)$ when $\delta$ is sufficiently small. Indeed, since $T(p,\lambda_*)$ is Fredholm with zero index (by Lemma \ref{particle integral operator at the Dirac point}) and $\lim_{\delta\to 0}\|T_\delta(p,\lambda_*)-T(p,\lambda_*)\|=0$ (by Lemma \ref{taylor approximation of T_delta}), $T_\delta(p,\lambda_*)$ is also Fredholm with zero index. On the other hand, for each $p\neq \pi$, $\lambda_*$ is not a characteristic values of $T_\delta(\lambda;p):=T_\delta(p,\lambda)$ (see Corollary \ref{common band gap}). Therefore $T_\delta(p,\lambda_*)$ is invertible for $p\notin \mathcal{S}(\lambda_*)$. 
As a result, we have the following expression
\begin{equation} \label{eq168_5}
\bm{\varphi}_\delta=-T^{-1}_\delta(p,\lambda_*)
\begin{pmatrix}
\left(\int_{\partial D} G^e(x+z_1-\delta\bm{e}_1,y;p,\lambda_*) (f_\delta^0 \psi)(y) dy\right)\Big|_{\partial D} \\
\left(\int_{\partial D} G^e(x+z_2+\delta\bm{e}_1,y;p,\lambda_*) (f_\delta^0 \psi)(y) dy\right)\Big|_{\partial D}
\end{pmatrix}.
\end{equation}
Moreover $\lim_{\delta\to 0}\|T^{-1}_\delta(p,\lambda_*)-T^{-1}(p,\lambda_*)\|=0$. 
By substituting \eqref{eq168_5} into \eqref{eq168_4} and then taking trace to $\Gamma$, we can express $\mathbb{T}^{evan}_{\delta}(\lambda_*;p)$ as the composition of a sequence of operators, which all converge in operator norm for each $p\notin \mathcal{S}(\lambda_*)$. Therefore \eqref{eq168_3} holds almost everywhere for $p\in[0,2\pi]$.

%From \eqref{eq168_3} and \eqref{eq168_2}, \eqref{eq168} is proved as promised.

\medskip

Step 9. Finally, in view of \eqref{eq166}, \eqref{eq168} and the decomposition \eqref{eq161}, we have 
\begin{equation*}
    \lim_{\delta\to 0}\|\mathbb{T}_\delta(\lambda_*+\delta\cdot h)-\mathbb{T}_0\|_{\mathcal{B}(\tilde{H}^{-\frac{1}{2}}(\Gamma),H^{\frac{1}{2}}(\Gamma))}=0.
\end{equation*}
Following the same line of proof, it can be shown that 
\begin{equation*}
    \lim_{\delta\to 0}\|\mathbb{T}_{-\delta}(\lambda_*+\delta\cdot h)-\mathbb{T}_0\|_{\mathcal{B}(\tilde{H}^{-\frac{1}{2}}(\Gamma),H^{\frac{1}{2}}(\Gamma))}=0.
\end{equation*}
Combining the above with the results of Step 1 and the estimates of (\ref{eq92}) and (\ref{eq92-1}),  we conclude the proof.
\end{proof}

\bibliographystyle{plain}
\bibliography{ref}

\end{document}